\documentclass[reqno,11pt]{amsart}
\usepackage{amsmath, latexsym, amsfonts, amssymb, amsthm, amscd}
\usepackage[foot]{amsaddr}
\usepackage{multirow,stmaryrd}
\usepackage{framed}
\usepackage{graphics,epsf,psfrag}
\usepackage{graphicx}
\usepackage{mathabx}

\usepackage{xcolor}

\usepackage{hyperref}

\newcommand{\comment}[1]{}

\numberwithin{equation}{section}

\newtheorem{theorem}{Theorem}[section]
\newtheorem{lemma}[theorem]{Lemma}
\newtheorem{proposition}[theorem]{Proposition}
\newtheorem{cor}[theorem]{Corollary}
\newtheorem{rem}[theorem]{Remark}

\DeclareMathOperator{\sign}{\mathrm{sign}}

\newcommand{\ind}{\mathbf{1}}

\newcommand{\R}{\mathbb{R}}

\renewcommand{\tilde}{\widetilde}

\newcommand{\cN}{{\ensuremath{\mathcal N}} }
\newcommand{\cL}{{\ensuremath{\mathcal L}} }

\DeclareMathSymbol{\leqslant}{\mathalpha}{AMSa}{"36} 
\DeclareMathSymbol{\geqslant}{\mathalpha}{AMSa}{"3E} 
\DeclareMathSymbol{\eset}{\mathalpha}{AMSb}{"3F}     
\renewcommand{\leq}{\;\leqslant\;}                   
\renewcommand{\geq}{\;\geqslant\;}                   
\newcommand{\dd}{\,\text{\rm d}}             

\newcommand{\bbC}{{\ensuremath{\mathbb C}} }

\newcommand{\bbE}{{\ensuremath{\mathbb E}} }

\newcommand{\bbN}{{\ensuremath{\mathbb N}} }

\newcommand{\bbP}{{\ensuremath{\mathbb P}} }

\newcommand{\bbR}{{\ensuremath{\mathbb R}} }

\newcommand{\bbZ}{{\ensuremath{\mathbb Z}} }

\newcommand{\ga}{\alpha}
\newcommand{\gb}{\beta}
\newcommand{\gd}{\delta}
\newcommand{\gep}{\varepsilon}       

\newcommand{\gvr}{\varrho}

\newcommand{\gD}{\Delta}

\newcommand{\gl}{\lambda}

\newcommand{\gs}{\sigma}

\makeatletter
\def\captionfont@{\footnotesize}
\def\captionheadfont@{\scshape}

\long\def\@makecaption#1#2{%
  \vspace{2mm}
  \setbox\@tempboxa\vbox{\color@setgroup
    \advance\hsize-6pc\noindent
    \captionfont@\captionheadfont@#1\@xp\@ifnotempty\@xp
        {\@cdr#2\@nil}{.\captionfont@\upshape\enspace#2}%
    \unskip\kern-6pc\par
    \global\setbox\@ne\lastbox\color@endgroup}%
  \ifhbox\@ne 
    \setbox\@ne\hbox{\unhbox\@ne\unskip\unskip\unpenalty\unkern}%
  \fi
  \ifdim\wd\@tempboxa=\z@ 
    \setbox\@ne\hbox to\columnwidth{\hss\kern-6pc\box\@ne\hss}%
  \else 
    \setbox\@ne\vbox{\unvbox\@tempboxa\parskip\z@skip
        \noindent\unhbox\@ne\advance\hsize-6pc\par}%
\fi
  \ifnum\@tempcnta<64 
    \addvspace\abovecaptionskip
    \moveright 3pc\box\@ne
  \else 
    \moveright 3pc\box\@ne
    \nobreak
    \vskip\belowcaptionskip
  \fi
\relax
}
\makeatother
\def\writefig#1 #2 #3 {\rlap{\kern #1 truecm
\raise #2 truecm \hbox{#3}}}

\newcommand{\tf}{\textsc{f}}

\newcommand{\floor}[1]{\left\lfloor #1 \right\rfloor}
\newcommand{\DLMF}[2]{\href{http://dlmf.nist.gov/#1.E#2}{\cite[#1.#2]{DLMF}}}
\newcommand{\DLMFs}[1]{\href{http://dlmf.nist.gov/#1}{\cite[\S#1]{DLMF}}}
\newcommand{\const}{\textup{const.}  }
\newcommand{\ux}{\underline{x}}
\newcommand{\NMW}{\mathtt{N}}
\newcommand{\Neps}{\epsilon}

\begin{document}


\title[Continuum limit of statistical mechanics random matrix products]{Continuum limit of random matrix products \\ in statistical mechanics of disordered systems}

\author[F. Comets, G. Giacomin and R. L. Greenblatt]{Francis Comets, Giambattista Giacomin and Rafael L. Greenblatt}
\address[FC,GG]{Universit\'e Paris Diderot, Sorbonne Paris Cit\'e,  Laboratoire de Probabilit{\'e}s, Statistiques  et Mod\'elisation, UMR 8001,
            F-75205 Paris, France}
\address[RLG]{Universit\`a Roma 3, Dipartimento di Matematica e Fisica, 
Largo San Leonardo Murialdo 1,
00146, Roma, Italy}

\begin{abstract}
{We consider a particular weak disorder limit (\emph{continuum limit}) of matrix products that arise in the analysis of  disordered statistical mechanics systems, with a particular focus on random transfer matrices.
The limit  system is a diffusion model for which the leading Lyapunov exponent can be expressed explicitly in terms of modified Bessel functions, a formula that appears in the physical literature on these disordered systems.  We provide an analysis of the diffusion system as well as of the link with the matrix products. 
We then apply the results to the framework considered by  Derrida and Hilhorst in \cite{DH}, which 
deals in particular with the strong interaction limit for disordered Ising model in one dimension and that identifies 
a singular behavior of the Lyapunov exponent (of the transfer matrix),
and to the two dimensional Ising model with columnar disorder (McCoy-Wu model). We show 
that the continuum limit sharply captures the  Derrida and Hilhorst singularity. 
Moreover we revisit the analysis by McCoy and Wu \cite{MW1} and remark that it can be interpreted in terms of  the continuum limit approximation. We provide a mathematical  analysis of the continuum approximation of 
the free energy of the McCoy-Wu model, clarifying the prediction (by McCoy and Wu)
that, in  this approximation, the free energy of the two dimensional Ising model with columnar disorder is $C^\infty$ but not analytic at the critical temperature. }
\bigskip

\noindent  \emph{AMS  subject classification (2010 MSC)}:
82B44, 
60K37,  
82B27, 
60K35 

\end{abstract}

\maketitle

\section{Introduction}
Products of random matrices can often be interpreted, in a statistical mechanics perspective, as models of disordered systems. 
The leading Lyapunov 
exponent  may then be identified with some physical quantity such as the free energy density  or  persistence length.
We can also take the opposite viewpoint and ask whether a disordered system can be written in terms of, or at least approximated by, a suitable product of random matrices.
 It turns out that there are several examples in which this can be done. Examples include essentially all statistical mechanics systems in which there is a natural one dimensional structure, but it goes also beyond this: the literature is too wide to be properly cited here and we refer to the reviews \cite{cf:contemp,CPV}.
 Of particular interest for us are the examples arising from the transfer matrix approach in the statistical mechanics of disordered systems.
 For one  dimensional (let us say, Ising or Potts) models  with finite range interaction   one can write the partition function in terms 
 of a product of matrices \cite{cf:Baxter}: 
  if the interactions are only one body and nearest neighbor two body  the transfer matrix of an Ising model is  a two by two matrix, and the size  is larger for Potts and/or longer range models.  But for  two or more dimensional systems
 the size of the  transfer matrix   tends to infinity in the thermodynamic limit and the transfer matrix should be thought more as a transfer operator: this is true also in one dimension if the spin variable can take an infinite number of values \cite[Ch.~5]{cf:Ruelle-book}.
 Nevertheless, also in these cases finite dimensional matrix models can  be helpful (for numerical approximations for example, but also for rigorous bounds, see for example \cite[Ch.~9]{cf:GB} and references therein). It is however  remarkable that also the solution of the two dimensional Ising model with nearest neighbor interactions and no external field can ultimately be  expressed in terms of products of  two by two  matrices: this is the essence of {several formulations of} the celebrated solution of Lars Onsager \cite{cf:Baxter,MWbook}. What is even more remarkable from our viewpoint is that this structure still holds when  special types of disorder are introduced, giving a product of \emph{random} matrices \cite{MWbook,MW1,ShankarMurthy}. 
  
  The two by two matrices that arise in the problems we have just mentioned have a particular form: it is
\begin{equation}
\label{eq:matrix}
\begin{pmatrix}
1& \gep\\ 
\gep Z  & Z
\end{pmatrix}\, , 
\end{equation}
 where $\gep$ is a real number -- say $\vert \gep \vert \le 1/2$ to keep far from the zero determinant case $\gep= \pm 1$ --  and $Z$ is a positive random variable with $\bbE \log_+(Z) < \infty$.  Let us call (informally for the moment)
 $\widehat \cL_Z(\gep)$ the Lyapunov exponent of a product of IID matrices of the form \eqref{eq:matrix}, which appear notably in the following two contexts.

  \smallskip
  
 \begin{itemize}
 \item In the one dimensional Ising model with random  external field $h=h_j$ -- that is, $\{h_j\}_{j=1,2, \ldots}$ is a sequence of independent identically distributed (IID) random variables -- and nearest neighbor interaction $J$, 
 the transfer matrix can be cast in the form \eqref{eq:matrix}, with $Z=\exp(-2h)$ and $\gep=\exp(-2J)$. The free energy density 
 is therefore precisely $\widehat \cL_Z(\gep)$ 
  and the $\gep\searrow 0$ limit is the limit of strong ferromagnetic 
 interaction. 
 \item In a much less straightforward way (detailed in Appendix~\ref{sec:MW}), also the free energy of the two dimensional Ising model with a special type   of disordered nearest neighbor interactions (\emph{columnar disorder}), and no external field, is (essentially) just  $\int_0^{1/2}\widehat \cL_Z(\gep)\dd \gep$, of course with a proper choice of $Z=Z_\gb$ that contains  the inverse temperature $\gb$ of the system.
 The phase diagram of this model (that is, the presence and nature of phase transitions) is determined by the regularity of this expression as a function of $\gb$; 
 the most notable prediction for this model, which goes now under the name of McCoy-Wu model, is that the second order transition
 of the two dimensional non-disordered Ising model (for which the second derivative of the free energy diverges
 at criticality like $-\log\vert \gb -\gb_c\vert$) becomes of infinite order when the columnar disorder is introduced:  that is, the free energy at the critical point is $C^\infty$ but not analytic. The precise nature of the singularity is characterized in \cite{MW1} by means of a divergent power series for the free energy at $\gb_c$, where the value of $\gb_c$ depends on the disorder: a summary of the expected effect of disorder on the transition for the two dimensional Ising model is in \cite[\S~5.3]{cf:G}. The McCoy-Wu model has a prominent role in physics because it can be mapped  to the one dimensional quantum spin chain with transversal magnetic field  \cite{cf:dfisher}
 and because it
  has played a central role in the development of the \emph{real space strong/infinite disorder renormalization group}
(see e.g.  \cite{cf:dfisher} and \cite[\S~5.3]{cf:G}).
 \end{itemize}
 
 \smallskip
 
 Other contexts in which  \eqref{eq:matrix} and  $\widehat \cL_Z(\gep)$ arise include one dimensional random walk in random environment and  a number of random hopping problems (see \cite{CPV} and references therein), and the key issue for us is that 
 all this vast literature focuses on the $\gep\to 0$ behavior of $\widehat \cL_Z(\gep)$, see notably \cite{CPV,DH,MWbook,MW1,NL}.
 From a mathematical viewpoint this limit is of interest because, thanks to \cite{Ruelle} (see also \cite{cf:dubois}), we know that 
 $\gep \mapsto \widehat \cL_Z(\gep)$ is real analytic if $\vert \gep\vert \in (0,1)$ under additional mild hypotheses on  $Z$ (for example: $\bbP(Z=c)=0$ for every $c$). But the regularity at $\gep=0$ is not obvious, as well as  if there is a singularity at all. And this is precisely the question addressed in \cite{CPV,DH,MWbook,MW1,NL}. In particular $\widehat \cL_Z(\gep)$
  is expected to have a fractional or logarithmic  scaling when $\gep \to 0$ under the  \emph{frustration hypothesis} that $\bbP (Z>1)$ and $\bbP(Z<1)$ are both positive: this is the case for example of the \emph{$\gep ^{2 \ga}$ singularity}  found in \cite{DH}, and proven mathematically in \cite{cf:GGG},  and that we will explain in detail in Section~\ref{sec:DH83}.  
   
  \medskip
 
 Here we do not address the study of the Lyapunov exponent of products of matrices of the form  \eqref{eq:matrix}. Rather we 
 focus on a continuous time model that arises as a diffusive limit of the matrix product (we call it \emph{continuum limit}). Roughly, the limit is achieved 
 by considering matrices close to the identity: $\gep$ is replaced by $\gep \gD$, with $\gD\searrow 0$ and we consider $Z=Z^\gD$ that is very concentrated around one: both $\bbE[Z^\gD]-1$ and 
 var$(Z^\gD)$ are of order  $\gD$. The dynamics will therefore happen on a timescale $1/\gD$ and it will be governed by a two dimensional  stochastic differential system. We then  study the leading Lyapunov exponent $\cL(\gep)$ of this limit system: we will actually show that  $\widehat \cL_Z(\gep) \sim \gD \cL(\gep)$ for ${\gD\searrow 0}$ (we use $\sim$ for asymptotic equivalence: the ratio of left-hand and right-hand sides converges to one). This limit has been already considered in several works  and even in greater generality (\emph{matrices close to the identity}: see e.g.\cite{cf:FrLl,MW1,cf:Zanon,cf:CLTT}), {but mathematically rigorous results are lacking} (with the exception of \cite{cf:Sadel}, whose assumptions however exclude the case we treat): 
 the type of results one finds are expansions of the type
 \begin{equation}
 \label{eq:DZexp}
 \widehat \cL_{Z^\gD}(\gep) \, =\, c_1(\gep) \gD + c_2(\gep) \gD^2 +\ldots \, ,
 \end{equation}
 where of course $c_1(\gep)=\cL(\gep)$ and expressions or at least procedures to compute the $c_j(\gep)$ are given. {To be precise, a full expansion like \eqref{eq:DZexp} is not expected to hold in general and, even in the cases in which it holds, e.g. \cite{cf:Sadel}, and assuming smooth dependence in $\gD$ of the coefficients of the matrix,  there is to our knowledge no proof that  }
 $ \gD \mapsto \widehat \cL_{Z^\gD}(\gep)\in C^\infty$. {We point out, however,  a  very special example} in \cite{cf:Zanon} that has been worked out explicitly 
and for which the Lyapunov exponent is  analytic except at zero 
where it is nonetheless  $C^\infty$ (note also that \cite{Ruelle} cannot be applied because for $\gD=0$
the matrix is the identity matrix).   
 
 {We will focus only on $\cL(\gep)$: in other words, the continuum limit we consider captures only   the leading order term in \eqref{eq:DZexp}. 
 The first remarkable fact is that $\cL(\gep)$ has an explicit expression in terms of a ratio of modified Bessel functions: we provide a proof of this fact, which has long been known in the physics literature. To our knowledge, it is found for the first time in  \cite[(4.31)]{MW1}, and it then reappears in other works and contexts, see for example \cite{cf:CLTT} to which we refer also for a comprehensive review of the literature.} It is rather surprising that, while 
 {a detailed} analysis of the $\gep \to 0$ limit  of $\cL(\gep)$ is rather 
 straightforward {(the case of $\ga\in [0, 2)$ is worked out in  \cite[first formula {on} p.~248]{cf:grabsch}), a full analysis appears to be lacking, as well as {an} emphasis on the rather striking fact that  
 the $\gep \to 0$ behavior   of $\cL(\gep)$  \emph{captures all known and conjectured features} of the $\gep \to 0$ behavior   of $ \widehat \cL_{Z^\gD} (\gep)$, i.e. for matrix products (without assuming the  disorder to be small).
  In particular, the $\gep ^{2 \ga}$ singularity  found in \cite{DH}  is fully present in the continuum limit expression: mathematical results on this issue for $ \widehat \cL_{Z^\gD} (\gep)$  have been  recently obtained  \cite{cf:GGG,cf:benjamin}, but the control of the singular term is an open problem for $\vert \ga\vert \ge 1$. }

Turning to  the  McCoy-Wu model, we come back to the fact that this model appears prominently in the physical literature, in part of course because of its exactly solvable character. 
And  {the conventional wisdom in the mathematical community appears to be} that 
the  McCoy-Wu claims are exact. And this is correct as far as the free energy formula (in terms of the Lyapunov exponents) is concerned. 
The subsequent analysis is less sound: $\gb_c$ is identified via the equation $\bbE \log Z=0$ -- the random variable $Z$ depends on the inverse temperature $\gb$ -- and this assertion has some grounds at least at a heuristic level, but then one has to show that the free energy 
$\int_0^{1/2}\widehat \cL_{Z^\gD}(\gep) \dd \gep$ is not analytic at $\gb=\gb_c$.
And this is (ultimately) done by replacing $\widehat \cL_{Z^\gD}(\gep) $ with $\gD \cL(\gep)$ and this step 
is very weak on mathematical grounds because the McCoy-Wu claim (which provides the motivation for the whole exercise) is for $\gD>0$ (possibly very small, but non zero): making this step 
rigorous -- possibly by controlling the remainder of the series in \eqref{eq:DZexp} -- appears to be very challenging, and we do not address this in the present work. 
Once this approximation is done, McCoy and Wu are left with  studying  the regularity in $\gb$ of $\int_0^{1/2} \cL(\gep) \dd \gep$.
In spite of being a relatively explicit expression, this is still challenging.  
McCoy and Wu  do this by developing the ratio of Bessel functions in the expression for $\cL(\gep)$ for $\gb$ close to $\gb_c$
--  the dependence in $\gb$ is in the index
of the Bessel functions  --  and by identifying the leading-order (in magnitude) terms in an expansion of the Bessel functions as the most singular part. 
We provide a proof that $\int_0^{1/2} \cL(\gep) \dd \gep$ is $C^\infty$ but not analytic at $\gb=\gb_c$
and that the asymptotic series at $\gb_c$ is qualitatively the one found by McCoy and Wu (up to a multiplicative factor that they lost
when singling out the most singular term; a similar {correction was noted by Luck \cite{cf:Luck} in a related model}): technically, this is the most demanding part of our contribution. 
{We stress however that what we prove does not yield results  on the  transition for the McCoy-Wu model.} The challenging gap pointed out just above remains unclosed.
 But we believe that our contribution helps  understanding the true
 content of the remarkable McCoy-Wu analysis. As a side remark: the computation of McCoy and Wu is done 
 for a very special form of the disorder distribution while they affirm that they expect the result to be true in great generality. 
 We work under very general assumptions on the disorder, thus substantiating this   claim. 

\medskip

We begin by presenting the diffusion system  and its analysis. This is a 
 stochastic dynamical system that is interesting in its own right and we provide 
 a detailed analysis that goes beyond the strict purpose of what has been explained up to now. 
 In particular we  prove a Central Limit Theorem on the fluctuations
of the Lyapunov exponent  for  the diffusive  limit system, with an explicit formula for the variance and an analysis of the singular behavior.
We then provide a proof that the Markov chain associated to the matrix product described above 
does scale to the diffusion system and that to leading order (in $\gD$)  the Lyapunov exponent of the Markov chain 
is asymptotically proportional to the Lyapunov exponent of the diffusion system, that is
$\widehat \cL_Z(\gep) \sim \gD \cL(\gep)$. 
The rest of our work focuses on the regularity/singularity properties of $\cL(\gep)$ and of the expressions related to it 
that are of physical relevance.

\subsection{The diffusion model and its leading Lyapunov exponent}
We consider the  solution to the stochastic (It{\^o}) differential equations
\begin{equation}
\label{eq:sys}
\begin{cases}
\dd X_1(t)\, =\, \gep X_2(t) \dd t\, ,
\\
\dd X_2(t)\, =\,\left( \gep X_1(t)+ \frac{(1-\ga)\gs ^2}2  X_2(t) \right) \dd t + \gs X_2(t) \dd B_t\, ,
\end{cases}
\end{equation}
where $B_\cdot$ is a standard Brownian motion, $\gep\neq 0$, $\ga\in \bbR$  and $\gs>0$. We consider
deterministic   initial condition $(X_1(0), X_2(0))\in \bbR^2\setminus \{(0,0)\}$. The case $(X_1(0), X_2(0))=(0,0)$, as well as $\gep=0$, are excluded because they are atypical and trivial. The system 
\eqref{eq:sys} is linear  with a multiplicative noise so, given the initial condition, there exists a unique strong solution. Our focus is on the Lyapunov exponent $\cL (\gep)=\cL_{\gs, \ga} (\gep)$ that we introduce via our first statement in which we use the Euclidean norm $\Vert\cdot \Vert$ in $\bbR^2$ just for definiteness. Before stating it we need to recall
 one of the  definitions of the modified Bessel function of $2^{\text{nd}}$ kind of index $\ga\in \bbC$ and argument $x>0$
\DLMFs{10.25}
\begin{equation}
\label{eq:Bessel-def}
K_\ga (x)\, :=\, \int_0^\infty \exp\left(-x \cosh(t)\right) \cosh(\ga t) \, \dd t \, =\, \frac 12 \int_0^\infty  \frac{1}{y^{1+ \ga}} 
\exp\left(- \frac{x}{2}  \left(y + \frac 1y \right)\right)
\dd y\, .
\end{equation}
 We note from now that
$K_\ga (x)=K_{-\ga}(x)$.

\medskip

\begin{theorem}
\label{th:Lyap}
For every $\gep\neq 0$ and every 
$(X_1(0), X_2(0))\in \bbR\setminus \{(0,0)\}$ the limit 
\begin{equation}
\label{eq:Lyap}
\lim_{t \to \infty} \frac 1t \bbE \log \left \Vert (X_1(t), X_2(t))\right\Vert \, =: \, \cL_{\gs, \ga} (\gep)\, ,
\end{equation}
exists and  does not depend on $(X_1(0), X_2(0))$.  Moreover
\begin{enumerate}
\item the limit is unchanged if we replace $\Vert (X_1(t), X_2(t))\Vert$ with $\vert X_j \vert$, $j=1,2$ as well as if 
we remove the expectation (in this case the convergence is almost sure);
\item if $\gep>0$ [resp. $\gep<0$], then $\sign(X_1(t))= \sign(X_2(t))$ [resp. $\sign(X_1(t))\neq \sign(X_2(t))$)]
for all $t\ge \tau:= \inf\{t\ge 0: \, \sign(X_1(t))= \sign(X_2(t))\}$ [resp. $\tau:= \inf\{t\ge 0: \, \sign(X_1(t))\neq \sign(X_2(t))\}$]
and $\bbE [\tau ]< \infty$. Moreover 
$\cL_{\gs, \ga} (\gep)=\cL_{\gs, \ga} (-\gep)$.
\item For $\gep>0$ and every $\ga\in \bbR$ we have 
\begin{equation}
\label{eq:formula}
\cL_{\gs, \ga}(\gep)\, =\, \frac{\gs^2 }{4}\left(\frac{x K_{\ga-1}(x)}{ K_{\ga}(x)}\right)\,, \qquad {\rm with}\quad x:=\frac{4 \gep}{\gs^2}\,.
\end{equation}
\end{enumerate}
\end{theorem}

\medskip

\medskip

We draw the attention of the reader on the identification of $x$ with $4\gep/ \gs ^2$. This shortcut notation is kept in all statements and proofs.

\medskip

Some of the results in Theorem~\ref{th:Lyap} can be understood on the basis of a symmetry enjoyed 
by our system: If $(X_1(\cdot),X_2(\cdot))$ solves \eqref{eq:sys}, then 
$(-X_1(\cdot),X_2(\cdot))$ solves \eqref{eq:sys} with $\gep$ replaced by $-\gep$. 
We can therefore restrict our analysis to the case $\gep>0$ and we will show that 
only the (interior of the) quadrants in which both coordinates have the same sign -- first and third quadrant -- are recurrent for the dynamics:
all the rest is transient. By linearity we can then restrict to the first quadrant.  

\medskip

{As already mentioned in the introduction, most of the content of Theorem~\ref{th:Lyap} is known in the physical literature
and \eqref{eq:formula} appears in number of contexts. Besides the 
 pioneering work \cite{MW1} that we have already mentioned,    \eqref{eq:formula}  appears for example also in \cite[Sec.~3]{cf:Luck}, which deals with disordered quantum Ising chains with transverse magnetic field: this is not surprising, because 
this quantum model is mapped exactly (by a Suzuki-Trotter path integral) into a suitable limit of
the McCoy-Wu model  
(this is also  exploited   \cite{cf:dfisher}). It also appears in the analysis of one dimensional random Schr\"odinger equation and in   the analysis of a diffusion of a particle in a random force field, see e.g.\ \cite{cf:bouchaud,cf:CTT}, and \cite{cf:FrLl} which is possibly the first work addressing precisely what we refer to as the continuum limit. Mathematical works have also been done in this context 
and  \eqref{eq:formula} appears in  a study of the quenched large deviations of diffusions in a random environment \cite[Prop. 2.1]{cf:marina}:
By Kotani's formula (unpublished, 1988)
the Laplace transform of hitting times for the diffusion can be expressed in terms of a Riccati equation -- the $K_\ga(\cdot)$ Bessel function appears as solution of this equation -- that is equivalent  to \eqref{eq:EDS} below. 

Finally, the first  item  of Theorem~\ref{th:Lyap} is a  classical result at the random matrix level, and the second item is an elementary observation. In the continuum set-up, the first two items follow by applying standard tools of stochastic analysis (the proofs turn out to be rather concise and we give full details). The third item is a computation: it is not novel, but it is very short and we provide it for completeness.  }

\medskip

{We also have a rather explicit representation for the fluctuations}:
\medskip

\begin{proposition} 
\label{prop:clt}      
The family of random variables
\begin{equation}
\label{eq:cltLyap}
\left\{
 \frac {1}{\sqrt t} \big(  \log \left \Vert (X_1(t), X_2(t))\right\Vert - t \cL_{\gs, \ga} (\gep) \big) \right\}_{t\in [0, \infty)}
\end{equation}
converges in law for $t \to \infty$
to a centered Gaussian variable with variance $v_{\gs, \ga} (\gep)   \in (0, \infty)$,
\begin{equation}
\label{eq:varLyap}
 v_{\gs, \ga} (\gep) =
 \frac 2{\gs^2 K_\ga(x)} \int_0^\infty \frac{1}{y^{1-\ga}}e^{\frac{x}{2}\left(y+\frac{1}{y}\right)} \left( \int_0^y \frac{\gep z-\cL_{\gs, \ga}(\gep)}{z^{1+\ga}} e^{-\frac{x}{2}\left(z+\frac{1}{z}\right)} 
  \dd z \right)^2 \dd y
\,.
\end{equation}
\end{proposition}

\medskip

{In physics   the  behavior of fluctuations for matrix products has beed repeatedly addressed, see \cite{cf:ramola,cf:schomerus} and references therein; the same is true for the mathematical literature \cite{cf:BL}. Proposition~\ref{prop:clt} is about the fluctuations for the continuum limit: even taking into account the results  in  \cite{cf:ramola,cf:schomerus}, we think that
understanding how random matrix products fluctuations and continuum limit fluctuations are related is an open issue.
 }


\subsection{Small $\gep$ asymptotic expansions}

Thanks to Theorem~\ref{th:Lyap} item~(3), studying the small $\gep$ behavior of $\cL_{\gs, \ga}(\gep)$ is 
 just a book-keeping exercise that exploits  
the asymptotic behavior of $K_\ga(\cdot)$. 
For  $\ga\in [0, 2)$ this result can be found in   \cite[first formula, p.~248]{cf:grabsch}, see also
\cite[(3.45)]{cf:Luck} for $\ga=0$: 
we provide the general result and 
we  will explain the relevance of this exercise in Section~\ref{sec:DH83}. 
Throughout this work $\Gamma$ denotes the  Gamma function, see  \DLMFs{5.2} for definitions
and properties.

\medskip

\begin{proposition}
\label{th:asympt}
Recall that $x=4 \gep/ \gs ^2$.
For $\ga \in (0,\infty) \setminus \bbZ$ we have for $\gep\searrow 0$
\begin{equation}
\label{eq:asympt1}
\frac 4 {\gs^2} \cL_{\gs, \ga}(\gep) \, =\, 
c_1 (\ga)x^2+ \ldots  + c_{\lfloor \ga \rfloor} (\ga)x^{2\lfloor \ga \rfloor} +2
\frac{\Gamma(1-\ga)}{\Gamma(\ga)} \left( \frac x 2\right)^{2\ga} +O\left(x^{\min(2\lceil \ga \rceil, 4 \ga)}\right) , 
\end{equation}
where $c_j(\cdot)$ is a rational function (for explicit expressions, see \eqref{eq:firstsing}-\eqref{eq:fewterms}).

For $\ga \in  \{1,2, \ldots\}$ we have
\begin{equation}
\label{eq:asympt2}
\frac 4 {\gs^2} \cL_{\gs, \ga}(\gep) \, =\, 
c_1(\ga) x^2+ \ldots + c_{\ga -1}(\ga)x^{2(\ga -1)}+ (-1)^{ \ga }\frac{2^{2-2\ga} }{(( \ga -1)!)^2}
x^{2\ga } \log x+
O\left( x^{2\ga }\right)\, ,
\end{equation}
where  $c_j( \cdot )$ is the same rational function as in the non integer case. 

For $\ga=0$ we have 
\begin{equation}
\cL_{\gs, \ga}(\gep)\, = \, \frac  {\gs^2} {4 \log (1/x)} +O \left(\left(\log 1/x \right)^{-2}\right)\, ,
\end{equation}
and
the result for $\ga<0$ is directly recovered from \eqref{eq:asympt1}-\eqref{eq:asympt2} by using the identity
\begin{equation}
\label{eq:recover}
\frac 4 {\gs^2} \cL_{\gs, \ga}(\gep)\stackrel{\ga <0}= 2 \vert \ga\vert+ \frac 4 {\gs^2} \cL_{\gs, \vert \ga\vert}(\gep)\, .
\end{equation}
\end{proposition}

\medskip

The identity \eqref{eq:recover}
is a simple consequence of the Bessel identity
\begin{equation}
\label{eq:Bessel-id}
x K_{1+\ga}(x)\,=\, 2 \ga K_\ga(x)+ xK_{ -1+\ga}(x)\, ,
\end{equation}
that follows from \eqref{eq:Bessel-def} by integration by parts, together with the identity $K_\alpha(x) = K_{-\alpha}(x)$.

\medskip

We have chosen to give these expansions up to the leading singular term: one can of course be much more precise. 
Keeping only the leading term, Proposition~\ref{th:asympt} implies
\begin{equation}
\label{eq:specialcase}
\cL_{\sigma,\alpha}(\gep) \stackrel{\gep \searrow 0}\sim  \left(\frac{\gs^2} 4\right)
\begin{cases}
c_1(\ga) x^2 & \text{ if } \ga >1\, ,
\\
\frac{2\Gamma(1-\ga)}{\Gamma(\ga)}  \left(\frac x 2 \right)^{2\ga}& \text{ if } \ga \in(0,1)\, ,
\\
1/\log (1/x)   &\text{ if } \ga= 0\, , \\
2 \vert \ga \vert
&\text{ if } \ga\in (-\infty,0)\,;
\end{cases}
\end{equation}
recall again that $x = 4\varepsilon/\sigma^2$.
It is certainly worth observing for the benefit of those readers who are less at home with special function calculations that  
\eqref{eq:specialcase} can be derived by elementary asymptotic methods from 
\eqref{eq:Bessel-def}. This is of course also the case for the full Proposition~\ref{th:asympt}, but the exercise 
becomes particularly involved and resorting to the special functions literature is certainly wise, or  even necessary.


\medskip

\begin{rem}
	\eqref{eq:specialcase} directly entails $\lim_{\gep\searrow 0} \cL_{\gs, \ga}(\gep)=\frac{\gs^2}2\vert \ga \vert \ind_{\ga <0}$,  and $\ga \mapsto \frac{\gs^2}2\vert \ga \vert \ind_{\ga <0}$ is singular at the origin, while 
$\ga \mapsto \cL_{\gs, \ga}(\gep)$ is real analytic for $\gep \neq 0$ (and it is meromorphic in the whole $\bbC$: see beginning of Section~\ref{sec:MW_proof}).  It is possibly worth
observing that we have not defined $\cL_{\gs, \ga}(\gep)$ for $\gep=0$ because of the pathological nature of this case, but 
the limit in \eqref{eq:Lyap} exists also for $\gep=0$ and  $\cL_{\gs, \ga}(0)=
\frac{\gs^2}2\vert \ga \vert \ind_{\ga <0}$ (in agreement with $\lim_{\gep\searrow 0} \cL_{\gs, \ga}(\gep)$), but only if  $X(0) \neq 0$; otherwise the Lyapunov exponent is $-\infty$. 
Moreover
the (Laplace) asymptotic behavior of the two components for $\gep=0$ in general does not coincide with the 
Lyapunov exponent. 
\end{rem}

\medskip

Of course one could wonder about the behavior {as $\gep \searrow 0$ of the variance $v_{\gs, \ga}(\gep)$ in}
 Proposition~\ref{prop:clt}.  

\medskip

\begin{proposition}
\label{prop:asympt}
{Still} with the notation $x=4 \gep/ \gs ^2$, we have that for every $\ga$ there exists $C(\ga)>0$ 
(see \eqref{eq:C(alpha)} for an explicit expression) such that
\begin{equation}
\label{eq:propasympt}
v_{\gs, \ga}(\gep) \stackrel{\gep \searrow 0}\sim C(\ga) \frac{\gs^2}2 \times
 \begin{cases}
 1 & \text{ if } \ga \le 0\, ,\\
 x^{2\ga } & \text{ if } \ga \in (0,2)\, ,\\
  x^4\log (1/x) & \text{ if } \ga=2\, , \\
 x^{4} & \text{ if } \ga>2\, .
\end{cases}
\end{equation}
\end{proposition}

 \subsection{From matrix product to the diffusion model}
 We are now going provide rigorous results about  how  \eqref{eq:sys} emerges  as limit of matrix products. 
Consider, for  $\gD>0$ and given  an IID sequence $\{\cN_n\}_{n=1,2, \ldots}$  of standard Gaussian variables,  the discrete time stochastic 
process $\{(X_1^\gD(n),X_2^\gD(n))\}_{n=0,1, \ldots}$ defined recursively from the deterministic initial condition $(X_1^\gD(0), X_2^\gD(0))=
(X_1(0), X_2(0))$ by
\begin{equation}
\label{eq:Deltamodel}
\begin{cases}
X_1^\gD(n+1) \,=\, X_1^\gD(n) + \gep X_2^\gD(n) \gD\, ,
\\
X_2^\gD(n+1) \,=\,
e^{\gs \sqrt{\gD}\, \cN_{n+1} - \ga \frac{\gs ^2}2 \gD} 
\left( X_2^\gD (n) + \gep X_1^\gD(n) \gD  \right)\, .
\end{cases}
\end{equation}
 Defining 
\begin{equation}
\label{eq:Zdelta}
Z^\gD(n+1)= e^{\gs \sqrt{\gD} \, \cN_{n+1} - \ga\frac{\gs ^2}2 \gD},
\end{equation}
 we can write \eqref{eq:Deltamodel} as 
\begin{equation} \label{eq:AgD1}
X^\gD (n+1)= X^\gD(n) + A^\gD (n+1)  X^\gD  (n), 
\end{equation}
where
\begin{equation} \label{eq:AgD2}
X^\gD = \begin{pmatrix}X_1^\gD  \\ X_2^\gD\end{pmatrix} ,
\qquad A^\gD(n) = 
\begin{pmatrix}
0& \gep \gD\\ 
\gep \gD  Z^\gD(n) & Z^\gD(n)-1
\end{pmatrix} 
.
\end{equation}
In different terms: $X^\gD(n)$ results form the product of $n$ independent matrices of the form
$I+A^\gD$ and  for $\gD=1$ we have that the matrix $I+A^1$ coincides 
with
\eqref{eq:matrix} when $Z=Z^1$, that is when $Z$ is log-normal. The restriction to log-normal is just for ease
of exposition: we are going to prove a result (Theorem~\ref{prop:scalingDH83g}) for much more general distributions.

Note that the determinant of $I+A^\gD$ is $Z^\gD (1- \gep^2 \gD^2)$ and we want to exclude the degenerate case:  since
we are going to give a result for  $\gD\searrow 0$, we can assume that this requirement is automatically satisfied. 
The rate of growth of $X^\gD(n)$ is defined by the Lyapunov exponent
\begin{equation}
\label{eq:discL}
\widehat  \cL_{Z^\gD} (\gep)\,=\, \lim_{n \to \infty} \frac 1n \log \|X^\gD(n)\|\;,
\end{equation}
which exists a.s. and is deterministic, see e.g. \cite[Th.~4.1 in Ch.~1]{cf:BL}.

\medskip

\begin{theorem} 
\label{prop:scalingDH83}
For $\gD \searrow 0$, the random process
 \begin{equation} \label{eq:approx1}
\left\{
\left(X_1^\gD\left(\lfloor t/\gD\rfloor\right),X_2^\gD\left(\lfloor t/\gD\rfloor\right)\right)
\right\}_{t\in[ 0, \infty)}\, ,
\end{equation}
converges in law to the diffusion $(X_1(\cdot), X_2(\cdot))$ on the Skorokhod space ${\mathcal D}([0, \infty), (0, \infty)^2)$.
Moreover,
\begin{equation}  \label{eq:approx2}
\lim_{\gD \searrow 0} \frac{ \widehat  \cL_{Z^\gD} (\gep)}{ \gD }\,=\, \cL_{\gs, \ga} (\gep) \;.
\end{equation}
 \end{theorem}

\subsection{Continuum limits and the Derrida-Hilhorst singularity}
\label{sec:DH83}

In \cite{DH}, B. Derrida and H. J. Hilhorst study the $\gep \searrow 0$ limit of the Lyapunov exponent
$\widehat \cL _Z (\gep)$
of product of IID matrices of the form \eqref{eq:matrix}, under the hypothesis that 
$\bbE[Z]>1$ and $\bbE [\log Z]<0$. Since $\gb \mapsto \bbE[Z^\gb]$ is convex, by the hypotheses on $Z$ there exists a unique $\ga\neq 0$ such that $\bbE Z^\ga =1$, and
one readily realizes  that $\ga\in (0,1)$. 
It is claimed in \cite{DH} that
\begin{equation}
\label{eq:DH}
\widehat \cL_Z (\gep) \stackrel{\gep \searrow 0} \sim C \gep^{2\ga}\, ,
\end{equation}
with a semi-explicit expression for $C=C_Z>0$, that depends on the law of $Z$. Such a result directly implies 
a corresponding result for the case $\bbE[\log Z]>0$ and $\bbE [Z^{-1}]>1$: note that in this case $\bbE[Z^\ga]=1$, $\ga \neq 0$,
is again uniquely solved and $\ga \in (-1,0)$. So, by writing 
\begin{equation}
\begin{pmatrix}
1& \gep\\ 
\gep Z  & Z
\end{pmatrix}\, =\, Z \, 
\begin{pmatrix}
Z^{-1}& \gep Z^{-1}\\ 
\gep   & 1
\end{pmatrix}\,,
\end{equation}
we see that 
\begin{equation}
\label{eq:DH2}
\widehat \cL_Z (\gep) \, =\,
\bbE \log Z + \widehat \cL_{1/Z} (\gep)\stackrel{\eqref{eq:DH}}=  \bbE \log Z + C_{1/Z} \gep^{-2\ga} + o\left(  \gep^{-2\ga} \right)\,,
\end{equation}
and we recall that $\ga \in (-1,0)$ now.

Moreover one can find in \cite[Sec.~3]{DH} an argument telling us that for $\ga>1$,
$\ga \not\in \bbN$, one 
expects
\begin{equation}
\label{eq:expectedDH}
\widehat  \cL_Z (\gep)\,=\, c_1 \gep^2+ \ldots +\ldots + c_{\lfloor \ga \rfloor} \gep^{2\lfloor \ga \rfloor} +
C \gep^{2\ga} + o\left(\gep^{2 \ga }\right)\, ,  
\end{equation}
for  real constants $c_j$ and $C$ that are in principle computable. For the case $\ga=0$, 
 i.e. for the case $\bbE \log Z=0$ in which the only solution to $\bbE Z^\ga=1$ is $\ga=0$, one finds 
the prediction 
\begin{equation}
\label{eq:alpha0log}
\widehat \cL_Z (\gep) \stackrel{\gep \searrow 0} \sim \frac C {\log (1 /\gep)}\, ,
\end{equation}
with $C>0$, in more than one reference. We mention here
 \cite[(4.34)]{NL} in which \eqref{eq:alpha0log}
 is found for one dimensional Ising model with random field for a very specific choice of  the disorder (the interaction 
  $J$ of \cite{NL} corresponds to $\log (1/\gep)$). In the localization context \eqref{eq:alpha0log} has been found for example in
\cite[(3.17)]{cf:desbois}.

\medskip
{From our perspective, the significance of all this is that:}
\smallskip

\begin{enumerate}
\item  the continuum limit results of Proposition~\ref{th:asympt} fully match with the expected behaviors (to all orders!),  
	\eqref{eq:DH},  \eqref{eq:expectedDH} and \eqref{eq:alpha0log}, for the random matrix product. We {consider this to be rather striking, and it} 
 highlights the richness of  the continuum limit;
 \item we are going to review the mathematical results available about  \eqref{eq:DH},  \eqref{eq:expectedDH} and \eqref{eq:alpha0log}, but we want to point out that even at the level of physical predictions some results are more sound than others. Notably, it appears to be rather challenging to capture the $\gep^{2\ga}$ singularity for $\vert \ga \vert >1$ and the level of sharpness of the  $\vert \ga \vert \in (0,1)$ prediction \eqref{eq:DH}, even leaving aside mathematical rigor, does not appear to be 
	 easy to achieve. In this sense, the continuum limit goes beyond {what has been established so far for the discrete case}.  
 \end{enumerate} 

\medskip
 
From a mathematical standpoint
a proof of \eqref{eq:DH}, and \eqref{eq:DH2}, (i.e., \eqref{eq:expectedDH} with $\vert \ga \vert \in (0,1)$) has been achieved only recently and 
under the assumption that $Z$ has a $C^1$ density and that the support of $Z$ is bounded and bounded away from zero \cite{cf:GGG}. It is well known, see e.g. \cite{cf:BL}, that the problem of computing the Lyapunov exponents boils down to
identifying the invariant probability of a Markov chain associated to the matrix product. The arguments in \cite{DH}
aim at constructing a probability that for $\gep $ small is expected to be close to the invariant probability.
In \cite{cf:GGG} this construction is put on {rigorous} grounds and, above all, it is shown that 
this probability, although not invariant, is sufficiently close to the invariant one to make it possible to control the Lyapunov exponent with the desired precision. A result about \eqref{eq:expectedDH}, i.e. for $\vert\ga \vert \ge 1$, has been achieved recently \cite{cf:benjamin}, but the the expansion is fully controlled only up to (and excluding) the singular term $C \gep^{2\ga}$: for the moment results about this term remain very weak.

\subsection{On the two-dimensional Ising model with columnar disorder (McCoy-Wu model)} 
\label{sec:MWres}
It is possibly somewhat unexpected, but also computing the free energy of 
the two dimensional Ising model with columnar disorder
(McCoy-Wu model \cite{MWbook,MW1}) boils down to analyzing 
the Lyapunov exponent $\widehat \cL_Z$. 
The McCoy-Wu prediction is remarkable and folklore says that their model is the only non trivial  exactly solvable disordered statistical mechanics model:  we dedicate  Appendix~\ref{sec:MW}
to introducing in detail the model, keeping close to the McCoy-Wu notations. But the key point from the result viewpoint is that 
L. Onsager celebrated solution of the non disordered
case establishes that the free energy, as function of the temperature, has a (logarithmic) divergence
in the second derivative 
at the critical  temperature. B.~M.~McCoy and T.~T.~Wu predict that if a small amount of columnar disorder
(i.e. one dimensional: vertical bounds couplings are random and they are repeated -- i.e. no new randomness is introduced -- on each line) is introduced  
the transition persists but disorder is relevant (in the sense of the \emph{Harris criterion}, that is the disorder changes the critical behavior, see e.g.
\cite[\S~5.3]{cf:G}) and the transition becomes $C^\infty$. A precise form of the singularity is also given. 

As  explained in Appendix~\ref{sec:MW}, McCoy and Wu extend Onsager's approach to the columnar disorder case and the free energy can be written, up to additive analytic terms, in terms of an
integral in the $\gep$ variable of the Lyapunov exponent $\widehat \cL _Z(\gep)$, with $Z$ that has an explicit expression in terms of the parameters of the Ising model. The analysis by McCoy and Wu 
of this expression is performed in two steps:
\smallskip

\begin{enumerate}
\item They claim that in the limit of very narrow disorder 
the relevant -- i.e. singular -- contribution to the free energy can be written as 
\begin{equation}
\label{eq:toyMW}
\tf :\, \ga \mapsto \int_{(0,\eta)}  x
\frac{ K _{\ga-1}( x)}
{ K_\ga ( x)} \dd x\, ,
\end{equation}
with $\eta>0 $ arbitrary (the singular part comes from the small $x$ behavior of the integrand). 
The integrand is just $4\cL_{1, \ga} (x/4)$, and so it is clear from the estimates in Proposition~\ref{th:asympt} that the integral is well defined for all real $\ga$.
\item They argue, by  approximating the integrand by another expression for which 
the exact integration can be performed, that \eqref{eq:toyMW} is $C^\infty$ but not analytic at $\ga=0$.
\end{enumerate}
\smallskip

The approximation in the first step, see Appendix~\ref{sec:MW}, turns out to be precisely the diffusive limit we deal with:
this was possibly expected by comparing \eqref{eq:toyMW} and \eqref{eq:formula}.
What we do with the next result is 
providing     a rigorous analysis of the second step, that is the analysis of \eqref{eq:toyMW}.

\medskip

\begin{theorem}
\label{th:MW}
$\tf$ is real analytic in $(-1,1)\setminus\{0\}$. Moreover it
 is $C^\infty$ but not analytic in $0$.  The radius of convergence of its  Taylor series at the origin   $\sum_{n=0}^\infty c_n \ga^n$  is zero: in fact  $c_1=4\eta $,  $c_{2n+1}=0$ for every $n\in \bbN$ and  the even coefficients satisfy 
 \begin{equation}
c_{2n} \, \stackrel{n \to \infty} \sim \,  4 e^{-\gamma} (-1)^{n+1} \frac{(2n-1)!}{\pi^{2n}} \, ,
\end{equation}
with $\gamma$  the Euler-Mascheroni constant.
\end{theorem}

\medskip 

We are going to prove more. Namely that \eqref{eq:toyMW} defines an analytic  function for every $\ga\in \bbC$
with $0<\vert\Re \ga \vert <1$. We believe that the restriction to  $0<\vert\Re \ga \vert <1$ 
can be removed to get simply to $\vert \Re \ga \vert >0$. However this involves a certain number of complications 
connected to the fact that, with our approach, an \emph{ad hoc} analysis has to be developed for $\Re \ga \in \bbZ$. 
Since the focus is on $\ga=0$, we have made the choice not to develop this issue.  

\medskip

A (very) substantial gap remains between where our results lead and the proof of the McCoy-Wu claim that the
transition is $C^\infty$, even without the precise claim on the nature of the singularity. What we perform, and 
what McCoy and Wu do, is capturing the behavior the free energy near criticality when the disorder is vanishing -- this is  reminiscent 
of \emph{intermediate disorder limits} \cite{cf:AKQ,cf:CSZ} in which, like for us, one enters the framework of  integrable models  -- while the true issue is the behavior for (possibly) weak,  but non vanishing, disorder.

\section{On the Lyapunov exponent: the proof of Theorem~\ref{th:Lyap}}



We use the short-cut notation $\gd:= \gs^2(1-\ga)/ 2 \in \bbR$.
We recall that we can assume $\gep>0$ and let us start by showing that
the process does not hit $(0,0)$. Recall that $(X_1(0),X_2(0))\neq (0,0)$ and set $\tau_{(0,0)}:=
\inf\{ t>0: \, (X_1(t),X_2(t))= (0,0)\}$.
For this let us consider $R(t):= \sqrt{X_1^2(t)+X_2^2(t)}$.
By It\^o's formula:
\begin{equation}
\label{eq:dR}
\begin{split}
\dd R(t)\, &=\, \frac{X_1}R \dd X_1+ \frac{X_2}R \dd X_2 + \frac 12 \frac{X_1^2}{R^3}\dd  \langle  X_2,  X_2 \rangle
\\
&=\, \left(2 \gep \frac{X_1 X_2}{R} + \gd \frac{X_2^2} R + \frac {\gs ^2}2 \frac{X_1^2X_2^2}{R^3} \right)\dd t
+  \gs\frac{X_2^2}R \dd B_t
\\
&=\,
R \left(
2 \gep \frac{Y}{1+Y^2} + \gd \frac{Y^2}{1+Y^2} +
\frac{\gs^2}2 \frac{Y^2}{\left(1+Y^2\right)^2} 
\right) \dd t + R \left( \gs \frac{Y^2}{1+Y^2} \right) \dd B_t
\\
&=:\, R\, D\dd t+ R\, Q\dd B_t
\,
 ,
\end{split}
\end{equation}
where $Y:=X_2/X_1\in [-\infty, \infty]$ and $D=D(t)$ and $Q=Q(t)$ are uniformly bounded continuous stochastic processes ($\Vert D\Vert_\infty \le 2 \gep + \vert \gd \vert + \gs ^2/2$ and $Q\in [0, \gs]$),
defined up to $\tau_{(0,0)}$.
Since, again by It\^o's formula, we have 
\begin{equation}
\dd \log R(t)\, =\, \left(D(t) -\frac 12 Q^2(t) \right)\dd t + Q(t) \dd B_t\, ,
\end{equation}
we  see that $R(t)/R(0)$ is bounded away from zero on every compact time interval. 
This readily yields a contradiction if $\bbP (\tau_{(0,0)}< \infty)>0$. Hence  $\bbP (\tau_{(0,0)}< \infty)=0$
and we have proven that the process does not hit the origin.

\medskip

Now we are going to show that if the initial condition is in the (interior of the) second or fourth quadrant (in the counterclockwise sense), it hits the boundary of these quadrants in an a.s. finite time (in fact, this random time has finite expectation) and enters either the first or third quadrant.
And we show also that once the process is in the first (or third) quadrant, it stays there forever.  

Without loss of generality let us assume that $X_1(0)<0$ and $X_2(0)>0$ (second quadrant). 
For the analysis it  is helpful to consider 
$Y(t)(<0)$ up to $t\le \tau_{-\infty}:=\inf\{t\ge 0: \, Y(t)=-\infty\}$, which coincides 
a.s. with $\inf\{t\ge 0: \, X_1(t)=0\}$, and up to $t \le  
\tau_0:=\inf\{t\ge 0: \, Y(t)=0\}$. 
By It\^o formula
\begin{equation}
\label{eq:EDS}
\dd Y\, =\,  \left(\gep\left(1- Y^2 \right)  + \gd \, Y \right) \dd t + \gs Y \dd B_t\, ,
\end{equation}
and with the specific initial initial conditions we are using is somewhat helpful to 
work with the positive process $\tilde Y=-Y$:
 \begin{equation}
\label{eq:tildeEDS}
\dd \tilde Y\, =\,  \left(\gep\left( \tilde Y^2 -1 \right)  + \gd \, \tilde Y \right) \dd t + \gs \tilde Y \dd B_t\, ,
\end{equation}
which is in $(0, \infty)$ as long as the two dimensional process does not leave the interior of the quadrant.
We use the stopping times $\tilde \tau_0$ and $\tilde \tau _\infty$ with the obvious meaning. 
We are going to apply the Feller test for explosion  to show that  $\tilde \tau:=\min( \tilde\tau_0, \tilde\tau_\infty)$ is in $L^1$ so 
\begin{equation}
\bbP\left( \tilde\tau< \infty\right))\, =\, 1\, ,
\end{equation}
which means that, almost surely, the process hits the axes. 
And if $\tilde\tau_\infty< \tilde\tau_0$, that is if $(X_1(\tilde\tau),X_2(\tilde\tau))=(0,x_2)$, $x_2>0$, we readily see
from \eqref{eq:sys} that $X_1(\tilde\tau+t)>0$, at least for $t>0$ small.
If instead $\tilde\tau_0< \tilde\tau_\infty$, then $(X_1(\tilde\tau),X_2(\tilde\tau))=(x_1,0)$, $x_1<0$, 
and again from \eqref{eq:sys} one sees that $X_2(\tilde\tau+t)<0$ for $t>0$ small: {since the equation solved by $X_2$ is stochastic, the argument is slightly more delicate than for the previous case and we give some details. 
By the Strong Markov property it suffices to consider $(X_1(0),X_2(0))=(x_1,0)$, $x_1<0$, and 
$X_2(t)= \int_0^t (\gep X_1(s) +c X_2(s)) \dd s +M(t)$, with the constant $c$ and the centered Martingale $M$ easily read out of  \eqref{eq:sys}. Note that $M$ is a time changed Brownian motion. By continuity of $(X_1(\cdot), X_2(\cdot))$ we readily see that $X_2(t)-M(t) \le -\vert x_1\vert t/2$ for $t$ small. It is therefore  clearly impossible that
$\inf\{t>0: X_2(t)\neq 0\}$ is positive, because this would imply $X_2(t)<-\vert x_2 \vert /2$ for small $t$. Therefore $t\mapsto \int_0^t X_2(s)^2\dd s$ is increasing  at least for $t$ small, which implies that  the time change is non degenerate at least for small times. Hence $M(t)$ becomes negative for arbitrarily small values of $t$. Therefore   $X_2(t) < -t\vert x_1\vert/2$ -- in particular, it is negative -- for arbitrarily small values of $t$.}
An application of the Feller test, this time applied to $Y$ and not to $\tilde Y$, actually shows that
if $Y$ is in  $(0, \infty)$, then it will stay so for all times, that is the interior the first and third quadrants are stable sets
for the dynamics.

Let us detail the application of the Feller test. Let $Z$ is a one dimensional diffusion with
$Z(0)\in (0, \infty)$ and
\begin{equation}
\dd Z (t)\, =\, b(Z(t)) \dd t + q(Z(t))\dd B_t\, ,
\end{equation} 
$b(\cdot)$ and $q(\cdot)(>0)$ differentiable functions. We set $\tau:= \inf\{ t>0: \, Z(t)=0$ or $Z(t)=\infty\}$ and
\begin{equation}
\label{eq:sandv}
s(z)\, :=\, \int_1 ^z \exp \left( -2 \int_1^y \frac{b(r)}{q^2(r)}\dd r\right)\dd y
\ \text{ and } \
v(z)\, :=\, \int_1^z s'(y)\left( \int_1^y \frac 2{s'(r)q^2(r)} \dd r \right) \dd y\, .
\end{equation}
By monotonicity the limits of $v(z)$ for $z\searrow 0$ and $z \nearrow \infty$ exist in $[0, \infty]$
and they will be simply denoted by $v(0)$ and $v(\infty)$. 
If both $v(0)<\infty$ and $v(\infty)<\infty$ then 
$\bbE [\tau]< \infty$ \cite[Prop.~5.32, Ch.~5]{cf:KS}. On the other hand, if $v(0)=v(\infty)= \infty$ then 
$\bbP(\tau= \infty)=1$ \cite[Th.~5.29, Ch.~5]{cf:KS}  .  


Let us start with the $\tilde Y$ case (cf. \eqref{eq:tildeEDS}): we have 
\begin{equation}
\label{eq:pp1}
s'(z)\,=\, \frac C {z^{1-\ga}}\exp\left(-\frac{2\gep}{\gs ^2}\left(z+\frac 1z\right) \right)
\,,
\end{equation}
so for $z\ge 1$
\begin{equation}
\label{eq:defasymp}
s'(z) \asymp
z^{-1+\ga} \exp\left(-\frac{2\gep}{\gs ^2}z \right) \, ,
\end{equation}
with the notation $f(z) \asymp g(z)$ if 
$f(z)/ g(z) \in [a, 1/a]$ on the prescribed interval for some  $a \in (0,1)$,
and 
\begin{equation}
v'(z) \asymp \frac 1{z^2}\, ,
\end{equation}
so $v(\infty)< \infty$. In a very similar way, for $z \in (0, 1]$
\begin{equation}
s'(z) \asymp
z^{-1+\ga} \exp\left(-\frac{2\gep}{\gs ^2} \frac 1z \right) \, ,
%
\end{equation}
and the estimate of $v(0)$ is identical to  the one for $v(\infty)$  via a (double) change of variable $z \mapsto 1/z$. Hence 
$v(0)< \infty$ and the diffusion $\tilde Y$ hits $0$ or $\infty$ at a random time  which has finite expectation.

For the case of $Y$ we turn to \eqref{eq:EDS} and the difference is that
the factor $(z+1/z)$ in the exponent in \eqref{eq:pp1} changes sign. Once again, we can replace
$(z+1/z)$ by $z$ for $z\ge 1$, and by $1/z$ for $z \le 1$. This implies that the integral with respect to 
$r$ in the expression for $v(z)$ in \eqref{eq:sandv} stays bounded and bounded away from zero
both for $y\nearrow \infty$ and for $y \searrow 0$. The integral with respect to $y$ therefore diverges 
both for $z \nearrow \infty$ and $z \searrow 0$ (once again, the two computations are identical, up to change of variables). Therefore, almost surely, $Y$ hits neither $0$ nor $\infty$.

\medskip

We are now going to show that the diffusion $Y$ has a unique invariant probability, that we will make explicit,
on $(0, \infty)$.  This corresponds to the two (extremal) invariant probabilities
for the normalized process $(X_1, X_2)/ \sqrt{X_1^2+X_2^2}$, supported on the intersection of the unit circle with the 
first (or third) quadrant. For this it is practical 
to observe that  the generator of the evolution \eqref{eq:EDS} acts on $C^2$ functions $f: (0, \infty)\to \bbR$ as
\begin{equation} 
\label{eq:generator}
L_\gep f(y)\, =\, \left( \gep (1-y^2) + \gd \, y\right) f'(y) +\frac {\gs^2}2 y^2 f''(y)\, =\, 
\frac{\gs^2}{2 p_\gep(y)} \left( y^2 p_\gep(y) f'(y) \right)'\,,
\end{equation}
where $p_\gep(\cdot)$ is the probability density 
\begin{equation}
\label{eq:p_gep}
p_\gep(y)
\, =\, \frac{C_\gep}{y^{1+ \ga }} \exp\left(- \frac{2 \gep}{\gs^2}  \left(y + \frac 1y \right)
\right)\ \ \ \ \text{ with }  C_\gep^{-1}
\, =\, 2 K_\ga \left( 4 \gep/\gs^2\right)\,,
\end{equation}
and $K_\ga(\cdot)$ is defined in \eqref{eq:Bessel-def}.
This already makes evident the reversible nature of the diffusion $Y$ and, in particular,  
\eqref{eq:p_gep} is an invariant probability. The transformation $S(t):= \log Y(t)$ makes things even more straightforward:
$S$ is a diffusion on $\bbR$ with constant diffusion coefficient and a strongly confining potential: 
\begin{equation}
\label{eq:S}
\dd S\, =\, -U'(S) \dd t + \gs \dd B_t\ \  \ \text{ with } \ U(s)\, :=\, \gep \left( \exp(-s) +\exp(s) - \left(\gd - \frac{\gs^2}{2\gep}\right) s\right)\, .
\end{equation}
An invariant probability of this diffusion is $\tilde p_\gep(s)\propto \exp(-2 U(s)/\gs^2)$ and the generator has the familiar
symmetric form $\tilde L_\gep g= (\gs^2/2) ( \tilde p_\gep g')'/\tilde p_\gep$, for $g \in C^2(\bbR, \bbR)$, see e.g \cite[p.111]{cf:FW}.

Uniqueness of this invariant measure as well as ergodic  properties can be established in a variety of ways:    \cite[Th. 5.1]{MaruyamaTanaka}
gives a Pointwise Ergodic Theorem  that one can directly apply to \eqref{eq:S}, and of course it implies uniqueness. Alternatively one
can put \eqref{eq:EDS} or \eqref{eq:S} in {\sl natural scale} via a time change and a scale function, see  \cite[Ch.~V]{RogersWilliams}, and apply 
the Ergodic Theorem \cite[Ch.~V, Th. 53.1]{RogersWilliams}.  Ergodic properties of $S$ are also given in   \cite{PardouxVeretennikov}. 
\medskip

Therefore for every choice of $Y(0)\in (0, \infty)$, almost surely  and in $L^1 $ we have that
\begin{equation}
\label{eq:mlterg}
\begin{split}
 \lim_{t \to\infty } \frac 1t\log X_1(t) \,&=\,  \gep \lim_{t \to\infty } \frac 1t \int_0^t Y(s) ds \\
 &=\, \gep \int_0^\infty y \, p_\gep (y) \dd y 
 \, =\,  \frac{\gep K_{\ga-1}\left(4 \gep/\gs^2\right)}
{K_{\ga}\left(4 \gep/\gs^2\right)}\, ,
\end{split}
\end{equation}
where in the first step we have used the first identity in
\begin{equation}
\label{eq:fsys}
\begin{split}
X_1(t) \, &=\, X_1(0)\exp \left( 
\gep \int_0^t  Y(s) \dd s\right)\, ,
\\
X_2(t) \, &=\, X_2(0)\exp \left( \gep
\int_0^t \frac {1}  {Y(s)} \dd s  - \ga \frac{\gs ^2}2 t + \gs B_t\right)\, ,
\end{split}
\end{equation}
which is directly derived from \eqref{eq:sys} and holds for all $t>0$ if both $X_1(0)$ and $X_2(0)$ are positive
(or both are negative: in general,  the formula holds up to the hitting time of the boundary of the quadrant in which $(X_1(0), X_2(0))$ lies).
The second step in \eqref{eq:mlterg} is the application of the Pointwise Ergodic Theorem and the last one is an explicit computation.
In the same way, by using the second identity in \eqref{eq:fsys}
we get to (with $x=4 \gep/\gs^2$)
\begin{equation}
\label{eq:mlterg2}
\begin{split}
 \lim_{t \to\infty } \frac 1t\log X_2(t) \,&=\,  \gep \lim_{t \to\infty } \frac 1t \left( \int_0^t \frac 1{Y(s)} ds\right) - \ga \frac{\gs ^2}2 \\
 &=\, \gep \int_0^\infty \frac 1y \, p_\gep (y) \dd y - \ga \frac{\gs ^2}2
 \\&=\, \frac{\gs^2} 4 \left( \frac{x K_{1+\ga} (x)}{K_\ga (x)} -2 \ga\right)
 \stackrel{\eqref{eq:Bessel-id}}=  
  \frac{\gs^2} 4 \frac{x K_{1-\ga} (x)}{K_\ga (x)}
  \, =\,  \frac{\gep K_{\ga-1}\left(4 \gep/\gs^2\right)}
{K_{\ga}\left(4 \gep/\gs^2\right)}\, ,
 \end{split}
\end{equation}
which coincides with what we found in \eqref{eq:mlterg}. This shows that both components have the same exponential growth rate, hence also the norm of $(X_1(t), X_2(t))$, and \eqref{eq:formula} is proven.
If instead of starting from the first quadrant, we were starting from the second quadrant, the result is unchanged because the second quadrant is abandoned after a random time that is in $L^1$. This completes the proof of Theorem~\ref{th:Lyap}.
\qed

\section{Fluctuations of the Lyapunov exponent: proofs}
\label{sec:fluct}

\phantom{a}

\noindent
\emph{Proof of Proposition~\ref{prop:clt}.}
Recall that the ratio $Y(t)=X_2(t)/X_1(t)$ converges.  
Hence it is sufficient to prove the convergence of  \eqref{eq:cltLyap}
with $X_1(t)$ instead of $ \left \Vert (X_1(t), X_2(t))\right\Vert $ in the logarithm, i.e., to prove convergence in law of
\begin{equation}
\label{eq:cltLyap2}
\left\{
 \frac {1}{\sqrt t} \left( \int_0^t f(Y(s)) ds  \right) 
 \right\}_{t \in (0, \infty)}\, ,
\end{equation}
where the function $f(y)=\gep y- \cL_{\gs, \ga} (\gep) $ is centered for the invariant density $p_\gep(\cdot)$. We start by solving the Poisson equation,
$L_\gep g = f$.
 \cite[Th.~1]{PardouxVeretennikov} applies for $S=\log Y$, see \eqref{eq:S}, and shows that the Poisson equation has $g(y)= \int_0^\infty \bbE_y[ f(Y(s))] ds$ as unique solution.
We need here an explicit form, and we solve the linear equation
\begin{equation}
\label{eq:edo}
\frac {\gs^2}2 y^2 h' + \left( \gep (1-y^2) + \gd\, y\right) h \, = f   
\end{equation}
for $h=g'$ by the method of variation of constants. The homogeneous equation -- when the  right-hand side of \eqref{eq:edo} is equal to 0 -- admits
$h_0(y) = \big(y^2 p_\gep(y)\big)^{-1}$ as a solution. Looking now for solutions of the form $h(y)=k(y) h_0(y)$ for  \eqref{eq:edo} itself,
we find that 
\begin{equation}
k(y) = \frac 2{\gs^2} \int_0^y  \big(\gep z-\cL_{\gs, \ga} (\gep) 
 \big) 
 p_\gep(z) dz + C\,,
\end{equation}
and we choose $C=0$ (this is the only choice that yields the required integrability properties in what follows).
Finally,
\begin{equation}
\label{eq:solg}
g'(y) \,=\,  \left(y^2 p_\gep(y)\right)^{-1}  \frac 2{\gs^2} \int_0^y  \big(\gep z-\cL_{\gs, \ga} (\gep) 
 \big) 
 p_\gep(z) \dd z\;, 
\end{equation}
and the value of $g(1)$ does not matter for our purpose. 
Now, we can  follow a standard proof of Central Limit Theorem  for reversible diffusions, e.g. \cite[Sec.~2]{CattiauxChafaiGuillin}. 
Note in fact that $g$ is smooth 
and that for $y \to \infty$ (recall the notation used in \eqref{eq:defasymp})
\begin{equation}
g'(y) \, \asymp\, y^{\ga -1}e^{\frac {\gep }2 y} \left(\int_y^\infty z^{-\ga} e^{- \frac {\gep }2 z} \dd z\right) 
\, \asymp \, \frac 1 y\, ,
\end{equation}
where in the first step we use that $\int_0^y  (\gep z-\cL_{\gs, \ga} (\gep) 
) 
 p_\gep(z) \dd z= \int_y^\infty  (\cL_{\gs, \ga} (\gep) -\gep z
) 
 p_\gep(z) \dd z$ and that $\cL_{\gs, \ga} (\gep)$ is just a constant. For $y\searrow 0$ instead 
 \begin{equation}
g'(y) \, \asymp\, y^{\ga -1}e^{\frac {\gep }{2 y}} \left( 
\int_0^y z^{-1-\ga } e^{-\frac {\gep }{2 z}} \dd z \right)  \, =\, 
 y^{\ga -1}e^{\frac {\gep }{2 y}} \left( 
\int_{1/y}^\infty z^{-1+\ga } e^{-\frac {\gep z}{2 }} \dd z \right)  \, \asymp \, 1\, .
 \end{equation}
Therefore $\sup_y \vert g'(y) \vert< \infty$ and by  It\^o's formula we obtain that
\begin{equation}
\label{eq:martM}
\begin{split}
M_t \,:=\,  \gs \int_0^t Y_s g'(Y_s) \dd B_s \, &=\,  g(Y_t)-g(Y_0) - \int_0^t L_\gep g (Y_s) \dd s \\
& =\, g(Y_t)-g(Y_0) - \int_0^t f (Y_s) \dd s  \, ,
\end{split}
\end{equation}
is a martingale with bracket 
\begin{equation} 
\langle M \rangle_t\, =\,  \gs^2 \int_0^t  Y_s^2 g'(Y_s)^2 \dd s\, .
\end{equation}
By the ergodic theorem, as $t \to \infty$, almost surely
\begin{equation} 
\label{eq:cvcrochet}
\begin{split}
\frac{1}{t} \langle M \rangle_t
&\longrightarrow  \gs^2  \int_0^\infty y^2 g'(y)^2 p_\gep(y) \dd y \\
 &\ = \frac 4{\gs^2}  \int_0^\infty  \frac 1{y^2 p_\gep} \left( \int_0^y \left(\gep z-\cL_{\gs, \ga} (\gep) 
 \right) p_\gep(z) \dd z \right)^2 \dd y \, 
 =\, v_{\gs, \ga} (\gep)\, . 
\end{split}
\end{equation}
This  deterministic limit is  finite  in view of the (exponential) decay of $\int_0^y \big(\gep z-\cL_{\gs, \ga} (\gep) 
 \big) p_\gep(z) \dd z$ as $y \to 0$ and $y \to \infty$. Then, 
the central limit theorem for martingales applies, and 
 $t^{-1/2}M_t$ converges in law to a centered Gaussian with variance given by 
$ v_{\gs, \ga} (\gep)$.  Now, the first two terms in the last line of \eqref{eq:martM} are bounded in probability, so 
$-t^{-1/2} \int_0^t f(Y_s)ds$ converges to the same limit as $t^{-1/2}M_t$, and \eqref{eq:cltLyap2} is proved. 
Therefore the proof of Proposition~\ref{prop:clt} is complete.
\qed

\medskip

\noindent
\emph{Proof of Proposition~\ref{prop:asympt}.} 
 In view of  \eqref{eq:varLyap} and of the fact that we know the asymptotic behavior of $K_\ga(x)\sim x^{-\vert \ga\vert}\Gamma (\vert\ga\vert)/2^{1-\vert\ga\vert}$ for $\ga\neq 0$ and $K_0(x) \sim \log (1/x)$, what we have to estimate is 
 \begin{multline}
 \label{eq:fluct-T12}
  \int_0^\infty \frac{1}{y^{1-\ga}}e^{\frac{x}{2}\left(y+\frac{1}{y}\right)} \left( \int_0^y \frac{\gep z-\cL_{\gs, \ga}}{z^{1+\ga}} e^{-\frac{x}{2}\left(z+\frac{1}{z}\right)} 
  \dd z \right)^2 \dd y
\, =
\\ 
\int_1^\infty \frac{1}{y^{1-\ga}}e^{\frac{x}{2}\left(y+\frac{1}{y}\right)} \left( \int_y^\infty \frac{\gep z-\cL_{\gs, \ga}}{z^{1+\ga}} e^{-\frac{x}{2}\left(z+\frac{1}{z}\right)} 
  \dd z \right)^2 \dd y + \ \ \ \ \ \ \ \ \ \ \ \ \ \ \ \ \ \
  \\
  \int_0^1 \frac{1}{y^{1-\ga}}e^{\frac{x}{2}\left(y+\frac{1}{y}\right)} \left( \int_0^y \frac{\gep z-\cL_{\gs, \ga}}{z^{1+\ga}} e^{-\frac{x}{2}\left(z+\frac{1}{z}\right)} 
  \dd z \right)^2 \dd y \, =:\, T_1(x)+T_2(x)\, .
\end{multline}

\smallskip

\begin{rem}
\label{rem:simplify}
In view of Proposition~\ref{th:asympt} we know that   $\cL_{\gs, \ga}=O(\gep^{\min(2\ga, 2)})$ for $\ga>0$, except for $\ga=1$
for which there is a  logarithmic correction. Therefore $ \cL_{\gs, \ga} =o( \gep )$ if  $\ga >1/2$ and, since $z\ge 1$, in dealing with  $T_1(x)$ we can safely neglect the term containing $ \cL_{\gs, \ga}$ for $\ga >1/2$.  On the other hand, 
in dealing with  $T_2(x)$ we can safely neglect the term not containing $ \cL_{\gs, \ga}$ for $\ga <1/2$. In fact 
$ \cL_{\gs, \ga}$ is much greater than $\gep$, hence of $\gep z$ ($z \le 1$ for $T_2$), for $\ga <1/2$. 
\end{rem}
\smallskip

To make the expressions more compact and readable we choose 
\begin{equation}
{\gs^2}\, =\, 2\, ;
\end{equation}
the general case is easily recovered by a scaling argument.

\medskip

We start with the analysis of $T_1$. By a change of variable we have:
\begin{equation}
\label{eq:T1-start}
T_1(x) \, =\, 
 \left( \frac x 2 \right)^{\ga }\int_{x/2}^\infty y^{\ga -1}e^{y+ (x/2)^2/y}  \left( \int_y^\infty \left( z^{-\ga}  -\cL_{\gs, \ga} z^{-\ga-1}\right)e^{-z- (x/2)^2/z} \dd z
\right)^2  \dd y\, .
\end{equation} 
We claim that for $\ga<0$ we simply have 
\begin{equation}
\label{eq:T1-alpha<0}
T_1(x) \stackrel{x \searrow 0}\sim 
 \left( \frac x 2 \right)^{\ga }\int_{0}^\infty y^{\ga -1}e^{y}  \left( \int_y^\infty \left( z^{-\ga}  -\vert \ga\vert z^{-\ga-1}\right)e^{-z} \dd z
\right)^2  \dd y\, =:\, 2^{\vert \ga \vert}   \Gamma(\vert \ga \vert)x^{\ga }\,  .
\end{equation} 
For this  choose $\gd\in(0,1)$ and split  the integral in $y$
in \eqref{eq:T1-start} as $\int_{x/2}^\infty \ldots=  \int_{x/2}^\gd \ldots +   \int_{\gd}^{1/\gd} \ldots +  \int_{x/2}^{1/\gd}\ldots:= I_1+I_2+I_3$. The limit $x \searrow 0$ is easily taken in $I_2$ and the dependence on $x$ disappears. Moreover we directly check that $\lim_{\gd\searrow 0} \lim_{x\searrow 0}I_2\in (0, \infty)$ is the integral in the right-hand side of \eqref{eq:T1-alpha<0}. 
We are left with showing that $\lim_{\gd \searrow 0} \sup_{x\in(0, 2\gd)} I_j=0$ for $j=1$ and $3$.
For $I_1$ recall that in \eqref{eq:T1-start}  we can replace $\int_y^\infty \ldots \dd z$ with  $\int_0^y \ldots \dd z$,  so that 
for $\gd$ sufficiently small (so $x$ is small too and we can use the asymptotic approximation of $\cL_{\gs, \ga}\sim \vert \ga \vert$) we have 
\begin{equation}
\sup_{x\in(0, 2\gd)}\vert I_1\vert \, \le \, 
\int_{0}^\gd y^{\ga -1}  \left( \int_0^y \left( z^{-\ga}  +2 \vert \ga \vert  z^{-\ga-1}\right) \dd z
\right)^2   \dd y \stackrel{\gd\searrow 0 } \longrightarrow 0\,.
\end{equation}
Moreover (again, $\gd$ small)
\begin{multline}
\sup_{x\in(0, 2\gd)}\vert I_3\vert \, \le \, 2
\int_{1/\gd}^\infty y^{\ga -1}e^{y}  \left( \int_y^\infty \left( z^{-\ga}  + 2 \vert \ga \vert z^{-\ga-1}\right)e^{-z} \dd z
\right)^2  \dd y\, \le \\
\int_{1/\gd}^\infty e^{3y/2}  \left( \int_y^\infty e^{-z} \dd z  
\right)^2  \dd y \stackrel{\gd\searrow 0 } \longrightarrow 0\, ,
\end{multline}
and \eqref{eq:T1-alpha<0} is proven.

\medskip

For $\ga=0$ we have 
\begin{equation}
\label{eq:T1-start0}
T_1(x) \, =\, 
 \int_{x/2}^\infty y^{-1}e^{y+ (x/2)^2/y}  \left( \int_0^y \left( 1  -\cL_{\gs, 0} z^{-1}\right)e^{-z- (x/2)^2/z} \dd z
\right)^2  \dd y\, .
\end{equation} 
We anticipate (for future use) that the result we are going to obtain would be the same if $1  -\cL_{\gs, 0} z^{-1}$ is
replaced by $\cL_{\gs, 0} z^{-1}$
Again,  $\int_0^y\ldots$ can be replaced by $\int_y^\infty\ldots$  and it suffices the splitting
$ \int_{x/2}^\infty\ldots =  \int_{x/2}^\gd \ldots +  \int_{\gd}^\infty \ldots =:I_1+I_2$.
In fact (recall that $\cL_{\gs,0}=o(1)$ as $x \searrow 0$, in particular $\cL_{\gs,0}$ becomes smaller than one) 
\begin{equation}
 \sup_{x\in(0, 2\gd)}\vert I_2 \vert \, \le \, 2\int_\gd ^\infty
 y^{-1} e^y \left(\int_y^\infty (1+ z^{-1})e^{-z} \dd z\right)^2
  \dd y\, ,
\end{equation}
and the right-hand side is just a finite expression that depends on $\gd$.
On the other hand $I_1$ diverges as $x \searrow 0$. In fact observe that
\begin{equation}
\label{eq:tilde-I}
e^{-2 \gd}\tilde I
\, \le \, I_1 \, \le \, 
e^{2 \gd} \tilde I\,, \text{ with } \tilde I \, :=\, \int_{x/2}^\gd y^{-1}  \left( \int_0^y \left( 1  -\cL_{\gs, 0} z^{-1}\right)e^{- (x/2)^2/z} \dd z
\right)^2  \dd y\, ,
\end{equation}
so we can focus on $\tilde I$. By using $\int_0^L (1/x)\exp(-1/x) \dd x \sim \log L$ for $L \to \infty$,
and the fact that $\gd/x^2\le y/x^2\le 1/(2x)$ we see that 
\begin{equation}
\int_0^y  z^{-1} e^{- (x/2)^2/z} \dd z \sim \log(y/ x^2)\, ,
\end{equation}
uniformly in the range of $y $ we are using, and as $x \searrow 0$. Therefore
\begin{equation}
\tilde I \sim \left( \frac 1{2 \log(1/x)}\right)^2 \int_{x/2}^\gd y^{-1}  \left( 2 \log(1/x) - \log (1/y)\right)^2 \dd y\, \sim\,
\frac 7{12} \log (1/x)\,.
\end{equation}
This concludes the $\ga =0$ case: $T_1(x) \sim \frac 7{12} \log (1/x)$.

\medskip

The case $\ga>0$ is quicker to treat for $\ga>1/2$ because of Remark~\ref{rem:simplify}. In reality also for $\ga\in (0,1/2]$ the term containing $\cL_{\gs, \ga}$ does not contribute: we will check  this fact after  estimating what is giving the main contribution:
\begin{equation}
\label{eq:T1-start>0}
 \left( \frac x 2 \right)^{\ga }\int_{x/2}^\infty y^{\ga -1}e^{y+ (x/2)^2/y}  \left( \int_y^\infty  z^{-\ga} e^{-z- (x/2)^2/z} \dd z
\right)^2  \dd y\, ,
\end{equation} 
and the final result is that for $\ga \in (0,2)$
\begin{equation}
\label{eq:forq1}
T_1(x) \,\sim\,  
\left( \frac x 2 \right)^{\ga }
\int_{0}^\infty y^{\ga -1}e^{y}  \left( \int_y^\infty  z^{-\ga} e^{-z} \dd z
\right)^2  \dd y\, =:\, q_1(\ga) x^{\ga }\,.
\end{equation}
This is proven like before by restricting the integral to $y \in (\gd,1/\gd)$ and estimating the rest before 
letting $\gd \searrow 0$. The function $q_1(\cdot)$ can be expressed with a  Meijer G-function, but this does not make it 
much more explicit. 
Let us quickly verify that the term we neglected for $\ga \in (0, 1/2]$ is of lower order: 
by focusing on $y\in (x/2, \gd)$ (otherwise the fact is obvious) we see that an upper bound on this contribution
is $O(x^{5\ga}) \int_{x/2}^\gd y^{-\ga -1} \dd y = O(x^{4\ga})$.

\medskip

The $\ga=2$ case generates a logarithmic correction: in fact from \eqref{eq:T1-start>0} we see that if we restrict the integral over $y\ge \gd$, the contribution is bounded by $x^2$ times a constant that depends only on $\gd$. 
The integral with $y\in (x/2, \gd)$ instead is controlled above and below, up to a factor that can be chosen 
arbitrarily close to one uniformly in $x\searrow 0$ by choosing $\gd$ small (like in \eqref{eq:tilde-I}), by
\begin{equation}
\left( \frac x 2 \right)^{2} \int_{x/2}^\gd y \left( \int_y^\infty z^{-2}e^{-z} \dd z\right)^2\, \dd y
\stackrel{x \searrow 0}\sim
 \left( \frac x 2 \right)^{2} \log (1/x)\,.
\end{equation}

\medskip

For $\ga >2$ we go back to \eqref{eq:T1-start>0}
\begin{multline}
\label{eq:T1-alpha>2}
\left( \frac x 2 \right)^{\ga }\int_{x/2}^\infty y^{\ga -1}e^{y+ (x/2)^2/y}  \left( \int_y^\infty  z^{-\ga} e^{-z- (x/2)^2/z} \dd z
\right)^2  \dd y\, = \left( \frac x 2 \right)^{\ga }\int_{1}^\infty \ldots + \left( \frac x 2 \right)^{\ga }\int_{x/2}^1 \ldots
\le 
\\ 
C x^\ga + 
 3\left( \frac x 2 \right)^{\ga }\int_{x/2}^\infty y^{\ga -1}  \left( \int_y^\infty  z^{-\ga} e^{-z} \dd z
\right)^2  \dd y\,
\le \, C' x^{\ga} \left(1+ \int_{x/2}^\infty y^{-\ga +1}    \dd y\right)\, =\, O(x^2)\, ,
\end{multline}
where $C$ and $C'$ are constants independent of $x$.

\medskip

We collect what we have obtained:
\begin{equation}
\label{eq:T1-finale}
T_1(x)\, \sim\, \begin{cases}
2^{\vert \ga \vert}   \Gamma(\vert \ga \vert)
 x^\ga & \text{ if } \ga \in (-\infty,0)\, ,\\
\frac 7{12} \log (1/x) & \text{ if } \ga=0\, , \\
q_1(\ga) x^\ga & \text{ if } \ga \in (0,2) \, ,\\
 \frac 14 x^2\log (1/x) & \text{ if } \ga=2\, , \\
O(x^2)& \text{ if } \ga>2\, .
\end{cases}
\end{equation}

\medskip

We now turn to $T_2(x)$ and the basic expression is after a change of variables  (still, $\gs =\sqrt{2}$)
\begin{multline}
 T_2(x)\, =\\
   \left( \frac x 2 \right)^{4-\ga}\int_{x/2}^\infty u^{-1-\ga} e^{u+(x/2)^2/u} \left(
 \int_u^\infty \left(  v^{-2+\ga} 
 - \left( \frac {2 \cL_{\gs, \ga}}{x^2} \right) v^{-1+\ga}\right) e^{-v-(x/2)^2/v} \dd v
  \right)^2 \dd u\, .
 \end{multline} 

Let us start with $\ga<0$ and recall that by Remark~\ref{rem:simplify} it suffices to consider 
\begin{equation}
\left( \frac x 2 \right)^{4-\ga}  \left( \frac {2 \cL_{\gs, \ga}}{x^2} \right) ^2
\int_{x/2}^\infty u^{-1-\ga} e^{u+(x/2)^2/u} \left(
 \int_u^\infty   v^{-1+\ga} e^{-v-(x/2)^2/v} \dd v
  \right)^2 \dd u\, .
  \end{equation}
The pre-factor behaves asymptotically as $\ga^2 (x/2)^{-\ga}$ and the integral can be bounded by two times 
\begin{equation}
\int_{x/2}^\infty u^{-1-\ga} e^{u} \left(
 \int_u^\infty   v^{-1+\ga} e^{-v} \dd v
  \right)^2 \dd u\, \sim \, \frac 1{\vert \ga \vert^3}(x/2)^\ga\, .
\end{equation}
So $T_2(x)=O(1)$ for $\ga <0$.

\medskip

For $\ga=0$ the expression to evaluate is
\begin{equation}
\cL_{\gs,0}^2 \int_{x/2}^\infty u^{-1} e^{u+(x/2)^2/u} \left(
 \int_u^\infty   v^{-1} e^{-v-(x/2)^2/v} \dd v
  \right)^2 \dd u\, .
\end{equation}
But this term is minimally different from
\eqref{eq:T1-start0} (see observation right after\eqref{eq:T1-start0}) and exactly in the same way we arrive at
$T_2(x) \sim T_1(x) \sim \frac 7{12} \log (1/x)$. 

\medskip

For $\ga\in (0,2)$ we split the integral with respect to $u$ and the contribution when $u\ge 1$ is bounded 
so the contribution to $T_2(x)$ is $O(x^{4-\ga})$. For $u<1$ we make an upper on the  contribution to $T_2(x)$:
\begin{multline}
3\left( \frac x 2 \right)^{4-\ga}\int_{x/2}^1 u^{-1-\ga} \left(
 \int_0^u \left(  v^{-2+\ga} 
 + \frac {2 \cL_{\gs, \ga}}{x^2}  v^{-1+\ga}\right) e^{-v} \dd v
  \right)^2 \dd u
   \\
\le \,   C x^{4-\ga } \left(
  \int_{x/2}^1 u^{-3+\ga} \dd u+ \left( \frac {2 \cL_{\gs, \ga}}{x^2} \right)^2 \int_{0}^1 u^{-1+\ga} \dd u\right) \, 
  \le \, C' \left( x^2  + \cL_{\gs, \ga}^2 x^{-\ga} \right)\, ,
\end{multline} 
and since $\cL_{\gs, \ga}= \max(x^{2\ga}, x^2)$, except for a logarithmic correction for $\ga=1$,  
we conclue that $T_2(x)=O(\max(x^{3\ga}, x^2))$ for $\ga \in (0,2)$.

\medskip

For $\ga=2$ we again split the integral with respect to $u\ge \gd$ and $u< \gd$. The integral for 
$y\ge \gd$ is bounded by a constant that depends only on $\gd$. 
Arguing as in  \eqref{eq:tilde-I} we see that what it suffices to control
\begin{equation}
\int_{x/2}^\gd  u^{-3}  \left(
 \int_0^u \left(  1
 - \left( \frac {2 \cL_{\gs, 2}}{x^2} \right) v \right) e^{-v} \dd v
  \right)^2 \dd u\,  \sim \, \log (1/x)\,,
\end{equation}
where we have used that $ \cL_{\gs, 2}=O({x^2})$. Therefore $T_2(x) \sim (x/2)^2  \log (1/x)$ for $\ga =2$.

\medskip

Finally, for $\ga>2$ we have
\begin{equation}
\label{eq:forq2alpha}
T_2(x) \sim \left( \frac x 2 \right)^{4-\ga}\int_{0}^\infty u^{-1-\ga} e^{u} \left(
 \int_u^\infty \left(  v^{-2+\ga} 
 -  \frac{v^{-1+\ga}}{\ga-1}\right) e^{-v} \dd v
  \right)^2 \dd u =  \frac{2^{\ga-4}\Gamma(\ga-2)}{(\ga-1)^2} x^{4-\ga}.
\end{equation}
The proof of this claim 
follows the same line as the proof of \eqref{eq:T1-alpha<0}, that is, splitting of the $y$ integral in three parts and taking 
the limit $\gd \searrow 0$. 

\medskip

We have got to:
 \begin{equation}
 T_2(x)\, \sim \begin{cases} 
 O(1) & \text{ if } \ga<0\, , \\
   \frac 7{12} \log (1/x) & \text{ if } \ga=0\, , \\
 O(\max(x^{3\ga}, x^2))  & \text{ if } \ga\in (0,2)\, , \\
 \frac 14 x^2 \log (1/x) & \text{ if } x=2\, , \\
    \frac{2^{\ga-4}\Gamma(\ga-2)}{(\ga-1)^2} x^{4-\ga} & \text{ if } \ga>2\, .
 \end{cases}
 \end{equation}
 Therefore only $T_2$ contributes to the final result  for $\ga >2$. Otherwise only  $T_1$  contributes, except 
 at $\ga=0$ and $2$ where they both contribute and exactly with the same amount: 
 \begin{equation}
 T_1(x) +T_2(x) \, \sim\, 
 \begin{cases}
 2^{\vert \ga \vert}   \Gamma(\vert \ga \vert) x^\ga & \text{ if } \ga \in (-\infty,0)\, ,\\
\frac 76  \log (1/x) & \text{ if } \ga=0\, , \\
q_1(\ga) x^\ga & \text{ if } \ga \in (0,2) \, ,\\
 \frac 12 x^2\log (1/x) & \text{ if } \ga=2\, , \\
  \frac{2^{\ga-4}\Gamma(\ga-2)}{(\ga-1)^2} x^{4-\ga} & \text{ if } \ga>2\, .
\end{cases}
\end{equation}
The final
 result, i.e. \eqref{eq:propasympt}, is recovered by dividing by $K_\ga(x)$ and using
$K_\ga(x)\sim x^{-\vert \ga\vert}\Gamma (\vert\ga\vert)/2^{1-\vert\ga\vert}$ for $\ga\neq 0$ and $K_0(x) \sim \log (1/x)$.
The constant
$C(\ga)$ in \eqref{eq:propasympt} is
\begin{equation}
\label{eq:C(alpha)}
C(\ga)\, :=\begin{cases}
2 & \text{ if } \ga <0\, ,
\\
7/6 & \text{ if } \ga =0\, ,
\\
q_1(\ga) 2^{1-\vert \ga\vert}/\Gamma(\vert \ga \vert) & \text{ if } \ga \in(0,2)\, ,
\\
1/4 &\text{ if } \ga =2\, ,
\\
1/\left( 8(\ga-1)^3 (\ga-2) \right) & \text{ if } \ga >2\, ,
\end{cases}
\end{equation}
with $q_1(\ga)$  given in  \eqref{eq:forq1}.
The proof of Proposition~\ref{prop:asympt} is therefore complete.
\qed

 \section{Lyapunov exponent and singularities: the proof of Proposition~\ref{th:asympt}}

In view of \eqref{eq:recover} we just consider $\ga \ge 0$.
We treat first the non integer case.

\subsubsection{The case $\ga \in (0, \infty) \setminus \bbN$}
By the \emph{connection} formula with the other modified Bessel function 
$I_\ga (x)$, 
we have \DLMF{10.27}{4} and \DLMF{10.25}{2}
\begin{equation}
\label{eq:connection}
K_\ga (x)\, =\, \frac \pi {2 \sin(\pi \ga)} \left(I_{-\ga}(x)-I_{\ga}(x)\right)\, ,
\end{equation}
with
\begin{equation}
\label{eq:connection-I}
I_\ga(x)\, :=\, \left( \frac x2\right)^\ga \sum_{k=0}^\infty
\frac{\left(  {x^2}/4\right)^k}{k! \, \Gamma(\ga+k+1)} \, =: \,
\left( \frac x2\right)^\ga \tilde I_\ga(x)\, ,
\end{equation}
where $\tilde I_\ga(x)$ is a non standard notation, but it singles out the analytic part of 
the $I_\ga(\cdot)$: in fact $\tilde I_\ga(\cdot)$ is an entire function. 
By elementary manipulations we  obtain
\begin{equation}
\frac{x K_{\ga -1}(x)}{K_\ga(x)}\, =\, 2\,  \frac{
(x/2)^{2\ga}\left( \tilde I_{\ga-1}(x)/ \tilde I_{-\ga}(x)\right)-
(x/2)^{2}\left( \tilde I_{-\ga+1}(x)/ \tilde I_{-\ga}(x)\right)
}{1-(x/2)^{2\ga}\left( \tilde I_{\ga}(x)/ \tilde I_{-\ga}(x)\right)}\, .
\end{equation}
Therefore, aiming at expanding this expression for $x \searrow 0$ up to the first singular term, we obtain
\begin{equation}
\label{eq:firstsing}
\begin{split}
\frac{x K_{\ga -1}(x)}{K_\ga(x)}\, &=\,
-2 \left(\frac x 2\right)^{2}
\left( \frac{\tilde I_{-\ga+1}(x)}{ \tilde I_{-\ga}(x)}\right)+
\left(\frac x 2\right)^{2\ga} \frac{2\Gamma(1-\ga)}{\Gamma(\ga) }
+O(x^{4\ga})+O\left(x^{2\ga +2}\right)
\\
&=\, p_{\ga, \lfloor \ga \rfloor}\left(x^2\right)+ \left(\frac x 2\right)^{2\ga} \frac{2\Gamma(1-\ga)}{\Gamma(\ga) }
+ O\left(x^{ 2\lfloor \ga \rfloor +2}\right)+O(x^{4\ga})\, ,
\end{split}
\end{equation}
where $p_{\ga, j}(y)$ is the Taylor expansion up to degree $j$ of 
$p_\ga(y):=- (y/2) \tilde I_{-\ga+1}(\sqrt{y})/ \tilde I_{-\ga}(\sqrt{y})$. 
It is not difficult to realize that 
the coefficients of this Taylor expansion are  just rational function of $\ga$.
Let us detail this point that is going to be important also for the passage to $\ga$ integer:
if we introduce for $k\in \bbN\cup  \{0\}$ the Pochhammer's symbol 
\begin{equation}
(\nu)_k\, :=\, \frac{\Gamma(\nu+k)}{\Gamma(\nu)}\stackrel{k=1,2, \ldots}{=} (\nu+k-1) (\nu+k-2) \cdots (\nu+1) \nu \, , 
\label{eq:Pochhammer}
\end{equation}
we see that $p_\ga(y)$ can be written in terms of Pochhammer's symbols:
\begin{equation}
\label{eq:withPoch}
p_\ga(y)\, =\, - \frac y2 \left({
\sum_{k=0}^\infty
\frac{y^k}{k! \,(-\ga+1)_{k+1}2^{2k}} 
}\right)\bigg/
\left({
\sum_{k=0}^\infty
\frac{ {y}^k}{k! \,(-\ga+1)_{k}2^{2k}} 
}\right)
.
\end{equation} 
From now the explicit determination of $p_\ga(y)$ is elementary, but cumbersome (to the point of requiring symbolic 
computations).
 We give the first four terms
\begin{equation}
\label{eq:fewterms}
\begin{split}
p_\ga(y)\, =\, &\sum_{j=1}^4 c_j(\ga)y^j + 
 \ldots =\, \frac{1}{2(\ga -1)} y
- \frac1{8(\ga-2)(\ga-1)^2}y^2 +\\
& \ \ \frac 1{16  (\ga-3)(\ga-2)(\ga-1)^3} y^3 
-\frac{5 \ga -11}{128 (\ga -4)(\ga-3)(\ga-2)(\ga-1)^4} y^4
+ \ldots 
\end{split}
\end{equation}
This completes the proof for the non integer case.

\subsubsection{The case $\ga=0,1, 2, \ldots$}
We need to treat separately the case $\ga=0$ because it involves $K_{-1}(x)$, which
however it is just $K_{1}(x)$, but it requires an ad hoc (much simpler) analysis. So we start off
with the case $\ga=1, 2, \ldots$. From now $\ga$ will be replaced by $n$ and when we write $\ga$ we mean a 
quantity that is not integer.
 By  \DLMF{10.31}{1} we have that for $n=0,1,2,\ldots$
 \begin{equation}
 \label{eq:Kn}
\left( \frac x2\right)^{n} K_n(x)\, =\, \frac 12  \sum_{k=0}^{n-1} (-1)^k
 \frac{(n-k-1)!}{k!}  \left( \frac x2\right)^{2k}+ \frac{(-1)^{n}}{n!}  \left( \frac x2\right)^{2n}\log (x) 
 +O\left(x^{2n}\right)\, ,
 \end{equation}
 and the sum as to be interpreted as empty if $n=0$.
 For $n=1, 2, \ldots$ we write the degree $2n-2$ polynomial in the right-hand side as a $n-1$ degree polynomial 
 with argument $x^2$:
 \begin{equation}
 q_{n}\left(x^2\right)\, :=\, 
\frac 12 \sum_{k=0}^{n-1} (-1)^k
 \frac{(n-k-1)!}{k!\, 2^{2k}}  \left( x^2\right)^{k} \,.
 \end{equation}
 With this notation we have for $n=1, 2, \ldots$
 \begin{equation}
 \begin{split}
 \frac{x K_{n-1}(x)}{K_n(x)}
 \,&=\, \frac{
 2 (x/2)^2 \frac{q_{n-1}(x^2)}{q_n(x^2)} + \frac{2(-1)^n}{(n-1)!} \frac{(x/2)^{2n}}{q_n(x^2)} \log x + O\left( x^{2n}\right)
 }
 {1+O\left( x^{2n} \vert \log x\vert \right)
 }
 \\& =\, \frac{x^2}2 t_n(x^2)+
 \frac{2^{2-2n}(-1)^n}{((n-1)!)^2}x^{2n} \log x + O\left( x^{2n}\right)\, ,
 \end{split}
 \end{equation}
 where $t_1(y):=0$ and, for $n=2, 3, \ldots$, $t_n (y)$ is the polynomial of degree $n-2$ 
 given by the Taylor expansion of the rational function $q_{n-1}(y)/q_{n}(y)$.  
 
 We are therefore left with showing that the coefficients of the polynomial ${y} t_n(y)/2$ of degree $n-1$ coincide with the corresponding coefficients of the Taylor polynomial of
 $p_\ga(y)$, when $\ga=n$ (we can consider $n=2,3,\ldots$ becaue if $n=1$ the polynomial is identically zero and our claim  is trivially verified).
 In different terms, we have to show that 
 if we set $\ga=n=2,3 , \ldots$
 in \eqref{eq:fewterms} up to the degree $n-1$ (it is readily seen that
 the $n$-th Taylor coefficients diverges as $\ga \to n$), then 
 we obtain ${y} t_n(y)/2$. To prove this
 it is useful to remark that, in order to obtain the Taylor expansion of $p_\ga(y)$
 up to order $\lfloor \ga \rfloor$ it is sufficient to expand the rational function
 \begin{equation}
\label{eq:withPochtilde}
\tilde p_\ga(y)\, =\, - \frac y2 \left({
\sum_{k=0}^{\lfloor \ga \rfloor-1}
\frac{y^k}{k! \,(-\ga+1)_{k+1}2^{2k}} 
}\right)\bigg/
\left({
\sum_{k=0}^{\lfloor \ga \rfloor-1}
\frac{ {y}^k}{k! \,(-\ga+1)_{k}2^{2k}} 
}\right)
.
\end{equation} 
In this expression we can set  $\ga = n$ obtaining thus the rational function
\begin{equation}
\label{eq:withPoch-n}
\tilde p_n(y)\, :=\, - \frac y2 \left({
\sum_{k=0}^{n-2}
\frac{y^k}{k! \,(-n+1)_{k+1}2^{2k}} 
}\right)\bigg/
\left({
\sum_{k=0}^{n-2}
\frac{ {y}^k}{k! \,(-n+1)_{k}2^{2k}} 
}\right)
.
\end{equation} 
and Taylor coefficients up to degree $n-1$ of this function are precisely the limit for $\ga\to n$
of the Taylor coefficients up to degree $n-1$ of $p_\ga(y)$, cf. \eqref{eq:withPochtilde}-\eqref{eq:fewterms}.
We are left with showing that the coefficients up to degree $n-1$ of $\tilde p_n(y)$
coincide with the corresponding coefficients of ${y} t_n(y)/2$. This is equivalent to showing that 
\begin{equation}
-
\left({
\sum_{k=0}^{n-2}
\frac{y^k}{k! \,(-n+1)_{k+1}2^{2k}} 
}\right)\bigg/
\left({
\sum_{k=0}^{n-2}
\frac{ {y}^k}{k! \,(-n+1)_{k}2^{2k}} 
}\right)\, =\, \frac{q_{n-1}(y)}{q_{n}(y)}
+ O \left( y^{n-1}\right)
\end{equation}
and this is implied by the stronger (non asymptotic) condition
\begin{multline}
-
\left({
\sum_{k=0}^{n-2}
\frac{y^k}{k! \,(-n+1)_{k+1}2^{2k}} 
}\right)\bigg/
\left({
\sum_{k=0}^{n-2}
\frac{ {y}^k}{k! \,(-n+1)_{k}2^{2k}} 
}\right)\, =\\
\left(
{\sum_{k=0}^{n-2} (-1)^k
\frac{(n-k-2)!}{k! \, 2^{2k}}  y^{k} }\right)
 \bigg/
 \left(
\sum_{k=0}^{n-2} (-1)^k
\frac{(n-k-1)!}{k! \, 2^{2k}}  y^{k} 
 \right)
 \, , 
\end{multline}
For $n=2, 3, \ldots$ this identity is verified directly by using that for $n=1,2, \ldots$ and $k\in 0,1, \ldots$
\begin{equation}
(-n)_k\,=\, (-1)^k \frac{n!}{(n-k)!}\, ,
\end{equation}
and this completes the proof in the case of $\ga=1,2, \ldots$.

We are left with the case $\ga=0$:
\begin{equation}
\frac{x K_{-1}(x)}{K_0(x)}\, =\, 
\frac{x K_{1}(x)}{K_0(x)}\stackrel{\eqref{eq:Kn}}{\sim} \frac 1{\log (1/x)}\,,
\end{equation}
but this is easily improved by going back to \DLMF{10.31}{1} and using
\begin{equation}
 \label{eq:K0}
 K_0(x)\, =\, -\log (x/2) - \gamma + O\left(x^2\right)\, ,
 \end{equation}
 where $\gamma$ is the Euler-Mascheroni constant. This, with \eqref{eq:Kn} for $n=1$, implies that
$x K_{-1}(x)/{K_0(x)}$ is equal to $1/(\log(1/x) +(\log 2- \gamma))+ O(x^2)$.
This completes the proof of Proposition~\ref{th:asympt}.
\qed

 \section{Lyapunov exponent and singularities: the proof of Theorem~\ref{th:MW}}
 \label{sec:MWproof}
 
 The proof of Theorem~\ref{th:MW} is somewhat involved, since the ratio of Bessel functions in the integrand becomes quite singular at one end of the domain of integration. 
 For every fixed $x>0$ the numerator and denominator are entire functions of $\ga$, so the ratio is analytic apart from the zeros of the denominator; 
 these are all on the imaginary axis and they are bounded 
 away from the real axis as long as $x$ is bounded away from zero (\cite[Appendix~A]{Friedlander}, Table~\ref{table:1}).
 But integrating over $x$ down to zero adds contributions that are
 less and less regular as the gap between the origin and the poles of the integrand 
 shrinks when $x$ becomes small. 
  
 \begin{table}[h!]
\centering
  \begin{tabular}{ | c | l | c |  c | c | c |  c | c | c |  }
    \hline
    $x$ & $n$ & $\nu_1$  & $\nu_2$ & $\nu_3$ &$\nu_4$ &$\nu_5$ &$\nu_6$ & $\nu_7$ \\ \hline \hline
    1 & 1 & 2.96 & & & & & & \\ \hline
    1/10 & 3 & 1.14 & 2.04& 2.85 & & & & \\  \hline
    1/100 & 5 & 0.64 & 1.23 & 1.78 & 2.30 & 2.81 & & \\  \hline
    1/1000 & 7 & 0.44 & 0.87& 1.27 & 1.66 & 2.04 & 2.42 & 2.78\\  \hline
  \end{tabular}
\caption{ The zeros of $K_\ga(x)$ are all for $\ga=i\nu$ with $\nu\in \bbR$. Since $K_{i\nu}(x)=K_{-i\nu}(x)$,
we put in the table the set  $\{ \nu \in [0,3]:\, K_{i \nu}(x)=0\}=\{\nu_1, \ldots, \nu_n\}$,
with $\nu_j=\nu_j(x)$ and $n=n(x)$, for four values of $x$. The numerical values are rounded to the closest decimal. }
\label{table:1}
\end{table}
 
One way to obtain a proof is to exploit  
once again  the
  connection formula \eqref{eq:connection}-\eqref{eq:connection-I} which gives an expansion 
 of both numerator and denominator:  but controlling the ratio is of course not straightforward. 
 And in fact  McCoy and Wu approach the problem this way, but keeping only the 
\emph{leading terms} of the series in the connection formula. 
The validity of this procedure is not obvious, since {\it a priori} the resulting correction could be less regular than the leading terms, 
but it is nonetheless helpful to begin by examining this simplified problem
that has the nice feature of leading to a solution in terms of special functions, since we shall see that it correctly illustrates the main features of the proof. 
With this aim in mind we examine
a simplified McCoy-Wu formula -- this corresponds to studying a function
$\tilde \tf$ which is defined, like $\tf$, as the integral over $x\in (0, \eta)$ of a suitable function
$\tilde f_x(\ga)$ (see \eqref{eq:tildetf}) -- and

\begin{enumerate}
\item we perform the integration explicitly and discuss the regularity  of $\ga \mapsto \tilde \tf (\ga)$; 
\item  we then argue how understanding the location of (some of) the poles
of $\tilde f_x$ in the complex plane, and the corresponding residues, gives another way to understand the regularity.
\end{enumerate}  
 
 All of this is done in Section~\ref{sec:heuristics}. Then 
in Section~\ref{sec:MW_proof} we give   the proof of Theorem~\ref{th:MW}, based on a treatment of the poles of $f_x$.
 
\subsection{Heuristic arguments and idea of the proof}
\label{sec:heuristics}
\subsubsection{The simplified McCoy and Wu problem: exact solution}
Much like McCoy and Wu did in  \cite[p.~642]{MW1}, we can consider the  
leading contribution to the integrand in \eqref{eq:toyMW} for
$x\searrow 0$ and $\ga\in \bbC$ tending also to zero. 
This step is an uncontrolled approximation that is obtained by keeping 
\emph{the first terms} in \eqref{eq:connection}-\eqref{eq:connection-I}
and by using $\Gamma(\ga) \sim 1/ \ga$ 
\begin{equation}
\label{eq:gblmr3}
K_{\ga - 1}(x) \, \simeq \, \frac 12 \left( \frac x 2 \right)^{\ga -1}\, ,
\end{equation}
and 
\begin{equation}
K_\ga(x) \, \simeq \, \frac 1{2\ga} \left( \left( \frac x 2 \right)^{- \ga}- \left( \frac x 2 \right)^{\ga} \right)\, ,
\end{equation}
 so
 \begin{equation}
 \label{eq:truecaseleft}
 \frac{x K_{\ga-1}(x) }{K_\ga(x) } \simeq \frac{2 \ga}{ \left( \frac 2 x \right)^{2 \ga}-1} \, =\, \frac{2 \ga}{\exp(2 \ga L(x))-1} \, =: \, \tilde f_x(\alpha) \, ,
 \end{equation}
 with $L(x):= \log(2/x)$. 
 We then have
 \begin{equation}
 \label{eq:tildetf}
 \tilde{\tf} (\ga)\, :=\,\int_0^\eta  \frac{2 \ga}{\exp(2 \ga L(x))-1} \dd x \,.
 \end{equation}
 for $\eta \in (0,2)$, cf.\ \cite[(4.44)]{MW1}.
 From now onward the analysis is rigorous.
 
$\tilde \tf(\ga)$ can then be made explicit  up to an additive contribution that is analytic near the real axis: in fact
if we set
\begin{equation}
\widecheck \tf(\ga)\, :=\, \int_0^2  \left(\frac{2 \ga}{\exp(2 \ga L(x))-1} -\frac 1{L(x)}\right)\dd x\, ,
\end{equation}
we have
\begin{equation}
\tilde \tf (\ga)\, =\, \widecheck \tf(\ga)
 -
  \int_\eta^2  \left(\frac{2 \ga}{\exp(2 \ga L(x))-1} -\frac 1{L(x)}\right)\dd x +
\int_0^\eta \frac {1}{L(x)} \dd x \, .
\end{equation}
But $\int_0^\eta (1/L(x)) \dd x
= \Gamma(0, \log(2/\eta))$ is just a constant
and the second addend in the right-hand side is (real) analytic in $\ga$. In fact
 the integrand in  is meromorphic with poles on the imaginary axis, precisely for $\ga$ equal to any integer multiple of $\pm  \pi /L(x)$. Therefore the second addend in the right-hand side is  analytic for $\ga$
in $\bbC \setminus \{iy:\, \vert y\vert \ge \pi /L(\eta)\}$. We can therefore focus on  
$ \widecheck\tf (\ga)$ and we start by observing that for $\ga \in \bbR$ 
\begin{equation}
\label{eq:symF}
 \widecheck\tf (-\ga)\, =\,     \widecheck\tf (\ga)+ 4\ga\,,
 \end{equation}
 as a result of a straightforward manipulation using 
 $1/(1-e^{-2\ga L})= 1/(e^{2\ga L} -1)+1$.
 
 \medskip
 
 \begin{rem}
 \label{rem:sym}
 We note that the same argument can be applied directly to $\tilde \tf_1 (\ga)$ obtaining
 \begin{equation}
\label{eq:symF2}
 \tilde\tf (-\ga)\, =\,     \tilde\tf (\ga)+ 2 \eta \ga\,.
 \end{equation}
 But in fact we have also 
 \begin{equation}
\label{eq:symF3}
 \tf (-\ga)\, =\,     \tf (\ga)+ 2 \eta \ga\,, 
 \end{equation}
 which follows by applying  \eqref{eq:Bessel-id} and  $K_\gb(x)=K_{-\gb}(x)$. 
 In other words all these functions -- $\widecheck \tf$, $\tilde \tf$ and $\tf$  -- are even, up to a linear term.
 \end{rem}
 
 \medskip

 Thanks to \eqref{eq:symF} we  can focus on the case  $\Re \ga>0$
 and compute (change the variable and move the contour of integration in the complex plane)
\begin{multline}
\label{eq:habv}
\widecheck\tf (\ga)\, =\, 2 \int_0^\infty \left(\frac1{e^v -1}- \frac 1v
 \right) e^{-v/(2\ga)}\dd v \,=\, 
 -2 \int_0^\infty \left( \frac 1v - \frac1{1-e^{-v} }
 \right) e^{-v/(2\ga)}\dd v -4 \ga 
 \\
=\, -4 \ga -2 \log (2 \ga) -2 \psi(1/(2\ga))\,
\stackrel{\alpha \searrow 0}\sim \, -2 \ga - \sum_{j=1}^\infty \frac{B_{2j}}j (2 \ga)^{2j}\,  , 
\end{multline}
where $\psi(z)= \Gamma'(z)/\Gamma(z)$ and the notion of $\sim$ is extended here
and it has to be interpreted in the sense of asymptotic series (i.e., that the difference of left-hand side and of the series in  the right-hand side
truncated to $j=n$ is $o(\ga^{2n})$)
: in the third step we have applied \DLMF{5.9}{13}  
and the asymptotic relation  \DLMF{5.11}{2} 
$\psi(z)\stackrel{z \to +\infty}\sim \log z -1/(2z) + \sum_{j=1}^\infty (B_{2j}/(2j))z^{2j}$: the rational numbers $B_{2n}$ are the Bernoulli numbers \DLMFs{24.2}. By \DLMF{24.9}{8}	
\begin{equation}
\label{eq:Btof}
B_{2n} \stackrel{n \to \infty}\sim 2 (-1)^{n+1} \frac{(2n)!}{(2\pi)^{2n}}\, ,
\end{equation}
so the series has radius of convergence zero. Note that
 \eqref{eq:symF} readily implies that
$\widecheck\tf (\ga)\sim -2 \ga - \sum_{j=1}^\infty (B_{2j}/j) (2 \ga)^{2j}$ 
holds also for $\ga \nearrow 0$ and not only for $\ga \searrow 0$. 
 Let us reorder what we have done (and more) into a statement:

\medskip

\begin{proposition}
\label{th:MWsimpler}
$\widecheck \tf(\cdot)$ is defined and analytic in the complex plane without the imaginary axis
and it can be extended by continuity on the whole real axis  by setting $\widecheck \tf(0)=0$. Then
$\widecheck \tf$ restricted to $\bbR$ is $C^\infty$ in the origin. On the other hand, $\widecheck \tf$ cannot be 
continued as an analytic function at any point on the imaginary axis.
\end{proposition}

\medskip

\noindent
\emph{Proof.}
Let us first show that
 $\widecheck \tf$ is $C^\infty$ at the origin. For this it is more practical  (and equivalent, since $\tilde \tf - 
 \widecheck \tf$  is real analytic) to go back to 
 $\tilde \tf$ and observe that with $y=2\ga \in \bbR \setminus\{0\}$
 and the change of variable $v=\log (2/x)$ we obtain
 \begin{equation}
 f(y)\,:=\, \tilde \tf (y/2)\,=\,   2\int_{c}^\infty q(yv) \frac{\exp(-v)}v \dd v\, ,
 \end{equation} 
 with  $c=\log (2/\eta)\in (0, \infty)$  and 
 $q(u) = \frac{u}{\exp(u)-1}$.
 It is straightforward to verify that $q(u)\in (0,1]$ for $u\ge 0$ and that $q(u) \in [1, u+1]$ for $u\le 0$
 and dominated convergence implies that $f$ is $C^0$ (on the whole $\bbR$, but of course our attention 
 is at the origin, outside we already know that $f$ is  real analytic).
 To show differentiability it suffices to observe that formally 
  \begin{equation}
  \label{eq:formal}
 f^{(n)}(y)\,=\,    \int_{c}^\infty q^{(n)}(yv) v^{n-1}{\exp(-v)} \dd v\, ,
 \end{equation} 
 where  $f^{(n)}(y)= (\dd /\dd y)^n f(y)$. But  it is straightforward to verify that $q^{(1)}(u)$ is monotonically decreasing from 
 $1$ to zero and $q^{(n)}(u)$, $n=2$ or larger, vanishes at $\pm \infty$. Hence $\sup _{u \in \bbR}\vert q^{(n)}(u)\vert
 < \infty$ and this implies that 
 \eqref {eq:formal} is not just formal: $f \in C^\infty$ and its derivatives are given by \eqref {eq:formal}.

 On the other hand, $\widecheck \tf$ is  not analytic in zero since the radius of convergence of the series is zero: 
we record that if the Taylor series for $\widecheck\tf (\ga)$ is
$\sum_{n} c_n \ga ^n$ then  
\begin{equation}
\label{eq:tcwpe}
c_{2n}\, \stackrel{n \to \infty}\sim\, 4 (-1)^{n+1} \frac{(2n-1)!}{(\pi)^{2n}}\, ,
\end{equation}
while $c_{2n+1}=0$ for $n=1,2, \ldots$: this follows directly from \eqref{eq:habv} and \eqref{eq:Btof}.

To conclude, we argue that $\widecheck \tf$ cannot be extended as an analytic function through the imaginary axis. 
Call $\tf^+: \{z \in \bbC: \, \Re( z) >0\} \longrightarrow \bbC$ the function that coincides  with $\widecheck \tf$ in its domain of definition (the right half-plane): keep in mind that $\widecheck \tf$ is defined also in the half plave with negative real part and it is analytic there.  
Suppose now that there exist $a\in \bbR$ and $\epsilon >0$ such that we can define $\widecheck \tf( \ga)$ for $\ga \in \{it: \, \vert t-a \vert <\epsilon \}$, in such a way that $\widecheck \tf$ is analytic in a neighborhood of $ia$. We can and do assume that $a>0$ as well as $a-\gep>0$ by symmetry and because we already know that $\widecheck \tf$
is not analytic at the origin. 
This extends  analytically $\tf^+$ too. 
But \eqref{eq:habv} yields $\tf^+(\ga) = -4 \ga - \log 4-2 \log  \ga -2 \psi(1/(2\ga))$ in the right-half plane and this expression, containing 
the function  $\psi$ (that is meromorphic on the whole $\bbC$ and has poles on the negative real axis \DLMFs{5.2}) and a logarithm which is defined and analytic 
on the whole complex plane except for a cut (that can be for example chosen to be $\{z\in \bbC: \Im z =\Re z \le 0\}$).
But $\tf^+$ must coincide with $\widecheck \tf$ in the whole region where $\tf^+$ is extended: this region includes the
negative semi-axis where  $\tf^+$ has poles and $\widecheck \tf$ is analytic. This is not possible,
 hence $\widecheck \tf$ cannot be extended 
analytically at any point on the imaginary axis.
\qed

\subsubsection{A different viewpoint on the McCoy and Wu simplified problem}
We restart from \eqref{eq:tildetf}, but we take a different approach: we avoid exact integration. 
As already noted,
the integrand in that expression for $\tilde \tf (\ga)$, for fixed $x>0$,
is a meromorphic function of $\ga$. The poles 
are on the imaginary axis and $0$ is not a pole because the singularity is removable:
the poles are $n \pi  i / L(x)$ for $n \in \bbZ \setminus \{0\}$, where recall that we write $L(x) := \log (2/x)$ for brevity.
The  integrand in the expression for $ \tf_1 (\ga)$ is therefore analytic in a ball of radius   $\pi / L(x)$ around zero
and the most singular part in the residue expression for the integrand comes from the two closest poles, which are 
$\pm \pi i/ L(x)$. If we compute the residues of these two poles for the integrand  we find $\pm 2\pi i /(L(x))^2$ and therefore
the contribution to the integrand of these two poles is 
\begin{equation}
\label{eq:pgax}
\begin{split}
P_\ga(x)\, :=\, 
\frac{2\pi i}{L(x)^2 \left(2 \ga- \frac{2\pi i}{L(x)}\right)}- \frac{2\pi i}{L(x)^2 \left(2 \ga+ \frac{2\pi i}{L(x)}\right)}
\, &=\,
- \frac{2}{L(x)}\frac 1{\left(\frac{\ga L(x)}{\pi}\right)^2+1}
\\
& =\, -2 \sum_{n=0}^\infty (-1)^n \ga^{2n} L(x)^{2n-1} \pi^{-2n}\, ,
\end{split}
\end{equation}
where the last equality holds only for $\vert \ga L(x)/ \pi \vert < 1$, but it is in any case useful to identify 
all the derivatives of $P_\ga (x)$ at $\ga=0$. It seems reasonable to believe that the singularity at the origin
of $\tilde \tf$ is induced by the poles of the integrand and that the two poles that are closest to the origin
give the leading part of the singularity. If this is the case $\ga \mapsto \int_0^\eta p_\ga(x) \dd x$ should capture
the leading behavior of the singularity of $\tilde \tf$. This is confirmed or at least highly suggested 
by observing that
\begin{equation}
\label{eq:abvhr}
\int_0^\eta L(x)^{2n-1} \dd x \, =\, 2 \int_{\log (2/ \eta)}^\infty y^{2n-1} e^{-y} \dd y \, =\, 2 \Gamma(2n, \log (2/\eta))
\stackrel{n \to \infty}\sim  2 \Gamma(2n)\, =\, 2 (2n-1)!\, ,
\end{equation}
so if we proceed at a completely 
 formal level, by  integrating term by term the series in the second line of   \eqref{eq:pgax} and using \eqref{eq:abvhr},
 we directly recover 
 \eqref{eq:tcwpe}!

\subsubsection{Approaching the true problem}

Going back to the true problem, that is $\tf$, for which the integrand is  the left-had side of \eqref{eq:truecaseleft}, we have to identify the zeros 
of $K_\ga(x)$: this problem has been studied and, as possibly expected, the heuristics coming from studying the poles 
of  the right-hand side of \eqref{eq:truecaseleft} is qualitatively correct, so the poles we have to study, or the zeros of  
$K_\ga(x)$, are on the imaginary axis. Moreover they accumulate on the origin when $x\searrow 0$. In fact, to leading order 
(for $n$ fixed and $x \searrow 0$) they are still in $\pm n \pi i  /\log (2/x)$. But, more precisely,
they are in $\pm n \pi i  /(\log (2/x)- \gamma+ o(1))$ as $x \searrow 0$ ($\gamma$ the Euler constant, see below). 

\smallskip

\begin{rem}
\label{rem:shift}
The subleading correction $-\gamma$ to $\log (2/x)$ for the location of the zeros can be inserted in the 
heuristic argument that we presented just by being keeping one more term in the expansion of the $\Gamma$ function in the first steps 
\eqref{eq:gblmr3}-\eqref{eq:truecaseleft}.
The change amounts simply 
 to work with $L(x)= \log (2e^{-\gamma}/ x)$, but this small offset in $L(x)$ leads to the multiplicative 
$e^{-\gamma}$ constant in the final asymptotic results, see Theorem~\ref{th:MW}. 
\end{rem}
\smallskip

Apart for the quantitative issue of Remark~\ref{rem:shift}, the proof
of Theorem~\ref{th:MW} requires taking care of two main issues:

\smallskip

\begin{enumerate}
\item 
Even admitting that the two closest poles give the main contribution, we have to set-up a rigorous procedure
corresponding to the formal argument that we developed using \eqref{eq:pgax} and \eqref{eq:abvhr}.
This procedure requires  controlling not only location of the two poles, but also the residues that
this time are given by ratio of series coming from \eqref{eq:K_expansion1}.
\item One needs to control the effect of the poles that are farther from the origin: note that these terms will in any case give contributions
that  generate divergent series, but this time the distance of the poles from the origin is at least (about) twice the distance of the two $n=1$
poles. Hence they will contain an exponential factor that is at least twice smaller. 
\end{enumerate}

\subsection{The proof} \label{sec:MW_proof}

The  proof of Theorem~\ref{th:MW}, that starts here, follows a main line separated by a number of lemmas and corollaries.   
As a preliminary, we set out precise versions of some of the statements made above about the regularity of 
\begin{equation}
	f_x: \alpha \mapsto x\frac{K_{1-\alpha}(x)}{K_\alpha(x)},
\end{equation}
for $x > 0$.
$K_\alpha(x)$ is entire as a function of $\alpha$ for all $x \neq 0$ \cite[section~10.25(ii)]{DLMF},
so $f_x$ is a ratio of two entire functions, and thus it is analytic except at the zeros of $\alpha \mapsto K_\alpha (x)$.  For $x>0$, these zeros are all pure imaginary and located in the region $\left|\alpha\right| > x$~\cite[Appendix~A]{Friedlander}.
Note that
this characterization
is already sufficient to show that $\int_X^\eta f_x (\alpha) d x$ is analytic on $\bbC \setminus \pm i [X,\infty)$ for any $0 <X < \eta < \infty$, so it suffices to prove the statements in Theorem~\ref{th:MW} with $\eta$ replaced by some sufficiently small $X$ in the integral defining $\tf$. 
Furthermore, $f_x$ is infinitely differentiable on the real axis for any $x > 0$; in the next few lemmata we will bound its derivatives in such a way as to show that this is also true of the integral $\tf$.

To begin with, note that combining \eqref{eq:connection} and \eqref{eq:connection-I} gives
\begin{equation}
	K_\alpha(x)
	=
	\frac{\pi}{2 \sin \pi \alpha} \sum_{k=0}^\infty \frac{\left( x^2 / 4 \right)^k}{k!}\left[ 
		\frac{\left( x/2 \right)^{-\alpha}}{\Gamma(-\alpha+k+1)} 
		- \frac{\left( x/2 \right)^\alpha}{\Gamma(\alpha+k+1)} 
	\right],
	\label{eq:K_expansion1}
\end{equation}
for $\alpha \in \bbC \setminus \bbZ$ and $x > 0$.

\begin{lemma}
	There exist $A,C_1,X > 0$ such that for all $x \in [0,X]$, $\alpha \in \bbC$ such that $\left|\alpha\right|\le \tfrac12$ and $\left|\Im \alpha\right| \log 1/x \le A$,
	\begin{equation}
		\left|\frac{x^\alpha}{\Gamma(1+\alpha)} - \frac{x^{-\alpha}}{\Gamma(1-\alpha)}\right|
		\ge C_1 \left|\alpha\right|x^{-|\Re \alpha|}.
	\end{equation}
	\label{lem:leading_diff}
\end{lemma}

\begin{proof}
	Letting $g_x(\alpha) := x^\alpha/\Gamma(1+\alpha)$,
	\begin{equation}
		\frac{x^\alpha}{\Gamma(1+\alpha)} - \frac{x^{-\alpha}}{\Gamma(1-\alpha)}
		=
		g_x(\alpha)-g_x(-\alpha)
		=
		\alpha
		\int_{-1}^1 g'_x(y\alpha) \dd y \, ,
	\end{equation}
	and so
	\begin{equation}
		\left|\frac{x^\alpha}{\Gamma(1+\alpha)} - \frac{x^{-\alpha}}{\Gamma(1-\alpha)}\right|
		\ge 
		|\alpha|
		\left|
		\int_{-1}^1 \Re g'_x(y\alpha) \dd y
		\right|\, 
		.
		\label{eq:leading_real_integral}
	\end{equation}
	Note that
	\begin{equation}
		g'_x(\alpha)
		=
		\left[ \log x - \psi(1+\alpha) \right]
		\frac{x^\alpha}{\Gamma(1+\alpha)}
		,
	\end{equation}
	where $\psi(z) := \Gamma'(z)/\Gamma(z)$ is the Psi function \DLMFs{5.2}.
	$\psi(1+\alpha)/\Gamma(1+\alpha)$ is analytic for $|\alpha| < 1$, so choosing $X$ small enough we can obtain $\left| \psi(1+\alpha)/\Gamma(1+\alpha) \right| \le (1-c_1) $ for any fixed $c_1 \in (0,1)$, whence
	\begin{equation}
		\left|\Re g'_x(\alpha) - \Re \frac{x^\alpha \log x }{\Gamma(1+\alpha)}\right|
		\le
		\left|g'_x(\alpha) - \frac{x^\alpha \log x }{\Gamma(1+\alpha)}\right|
		=
		\left| \psi(1+\alpha) \frac{x^\alpha}{\Gamma(1+\alpha)}\right|
		\le
		(1-c_1) 
		x^{\Re \alpha} \log 1/x
		,
	\end{equation}
	and noting that $1/\Gamma(1+\alpha)$ is entire and takes positive real values for $\alpha \in [-1/2,2/2]$, by choosing $A$ small enough we obtain
	\begin{multline}
		- \Re \frac{ x^\alpha \log x}{\Gamma(1+\alpha)}
		=
		x^{\Re \alpha} \log 1/x \left[ \cos \left( \Im \alpha \log x \right) \Re \frac{1}{\Gamma(1+\alpha)} 
			- \sin \left( \Im \alpha \log x \right) \Im \frac{1}{\Gamma(1+\alpha)} \right]
			\\ \ge c_2 x^{\Re \alpha} \log 1/x \ge 0\, ,
	\end{multline}
	for some $c_2>0$.
	Combining the above observations, we see that
	\begin{equation}
		\left|
		\int_{-1}^1 \Re g'_x(y\alpha) \dd y
		\right|
		\ge 
		c_1 c_2
		\log 1/x
		\int_{-1}^1 x^{y \Re \alpha } \dd y
		=
		c_1 c_2
		\frac{x^{-\Re \alpha} - x^{\Re \alpha}}{\Re \alpha}.
		\label{eq:leading1}
	\end{equation}
	The last expression is an even function of $\Re \alpha$, so without loss of generality we can consider $\Re \alpha = \rho \in [0,1/2]$.  Noting that 
	$\rho \mapsto 1 - x^{2\rho} = 1- \exp\left(- 2 \rho \log x \right)$ is concave for all $x > 0$ and checking the boundary cases, it is easy to see that $1 - x^{2\rho} \ge \rho$ for all $x \in (0,1/2)$, $\rho \in [0,1/2]$.
	Using this together with \eqref{eq:leading_real_integral}  and \eqref{eq:leading1}, we obtain
	\begin{equation}
		\left|\frac{x^\alpha}{\Gamma(1+\alpha)} - \frac{x^{-\alpha}}{\Gamma(1-\alpha)}\right|
		\ge 
		c_1 c_2
		|\alpha| x^{-|\Re \alpha|}.
	\end{equation}
This completes the proof of Lemma~\ref{lem:leading_diff}
\end{proof}

\begin{lemma}
	There exist some $X \in (0,1)$, $A > 0$, and $C_2 > 0$ such that 
	\begin{equation}
	|f_\alpha (x)|\, \le\, C_2\, ,
	\end{equation}
for all $x \in (0,X]$, $|\alpha| \le 1/2$, and $|\Im \alpha| \le A/\log(2/x)$.
	\label{lem:wide_integrand_bound}
\end{lemma}
\begin{proof}
	Rearranging \eqref{eq:K_expansion1}, we have
	\begin{equation}
		\begin{split}
			- \frac{2 \sin \pi\alpha}{\pi}K_{1-\alpha}(x)
			&=
			\frac{(x/2)^{\alpha-1}}{\Gamma(\alpha)}
			+
			\sum_{k=1}^\infty \frac{\left( x/2\right)^{\alpha+2k-1}}{k!\Gamma(\alpha+k)}
			-
			\sum_{k=1}^\infty \frac{k}{k-\alpha} \frac{\left( x/2\right)^{-\alpha+2k-1}}{k!\Gamma(k-\alpha)}
			\\
			& =
			\frac{(x/2)^{\alpha-1}}{\Gamma(\alpha)}
			+
			\sum_{k=1}^\infty \frac{(x/2)^{2k-1}}{k!} \left[\frac{(x/2)^{\alpha}}{\Gamma(k+\alpha)} - \frac{k}{k-\alpha}\frac{(x/2)^{-\alpha}}{\Gamma(k-\alpha)}  \right],
		\end{split}
		\label{eq:widelem_numerator1}
	\end{equation}
	Noting that the final sum vanishes term by term when $\alpha=0$ and that
	\begin{equation}
		\frac{\partial}{\partial\alpha}\left[ \frac{(x/2)^{\alpha}}{\Gamma(k+\alpha)}\right]
		=
		\frac{(x/2)^{\alpha}}{\Gamma(k+\alpha)}
		\left[ \log (x/2) - \psi(k+\alpha) \right]\, ,
	\end{equation}
	and
	\begin{equation}
		\frac{\partial}{\partial\alpha}\left[ \frac{(x/2)^{-\alpha}}{\Gamma(k-\alpha)} \frac{k}{k-\alpha}\right]
		=
		 \frac{\partial}{\partial\alpha}\left[ \frac{k (x/2)^{-\alpha}}{\Gamma(1+k-\alpha)} \right]
		=-
		\frac{k (\frac x2)^{-\alpha}}{\Gamma(1+k-\alpha)}
		\left[ \log \left(\frac x2\right) + \psi(1+k-\alpha) \right],
	\end{equation}
	and that $\left|\psi(k+\alpha)\right|\vee \left|\psi(1+k-\alpha)\right| \le \const \log (k+1)$ for $k \ge 1$ and $\left|\alpha\right| \le 1/2$, we then have also
	\begin{multline}
		\left| \frac{(x/2)^{\alpha}}{\Gamma(k+\alpha)} - \frac{k}{k-\alpha}\frac{(x/2)^{-\alpha}}{\Gamma(k-\alpha)}\right| \, 
		\le\\
		\const \left|\alpha\right| \left[ \left|\frac{(x/2)^{-\alpha}}{\Gamma(k+\alpha)}\right| \vee \left|\frac{k (x/2)^{\alpha}}{\Gamma(k-\alpha)}
\right| \right]
		\left[ \log \left(\frac x2\right) + \log (k+1) \right].
	\end{multline}
	Noting that $\Gamma(1+k \pm \alpha) = \Gamma(1 \pm \alpha) (1 \pm \alpha)_k$ (see \eqref{eq:Pochhammer}),
	\begin{equation}
		\left|(1 \pm \alpha)_k\right|
		=
		\left|1 \pm \alpha\right| \cdots \left|k \pm \alpha\right|
		\ge
		\left( \frac{1}{2} \right) \cdots \left( \frac{2k-1}{2} \right)
		= 2^{-k} \; (2k-1)!!,
	\end{equation}
	and that $\left|\Gamma(1 \pm \alpha)\right|$ is bounded for $\left|\alpha\right|\le 1/2$, we have 
	\begin{equation}
		\begin{split}
			& \left| \sum_{k=1}^\infty \frac{(x/2)^{2k-1}}{k!} \left[\frac{(x/2)^{\alpha}}{\Gamma(k+\alpha)} - \frac{k}{k-\alpha}\frac{(x/2)^{-\alpha}}{\Gamma(k-\alpha)}  \right] \right|
			\\ & \quad \le 
			\const |\alpha| \sum_{k=1}^{\infty} \frac{(x/2)^{2k-1}}{k!}\left[ \log(k+1) + \log (2/x) \right] \left[ \left| \frac{(x/2)^{-\alpha}}{\Gamma(k+\alpha)}\right|\vee \left|  \frac{k (x/2)^\alpha}{\Gamma(k-\alpha)} \right| \right]
			\\ & \quad \le 
			\const |\alpha| (x/2)^{1-|\Re \alpha|} \left[ \sum_{k=0}^\infty \frac{2^{k+1} \log (k+2)}{k!(2k+1)!!}(x/2)^{2k} 
			+ \log(2/x) \sum_{k=0}^\infty \frac{2^{k+1}(x/2)^{2k}}{k!(2k+1)!!} \right]
			\\ & \quad \le 
			\const |\alpha| \left[ 1 + \log (2/x) \right] (x/2)^{1-|\Re \alpha|},
		\end{split}
	\end{equation}
	where the last bound follows from the observation that each of the sums in the preceeding expression is $\alpha$-independent and defines an entire function of $x$, and is therefore bounded on any compact interval.  Combining this with \eqref{eq:widelem_numerator1} and noting that
	\begin{equation}
		\left| \frac{(x/2)^{\alpha-1}}{\Gamma(\alpha)} \right|
		\le
		\const |\alpha| (x/2)^{\Re \alpha -1},
	\end{equation}
	we have
	\begin{equation}
		\left| \frac{2 \sin \pi\alpha}{\pi}K_{1-\alpha}(x) \right|
		\le
		\const |\alpha| (x/2)^{\Re \alpha -1}.
		\label{eq:widelem_numerator_bound}
	\end{equation}

	As for the denominator, using \eqref{eq:K_expansion1}
	\begin{equation}
		\begin{split}
			\left|\frac{2 \sin \pi\alpha}{\pi}K_\alpha(x)\right|
			\ge &
			\left| \frac{\left( x/2 \right)^\alpha}{\Gamma(1+\alpha)} - \frac{\left( x/2 \right)^{-\alpha}}{\Gamma(1-\alpha)}\right|
			\\ & -
			\sum_{k=1}^\infty 
			\frac{(x/2)^{2k}}{k!}
			\left| \frac{\left( x/2 \right)^\alpha}{\Gamma(k+1+\alpha)}
			-
			\frac{\left( x/2 \right)^{-\alpha}}{\Gamma(k+1-\alpha)} \right| \,.
		\end{split}
		\label{eq:widelem1}
	\end{equation}
	The sum can be estimated in the same way as the one in \eqref{eq:widelem_numerator1}: we have
	\begin{equation}
			\left| \frac{\left( x/2 \right)^\alpha}{\Gamma(k+1+\alpha)}
			-
			\frac{\left( x/2 \right)^{-\alpha}}{\Gamma(k+1-\alpha)} \right|
			\le
			\const \frac{2^k \left[ \log(k+1) + \log (2/x) \right]}{(2k-1)!!} |\alpha| \left( \frac{x}{2} \right)^{-|\Re \alpha|}\, ,
	\end{equation}
	and so
	\begin{equation}
			\sum_{k=1}^\infty 
			\frac{(x/2)^{2k}}{k!}
			\left| \frac{\left( x/2 \right)^\alpha}{\Gamma(k+1+\alpha)}
			-
			\frac{\left( x/2 \right)^{-\alpha}}{\Gamma(k+1-\alpha)} \right|
			\le
			\const |\alpha| x^{2 - |\Re \alpha|} \log (1/x)\, ,
	\end{equation}
	and we see that this is dominated by the first term, which was bounded from below in Lemma~\ref{lem:leading_diff} above.  We then have
	\begin{equation}
		\left|\frac{2 \sin \pi\alpha}{\pi}K_\alpha(x)\right|
		\ge 
		\const \times |\alpha| \left( x/2 \right)^{-| \Re \alpha|},
		\label{eq:widelem2}
	\end{equation}
	and combining this with \eqref{eq:widelem_numerator_bound} we obtain the desired bound
	and the proof of Lemma~\ref{lem:wide_integrand_bound} is complete.
\end{proof}

Letting $C_R(w)$ denote the oriented circle of radius $R$ about $w$, the Cauchy formula implies that 
\begin{equation}
	\left|f^{(n)}(w)\right|
	=
	\frac{n!}{2\pi} \left| \oint_{C_R(w)} \frac{f(z)}{\left( z-w \right)^{n+1}} d z\right|
	\le 
	\frac{n!}{R^n}
	\max_{\left|z-w\right|=R} \left|f(z)\right|\, ,
	\label{eq:n_Cauchy}
\end{equation}
for any $f$ which is analytic on an open set containing $C_R(w)$ and its interior.  From Lemma~\ref{lem:wide_integrand_bound} we thus have

\begin{cor} \label{cor:Cauchy_wide}
	For the same $A,X,C_2$ as in Lemma~\ref{lem:wide_integrand_bound}, 
	\begin{equation}
		\left|\frac{\partial^n}{\partial a ^n} \left[ x \frac{K_{1-a}(x)}{K_a(x)} \right]\right|
		\le C_2 n! \left( \frac{\log (2/x)}{A} \right)^n
	\end{equation}
	for all $n \in \bbN$, $a \in [-1/4,1/4]$.
\end{cor}

Since the bounds in Corollary~\ref{cor:Cauchy_wide} are uniformly (in $a$) integrable (in $x$), if we take $\eta \le X$ this allows us to take derivatives inside the integral defining $\tf$ which is therefore infinitely differentiable on the real interval $(-1/4,1/4)$. This result
is going to be crucial for us at $0$: the fact that $\tf$ is $C^\infty$ outside of zero is also a byproduct  of the fact that
we are going to establish (via the next lemma) that $\tf$ is real analytic in $(-1,1)\setminus \{0\}$.

\bigskip

\begin{lemma}
	For any $A \in (0,1/2)$ and $I \in (0,\infty)$, there exist $X \in (0,2)$ and $C \in (0,\infty)$ such that
	\begin{equation}
		\left|f_\alpha (x) \right| \, \le\,  C\, ,
	\end{equation}
whenever $x \in (0,X]$, $\Re \alpha \in (A,1-A)$, and $\left|\Im \alpha \right| \le I$.
	\label{lem:analyticity_bound}
\end{lemma}
\begin{proof}
	Using \DLMF{5.6}{7} and noting that $\Gamma(x) > 1/2$ for all $x > 0$,
	\begin{equation}
		\left|\frac{1}{\Gamma(z)}\right| 
		\le
		\frac{\left( \cosh \pi \Im z \right)^{1/2}}{\Gamma(\Re z)}
		\le 2 \left( \cosh \pi \Im z \right)^{1/2}\, ,
	\end{equation}
	whenever $\Re z > 0$.  Also using \DLMF{5.6}{6} and noting that the Gamma function is concave for positive real arguments,
	\begin{equation}
		\left|\frac{1}{\Gamma(2-\alpha)}\right|
		\ge
		\frac{1}{\Gamma(2 - \Re \alpha)}
		\ge \frac{1}{\Gamma(1)} \vee \frac{1}{\Gamma(2)}
		=
		1.
	\end{equation}
	Applying these two bounds to \eqref{eq:K_expansion1}, we obtain
	\begin{equation}
		\begin{split}
			\left|\frac{\pi \sin \pi\alpha}{\pi} K_\alpha(x)\right|
			&=
			\bigg| 
			\frac{(x/2)^{-\alpha}}{\Gamma(1-\alpha)} 
			+ \left( \frac{x}{2} \right)^{2-\alpha} \sum_{k=0}^{\infty} \frac{(x/2)^{2k}}{(k+1)! \Gamma(k+2-\alpha)}
			\\& \phantom{movemovemovemovemove}
			- \left( \frac{x}{2} \right)^{\alpha} \sum_{k=0}^{\infty} \frac{(x/2)^{2k}}{k! \Gamma(k+1+\alpha)}
			\bigg|
			\\ & \ge
			|\alpha| \frac{(x/2)^{-\Re \alpha}}{\left|\Gamma(2-\alpha)\right|}
			- 
			2 \left[ \left( \frac{x}2 \right)^{2-\Re \alpha} +  \left( \frac{x}2 \right)^{\Re \alpha}\right]
			\left( \cosh \pi I \right)^{1/2}
			\sum_{k=0}^\infty \frac{(x/2)^{2k}}{k!}
			\\ & \ge
			A \left( \frac{X}{2} \right)^{-A}
			-
			4 \left( \frac{X}{2} \right)^A \exp \left( \frac{x^2}{4} \right) \left( \cosh \pi I \right)^{1/2}\, ,
		\end{split}
		\label{eq:an_den_bound}
	\end{equation}
	for all relevant $\alpha$ and $x$; choosing $X$ sufficiently small, the last bound can be made positive.

	Similarly, noting also that $\Gamma(x)$ is negative (resp.\ positive) and decreasing for $x \in (-1/2,0)$ (resp.\ $(0,1/2)$), 
	\begin{multline}
			\left|\frac{\pi \sin \pi\alpha}{\pi} K_{1-\alpha}(x)\right|
			=
			\bigg| 
			\frac{(x/2)^{\alpha-1}}{\Gamma(\alpha)} 
			-
			\frac{(x/2)^{1-\alpha}}{\Gamma(-\alpha)} 
			+ \left( \frac{x}{2} \right)^{\alpha+1} \sum_{k=0}^{\infty} \frac{(x/2)^{2k}}{(k+1)! \Gamma(k+1+\alpha)}\\
			 \phantom{movemovemovemovemovemove}
			- \left( \frac{x}{2} \right)^{3-\alpha} \sum_{k=0}^{\infty} \frac{(x/2)^{2k}}{(k+1)! \Gamma(k+1-\alpha)}
			\bigg|
			\\  \le
			2 \left( \cosh \pi I \right)^{1/2}
			\left\{
				\frac{(x/2)^{A-1}}{\Gamma(A)} 
				-
				\frac{(x/2)^{A}}{\Gamma(-A)} 
				+ 
				\left[ \left( \frac{x}{2} \right)^{A+1} +  \left( \frac{x}{2} \right)^{A+2} \right]
				\exp\left(\frac{x^2}{4}  \right)
			\right\} ,
		\label{eq:an_num_bound}
	\end{multline}
	and combining this with \eqref{eq:an_den_bound} we obtain a suitable bound on $\left|f_\alpha(x)\right|$.
	This completes the proof of Lemma~\ref{lem:analyticity_bound}.
\end{proof}

By Lemma~\ref{lem:analyticity_bound} we have 
that for any $\alpha \in \bbC $ with $\Re \alpha \in (0,1)$, we can choose $A,I$ to obtain such a bound on a neighborhood of $\alpha$; using the Cauchy formula this impies that $f'_x$ is uniformly bounded on some smaller neighborhood of $\alpha$, which allows us to exchange differentiation and integration to see that $\tf$ is holomorphic on that neighborhood.
We can then conclude that $\tf$ is analytic on $\left\{ \alpha \in \bbC \middle| \Re \alpha \in (0,1) \right\}$, and by the symmetry noted in \eqref{eq:symF3} it is also analytic on $\left\{ \alpha \in \bbC \middle| \Re \alpha \in (-1,0) \right\}$.

\medskip

All that remains is to show that the derivatives of $\tf$ at the origin grow as stated; since this will imply that the associated Taylor series is divergent, this will also prove that $\tf$ is not analytic there.
We begin by providing a more precise characterization of the poles of $f_x$ for small $x$.

\begin{lemma}
	\label{lem:nu_n}
  There exist $X,C>0$ and a sequence of functions $\nu_n : (0,\infty) \mapsto (0,\infty)$, satisfying
\begin{equation}
	\left| \nu_n(x)  - \frac{n \pi}{\log(2/x) - \gamma} \right|
	\le
	\frac{C n^3 }{(\log x )^4}
	\label{eq:nu_n_asymp}
\end{equation}
and $\nu_1(x) < \nu_2(x) < \dots$,
such that for all $x \in (0,X]$, $K_\alpha (x) = 0$ iff $\alpha = \pm i\nu_n(x)$.
\end{lemma}
\begin{proof}
	For $x> 0$ and $\nu \in \bbR$, \eqref{eq:K_expansion1}  can be rewritten using some properties of the Gamma function (\DLMF{5.2}{5} and \DLMF{5.4}{3}) as
\begin{equation}
	K_{i\nu}(x)
	=
	-\left( \frac{\pi}{\nu \sinh \left( \pi\nu \right)} \right)^{1/2}
	\sum_{k=0}^\infty
	\frac{\left( x^2/4 \right)^k}{k!}
	\frac{\sin\left( \theta_k (\nu) \right)}{\sqrt{(1^2+\nu^2)\dots(k^2+\nu^2)}}
	,
	\label{eq:K_sin_series}
\end{equation}
where
\begin{equation}
  	\theta_k (\nu) :=  \nu \log(x/2) -  \arg \Gamma(1+k+i\nu)\, ,
	\label{eq:th_kn_def}
\end{equation}
(cf.\ \cite[(2.7-8)]{Dunster}, where this expression is used to study the $x$-zeroes and their dependence on $\nu$).

Then the solutions of $K_{i\nu}(x) = 0$ are the nonzero solutions of
\begin{equation}
	\sin \theta_0(\nu)
	=
	S_x(\nu)
	:=
	-
	\sum_{k=1}^\infty
	\frac{\left( x^2/4 \right)^k}{k!}
	\frac{\sin\left( \theta_{k}(\nu) \right)}{\sqrt{(1^2+\nu^2)\dots(k^2+\nu^2)}}.
	\label{eq:K_sin_approx1}
\end{equation}

Using the definition of the $\psi$ function and the expansion~\DLMF{5.7}{6},
\begin{equation}
  \theta'_0(\nu)
  =
  \log(x/2)
  -
  \Re \psi (1+i\nu)
  =
  \log(x/2)
  + \gamma - \sum_{m=1}^\infty \frac{\nu^2}{m(m^2+\nu^2)}
  \le 
  \log(x/2)
  +\gamma,
  \label{eq:psi_1+in}
\end{equation}
so for $ 0 < x < 2e^{-\gamma} = 1.1229\dots$ $\theta_0$ is strictly decreasing.
From its definition in \eqref{eq:K_sin_approx1}, it is apparent that $S_x$ can be bounded 
\begin{equation}
	\left|
	S_x(\nu)
	\right|
	\le 
	\sum_{k=1}^\infty
	\frac{\left( x^2/4 \right)^k}{(k!)^2}
	\le
	\sum_{k=1}^\infty
	\frac{x^{2k}}{(2k)!}
	\le \cosh x - 1,
	\label{eq:K_remainder}
\end{equation}
uniformly in $\nu$, and so
we see that for each  $0< x < \cosh^{-1}(2) = 1.317\dots$
all solutions $\nu$ of 
\eqref{eq:K_sin_approx1} satisfy
\begin{equation}
  	-\theta_0(\nu)
	\in
	[n \pi- \sin^{-1} \left( \cosh x - 1  \right),n \pi+ \sin^{-1} \left( \cosh x - 1  \right)].
	\label{eq:nu_n_range}
\end{equation}
for some integer $n$, and there is at least one solution for each $n$.
More precisely, for such solutions the derivative of the left hand side of \eqref{eq:K_sin_approx1} satisfies
\begin{equation}
	(-1)^{n+1} \frac{\partial}{\partial\nu} 
	\sin\left( \theta_{0}(\nu) \right)
	=
	(-1)^{n+1}
	\theta'_0(\nu) \cos\left( \theta_{0}(\nu) \right)
	\ge
	|\theta'_0(\nu)|\sqrt{1-(\cosh x - 1)^2},
	\label{eq:nu_n_lhs_deriv}
\end{equation}
and in light of \eqref{eq:psi_1+in}, for $X$ small enough this can be bounded from below by any positive number uniformly in $\nu$. As for the right-hand side of \eqref{eq:K_sin_approx1}, first note that from \DLMF{5.5}{2} we have
\begin{equation}
	\psi(1+k+\alpha) 
	= 
	\psi(1+\alpha) 
	+
	\sum_{m=1}^k \frac{1}{m+\alpha}\, ,
	\label{eq:psi_recursion}
\end{equation}
for all $k \in \bbN, \alpha \in \bbC$,
and so
\begin{equation}
	\begin{split}
		\left|\theta'_k(\nu) - \theta'_0(\nu)\right|
		&=
		\left|\Re \psi(k+1+\alpha)- \Re \psi(1+\alpha)\right|
		\le
		\left|\psi(k+1+\alpha)- \psi(1+\alpha)\right|
		\\ & \le
		\sum_{m=1}^k \frac{1}{m}
		\le \int_1^{k+1} \frac{dm}{m}
		= \log (k+1).
	\end{split}
	\label{eq:psi_log_bound}
\end{equation}
Using this, we have
\begin{equation}
	\begin{split}
		\left| \frac{\partial}{\partial\nu}
		\frac{\sin\left( \theta_{k}(\nu) \right)}{\sqrt{(1^2+\nu^2)\dots(k^2+\nu^2)}}
		\right|
		=
		\left|\frac{\theta'_k(\nu) \cos \theta_k(\nu) - \sum_{m=1}^k \frac{\nu}{m^2 + \nu^2} \sin \theta_k(\nu)}{\sqrt{(1^2+\nu^2)\dots(k^2+\nu^2)}}\right|
		\\
		\le
		\frac{\left|\theta'_0(\nu)\right| + 2 \log(k+1)}{k!}
		\le
		\frac{\left|\theta'_0(\nu)\right| + 2}{(k-1)!},
	\end{split}
\end{equation}
which implies  
\begin{equation}
	\left|S'_x(\nu)\right|
	\le
	\sum_{k=1}^\infty
	\frac{(x^2/4)^{2k} }{k!} 
		\frac{\left|\theta'_0(\nu)\right| + 2}{(k-1)!}
	\le 
	\left[  \left|\theta'_0(\nu)\right| + 2\right]
	\sum_{k=1}^\infty \frac{x^{2k}}{(2k-2)!}
	\le 
	\left[  \left|\theta'_0(\nu)\right| + 2\right]
	x^2 \cosh x,
\end{equation}
which, for $x$ small enough, is smaller than the right-hand side of \eqref{eq:nu_n_lhs_deriv} for all $\nu$.
This suffices to show that there is only one solution of \eqref{eq:K_sin_approx1} in each of the intervals in \eqref{eq:nu_n_range}.  The solution for $n=0$ must be the trivial solution $\nu=0$ which does not correspond to a solution of $K_{i \nu}(x) = 0$.
This uniqueness also implies that the solutions for negative and positive $n$ are related by the symmetry $K_{i\nu}(x) = K_{-i\nu}(x)$, so we see that it is possible to relate the zeros to a family of functions as desired.

To see that the functions $\nu_n$ satisfy the bound~\eqref{eq:nu_n_asymp}, 
we first note that \eqref{eq:nu_n_range} and \eqref{eq:psi_1+in} together imply that
\begin{equation}
  \nu_n (x) 
  \le
  \frac{n \pi + \cosh x - 1}{\inf_{\nu \ge 0} |\theta'_0(\nu)| }
  \le \frac{n \pi + \cosh x - 1}{\log (2/x) - \gamma}
  \le \const \frac{n}{\log 1/x}\, ,
\end{equation}
for all $x \in (0,X)$;
and also that using the same expansion as in \eqref{eq:psi_1+in} we have
\begin{equation}
  \left|
	\gamma + \psi(1+i\nu)
  \right|
  \le
\sum_{m=1}^\infty \frac{\nu^2}{m(m^2+\nu^2)}
	\le \nu^2 \zeta(3)
	\textup{ and so }
	\left|
		\gamma \nu + \arg \Gamma (1+i \nu) 
	\right|
	\le \const |\nu|^3\,,
\end{equation}
for all $\nu \in \bbR$, where $\zeta(s) := \sum_{m=1}^\infty m^{-s}$ is the Riemann zeta function \cite[Section~25.2]{DLMF}.
Then recalling the definition of $\theta_0$, this imples
\begin{equation}
  \left|
  \left[ \log (x/2) - \gamma \right] \nu_n(x)  - \theta_0 (\nu_n(x)) 
  \right|
  =
	\left|
		\gamma \nu_n(x) + \arg \Gamma (1+i \nu_n(x)) 
	\right|
	\le \const \frac{n^3}{(\log 1/x)^3}\, ,
\end{equation}
and restating \eqref{eq:nu_n_range} as $\left|\theta_0\left( \nu_n (x)  \right) - n\pi\right| \le \const \ x^2$ this gives the desired bound and the proof of Lemma~\ref{lem:nu_n} is complete.
\end{proof}

We denote the residue of $f_x$ at $\pm i \nu_{n}(x)$ by $\pm R_n(x)$.
Letting 
\begin{multline}
\label{eq:K_tilde_def}
	\tilde{K}_\alpha(x)
	:=
	\frac{\pi}{2 \sin \pi\alpha}
	\sum_{k=0}^\infty
	\bigg\{ 
		\frac{(x/2)^\alpha}{\Gamma(k+1+\alpha)} \left[ \log(x/2) - \psi(k+1+\alpha) \right]
	+
	\\
	\frac{(x/2)^{-\alpha}}{\Gamma(k+1-\alpha)} \left[ \log(x/2) - \psi(k+1-\alpha) \right]
\bigg\},
\end{multline}
we have $\tfrac{\partial}{\partial\alpha} K_\alpha(x) = \tilde{K}_\alpha(x) - \pi \cot(\pi \alpha) K_\alpha(x)$, and so
\begin{equation}
	R_n (x) = x \frac{K_{1-i \nu_n(x)}(x)}{\tilde{K}_{i\nu_n(x)}(x)}.
	\label{eq:Rn}
\end{equation}
Noting that from Lemma~\ref{lem:nu_n} and \DLMF{5.7}{4} we have 
\begin{equation}
\begin{split}
	\left( \frac{x}{2} \right)^{\pm i \nu_1(x)}
	&=
	-1 \pm i \pi \frac{\gamma}{\log(2/x) - \gamma} + O\left( \frac{1}{|\log x|^4} \right), 
	\\
	\frac{1}{\Gamma(1 \pm i \nu_1)}
	&= 
	1 \pm i \gamma \nu_1(x) + O\left( \frac{1}{|\log x|^2} \right),
	\\
	\psi(1 \pm i \nu_1 (x))
	&=
	-\gamma \pm \zeta(2) \nu_1(x) + O\left( \frac{1}{|\log x|^2} \right),
	\end{split}
\end{equation}
paying attention to cancellations, \eqref{eq:K_tilde_def} gives
\begin{equation}
	\tilde{K}_{i \nu_1(x)}(x)
	=
	i \frac{L^2(x)}{\pi} + O\left( 1 \right),
	\label{eq:Ktilde_asymp}
\end{equation}
where the $k = 1,2, \ldots$ terms in the sum are bounded in the same way as the similar sum appearing in the proof of Lemma~\ref{lem:wide_integrand_bound}.
Similarly, noting that 
\begin{equation}
1/\Gamma(-i\nu_1(x)) \,=\, -i\nu_1(x) [1 - i \gamma \nu_1(x) + O(\nu_1^2(x)]\, ,
\end{equation}
(from \DLMF{5.7}{1}) and expanding $K_{1-\alpha}(x)$ as in \eqref{eq:widelem_numerator1}, we have
\begin{equation}
	K_{1 - i \nu_1(x)}(x)
	=
	\frac{1}{x} + O\left( \frac{1}{x |\log x|^2} \right),
\end{equation}
taking advantage of a cancellation between the subleading terms in $(x/2)^{1+i\nu_1(x)}$ and $1/\Gamma(1+i\nu_1(x))$,
and so
\begin{equation}
	R_1(x) = - i \frac{\pi}{\left[ \log(2/x)-\gamma \right]^2} + O\left( \frac{1}{|\log x|^4} \right).
	\label{eq:R1_asymp}
\end{equation}

\begin{lemma}
\label{lem:5.9g}
	For any $a \in (1,2)$, There exist $X_a, C_a > 0$ such that 
	\begin{equation}
		\left|
		x \frac{K_{1-\alpha}(x)}{K_\alpha(x)}
		+ 2 i R_1(x) \frac{\nu_1(x)}{\alpha^2 + \nu_1^2(x)}
		\right|
		< 
		\frac{C_a}{|\log 2/x - \gamma|}\, ,
	\end{equation}
whenever $x \in (0,X_a]$ and $|\alpha| = a (\pi / |\log (2/x) - \gamma|)$.
\end{lemma}

\begin{proof}
From Lemma~\ref{lem:nu_n} we see that we can choose 
$X_a < 2 e^{-\gamma}$ such that 
\begin{equation}
\nu_2 (x) \,>\, a (\pi / |\log (2/x) - \gamma|) \,>\, \nu_1(x)\, ,
\end{equation}
for all $x \in (0,X_a]$; then the quantity to be bounded is a continuous function of both $\alpha$ and $x$ for all relevant values except $x=0$,
so we need only check that
\begin{equation}
	\limsup_{x \searrow 0}
	\left|
	x L(x) \frac{K_{1- \tilde{\alpha}/ L(x)}(x)}{K_{ \tilde{\alpha}/L( x)}(x)}
	+ 2 i L(x) R_1(x) \frac{\nu_1(x)}{\left(\frac{\tilde{\alpha}}{L( x)}\right)^2 + \nu_1^2(x)}
	\right|\, ,
	\label{eq:limbound}
\end{equation}
is bounded uniformly for $|\tilde{\alpha}| = a \pi$, 
where for brevity $L(x) := \log 2/x - \gamma$; in fact we will show that both terms in the sum are suitably bounded.
In fact
\begin{equation}
	\lim_{x \searrow 0}\left( \frac{x}{2} \right)^{\tilde \alpha / \log x}
	\, =\,  e^{\tilde \alpha} \ \
	\textup{ and } \ \
	\lim_{x \searrow 0}\Gamma\left( 1 + \frac{\tilde \alpha}{\log x} \right)
	\, =\,  1\, ,
\end{equation}
and with \eqref{eq:K_expansion1} and \eqref{eq:K_remainder} this implies that
\begin{equation}
	\lim_{x \searrow 0}\frac{K_{\tilde \alpha/L(x)}(x) }{L(x)}
	\, =\, 
	\frac{e^{\tilde \alpha} - e^{-\tilde \alpha}}{2 \tilde \alpha}\, .
\end{equation}
Noting that $\lim \Gamma(\tilde \alpha/L( x)) / L( x) =1/\tilde \alpha$ we also have 
\begin{equation}
	x K_{1-\tilde\alpha/L(x)}(x)
	\stackrel{{x \searrow 0}}\sim
	x \frac{L(x)}{2 \tilde\alpha}\frac{(x/2)^{\frac{\tilde\alpha}{L(x)} - 1}}{\Gamma\left( \frac{\tilde\alpha}{L(x)} \right)}
	\sim e^{\tilde\alpha}; 
\end{equation}
then
\begin{equation}
	 \lim_{x \searrow 0}x L(x) \frac{K_{1-\tilde\alpha/ L(x)}(x)}{K_{\tilde\alpha/L(x)}(x)}
	\,=\, 
	\frac{2 \tilde\alpha}{e^{2\tilde\alpha} - 1},
	\label{eq:integrand_scaling}
\end{equation}
which, recalling $|\tilde{\alpha}| = a \pi \in (\pi,2\pi)$, is indeed uniformly bounded.

Recalling \eqref{eq:R1_asymp}, we have
\begin{equation}
	\lim_{x \searrow 0}L^2 (x) R_1(x) \,=\,  - i \pi\, ,
	\label{eq:R1_leading}
\end{equation}
and thus
\begin{equation}
	\lim_{x \searrow 0}2 i L(x) R_1(x) \frac{\nu_1(x)}{\left(\frac{\tilde{\alpha}}{L( x)}\right)^2 - \nu_1^2(x)}
	\,=\, 
	2 \frac{\pi^2}{\tilde{\alpha}^2 - \pi^2}\, ,
\end{equation}
which is also uniformly bounded in a suitable fashion. Hence \eqref{eq:limbound} is proven
and therefore also the proof of  Lemma~\ref{lem:5.9g} is complete.
\end{proof}

Noting that
\begin{equation}
	 - 2 i R_1(x) \frac{\nu_1(x)}{\alpha^2 + \nu_1^2(x)}
	=
	\frac{R_1(x)}{\alpha- i \nu_1(x)}
	-
	\frac{R_1(x)}{\alpha+ i \nu_1(x)},
\end{equation}
the expression examined above is an analytic function of $\alpha$ in the interior of the circles under consideration apart from removable singularities,
and so using Lemma~\ref{lem:5.9g} and \eqref{eq:n_Cauchy} we have
\begin{cor}
	For any $a \in (1,2)$, There exist some $C_a,X_a > 0$ and a sequence of functions $I_n:(0,X_a] \to \bbC$ such that
	\begin{equation}
		x \frac{K_{1-\alpha}(x)}{K_\alpha(x)}
		=
		\sum_{n=0}^\infty I_n(x) \alpha^n
		-2 i R_1(x) \frac{\nu_1(x)}{\alpha^2 + \nu_1^2(x)}\,,
		\label{eq:integrand_Taylor}
	\end{equation}
	whenever $x \in (0,X_a]$ and $|\alpha| \le a (\pi / |\log x|)$, and
	\begin{equation}
		|I_n(x)| \le C \left( \frac{ |\log (2/x) - \gamma|}{a \pi}\right)^{n-1}\, ,
	\end{equation}
	for all $n$.
\end{cor}
We then have 
\begin{equation}
	\left.
	\frac{1}{n!} \frac{\partial^n}{\partial\alpha^n} 
	\left[ x \frac{K_{1-\alpha}(x)}{K_\alpha(x)} \right]
	\right|_{\alpha=0}
	=
	I_n(x) 
	-
	\left\{
	\begin{matrix}
		\left( -1 \right)^{n/2} 2 i \frac{R_1(x)}{\left( \nu_1(x) \right)^{n+1}}, & n \text{ even}
		\\
		0, & n \text{ odd.}
	\end{matrix}
	\right. 
	\label{eq:integrand_taylor}
\end{equation}

From \eqref{eq:nu_n_asymp} and \eqref{eq:R1_asymp} we have that for any $\Neps < 2 e^{-1-\gamma}$ there exist finite $C_R, C_\nu$ such that 
\begin{equation}
	\left| R_1(x) +  \frac{i \pi}{L^2(x)} \right|
	\le \frac{C_R}{L^4(x)}
	\text{ and }
	\left|
	\frac{1}{\nu_1(x)} - \frac{L(x)}{\pi}
	\right|
	\le 
	\frac{C_\nu}{\pi L^2(x)}\, ,
\end{equation}
for all $x \in (0,\Neps]$, and so 
\begin{multline}
	\left|
	\frac{R_1(x)}{\left( \nu_1(x) \right)^{n+1}}
	+  i  \frac{L^{n-1}(x)}{\pi^n} 
	\right|
	= 
	\left|
	R_1(x) \left( \frac{L(x)}{\pi} +
	\left( \frac 1{\nu_1(x)}- \frac{L(x)}{\pi} \right)\right)^{n+1}
	-  i  \frac{L^{n-1}(x)}{\pi^n} 
	\right| 
	\\
	\le 
	\left \vert R_1(x) \left( \frac{L(x)}{\pi}\right)^{n+1}
	-  i  \frac{L^{n-1}(x)}{\pi^n} \right\vert
	+ \left \vert R_1(x) \right\vert
	\frac 1{\pi^{n+1}}
	\sum_{m=1}^{n+1} \binom{n+1}{m} 
	L(x)^{n+1-3m} C_\nu^m
	\\
	\le
	\frac{C_R}{\pi^{n+1}}L^{n-3}(x)
	+
	\left(\pi + \frac {C_R}{L^2(x)}\right)
	\frac 1{\pi^{n+1}}
	\sum_{m=1}^{n+1} \binom{n+1}{m} 
	L(x)^{n-3m-1} C_\nu^m
	.
	\label{eq:mainpoles_error}
\end{multline}
We have 
\begin{equation}
	\int_0^\Neps L^n(x) dx
	=
	2 e^{-\gamma} \int_{\log(2/\Neps) - \gamma}^\infty L^n e^{-L} dL
	=
	2 e^{-\gamma} \Gamma(n+1, \bar \Neps)\, ,
	\label{eq:logn_integral}
\end{equation}
for $n \ge 0$,
where $\bar \Neps := \log(2/\Neps) - \gamma$ for brevity (note $\bar \Neps > 0$ since we have assumed $\Neps < 2 e^{-\gamma}$), and where $\Gamma(n,\Neps) := \int_\Neps^\infty t^{n-1} e^{-t} dt	$ is the upper incomplete Gamma function \cite[Chapter~8]{DLMF}.
For $n<0$, since we have assumed $\Neps \le 2 e^{-\gamma-1}$ we have $L(x) \ge 1$ for $x \in (0,\Neps)$, and thus
\begin{equation}
	0 \le
	\int_0^\Neps L^n(x) dx
	\le 
	\Neps\, .
\end{equation}
Note that by combining \DLMF{8.8}{2} and \DLMF{8.10}{1} we obtain
\begin{equation}
	\frac{\Gamma(n+1,\bar \Neps)}{\Gamma(n,\bar \Neps)}
	=
	n + \frac{ {\bar \Neps}^n e^{-\bar \Neps}}{\Gamma(n,\bar \Neps)}
	\ge
	n + \bar \Neps
	\ge 
	n \,,
	\label{eq:incomplete_Gamma_ratio}
\end{equation}
which can be applied iteratively to obtain
\begin{equation}
	\frac{\Gamma(n+1,\bar \Neps)}{\Gamma(n+1-m,\bar \Neps)}
	\ge 
	\frac{n!}{(n-m)!}\, .
	\label{eq:incomplete_Gamma_growth}
\end{equation}
for $m \le n$.  Then
\begin{equation}
	\begin{split}
		\sum_{m=1}^n
		\binom{n}{m}
		C_\nu^m &
		\int_0^\Neps L^{n-3m-1}(x)
		\\
		& \le
		\sum_{m=1}^{\floor{(n-1)/3}}
		\binom{n}{m}
		\frac{(n-3m)!}{(n-1)!}
		C_\nu^m
		\Gamma(n,\bar \Neps)
		+
		\Neps
		\sum_{m=\floor{(n-1)/3}+1}^n 
		\binom{n}{m} C_\nu^m
		\\ 
		&=
		n
		\sum_{m=1}^{\floor{(n-1)/3}}
		\frac{1}{m!}
		\frac{(n-3m)!}{(n-m)!}
		C_\nu^m
		\Gamma(n,\bar \Neps)
		+
		\Neps
		\sum_{m=\floor{(n-1)/3}+1}^n 
		\binom{n}{m} C_\nu^m
		\\
		&\le
		n \Gamma(n,\bar \Neps)
		\sum_{m=1}^\infty 
		\frac{(n-3m)!}{(n-m)!}
		\frac{C_\nu^m}{m!}
		+\Neps \sum_{m=0}^n \binom{n}{m} C_\nu^m
		\\
		&=
		\frac{27}{2 (2n-3)(2n-6)}\left( e^{C_\nu}-1 \right) \Gamma(n,\bar \Neps)
		+\left( C_\nu +1 \right)^n \Neps \, ,
	\end{split}
\end{equation}
using the observation that
\begin{equation}
	\frac{(n-m)!}{(n-3m)!}
	\ge
	\frac23 n \left( \frac{2n-3}{3} \right)\left( \frac{2n - 6}{3} \right)\, ,
\end{equation}
for $1 \le m \le n/3$.
Using this to bound the second term on the right-hand side of of Inequality~\eqref{eq:mainpoles_error} and bounding the other two terms similarly,
we see that the integral in $x$ from $0$ to $\Neps$ of the right-hand side of Inequality~\eqref{eq:mainpoles_error} admits a bound of order $\pi^{-n-1}\Gamma(n,\bar \Neps)/n^2$ for large $n$.
We also have 
\begin{equation}
	\left| \int_0^\Neps I_n(x) dx \right|
	\le
	\int_0^\Neps |I_n(x)| d x 
	\le 
	\frac{2 C_a e^{-\gamma}}{a^{n-1}} \frac{ \Gamma(n,\bar \Neps)}{\pi^{n-1}},
\end{equation}
for any $a \in (1,2)$,
and so the dominant behavior of the even Taylor coefficients $\tf^{(2n)}(0)/(2n)!$ for $n$ large is that of
\begin{equation}
\begin{split}
	\left( -1 \right)^{n+1}\frac{2}{\pi^{2n}} \int_0^\Neps \left( \frac{\log(2/x) - \gamma}{\pi} \right)^{n-1} d x
	\, &=\, 
	4 e^{-\gamma} \left( -1 \right)^{n+1}\frac{\Gamma(2n,\bar \Neps)}{\pi^{2n}}
	\\
	\, &
	\sim\, 
	4 e^{-\gamma} \left( -1 \right)^{n+1}\frac{(2n-1)!}{\pi^{2n}}\, ,
	\end{split}
\end{equation}
noting $\Gamma(n,\bar \Neps) \sim \Gamma(n) = (n-1)!$ \cite[8.2.3, 8.11.4]{DLMF},
while the symmetry noted in \eqref{eq:symF3} imposes that $\tf'(0) = 4 \eta$ and $\tf^{(2n+1)}(0) = 0$  for $n =1,2, \ldots$. 

The proof of Theorem~\ref{th:MW} is therefore complete.
\qed

\section{Scaling limit of matrix product: proof of Theorem~\ref{prop:scalingDH83}}
\label{sec:scaling-proofs}

As announced, we generalize  the set-up of  \eqref{eq:Deltamodel}-\eqref{eq:AgD2} in the sense that 
we prove 

\begin{theorem}
\label{prop:scalingDH83g}
Consider a family of positive random variables $\{Z^\gD\}_{\gD\in (0, \gD_0)}$ such that
$\bbP(Z^\gD=y)=0$ for every $y$ and such that for some $\gs>0$ and $\ga \in \bbR$ we have
\begin{equation}
\label{eq:condDH83-g1}
\lim_{\gD \searrow 0} \frac{\bbE \left[Z^\gD -1 \right]} \gD \, =\, \frac 12 \gs ^2 (1-\ga)
\ \ \ \ \text{ and } \ \ \ \
\lim_{\gD \searrow 0} \frac{\bbE \left[\left(Z^\gD -1\right)^2 \right]} \gD \, =\,   \gs ^2 \, .
\end{equation} 
Assume moreover that for every $c>0$
\begin{equation}
\label{eq:condDH83-g1+}
\lim_{\gD\searrow 0}
\frac 1 \gD \bbP \left( \left \vert Z^\gD -1\right \vert \, >\, c \right) \, =\, 0\,,
\end{equation}
and 
\begin{equation}
\label{eq:condDH83-g2}
\limsup_{\gD\searrow 0}\left \vert \frac{\bbE[1/Z^\gD]-1}\gD \right\vert \, <\,  \infty\, .
\end{equation}
Then if we consider the model \eqref{eq:AgD1}-\eqref{eq:AgD2} with the IID sequence $\{Z^\gD (n)\}_{n=1,2, \ldots}$
generalized to an arbitrary IID sequence with common law satisfying \eqref{eq:condDH83-g1}-\eqref{eq:condDH83-g2}, then  
\eqref{eq:approx1} and \eqref{eq:approx2} hold true. 
\end{theorem}

\medskip 

Theorem~\ref{prop:scalingDH83g} directly implies Theorem~\ref{prop:scalingDH83}: the cases of two more classes of distributions are treated just before the proof. 
Note that with \eqref{eq:condDH83-g1} we are in reality just assuming the existence of the two limits and that the second limit is not zero.
The second assumption, i.e.  \eqref{eq:condDH83-g1+},
 barely fails to be a consequence of 
\eqref{eq:condDH83-g1}. 
The third assumption, i.e. \eqref{eq:condDH83-g2}, is used to control the  amount of the   mass of $Z^\gD$ that is  close to zero:
it is not difficult to realize that, given \eqref{eq:condDH83-g1}, replacing  \eqref{eq:condDH83-g2} with 
the stronger condition
\begin{equation}
\label{eq:condDH83-g3}
\lim_{\gD\searrow 0}\frac{\bbE[1/Z^\gD]-1}\gD \, =\, \frac 12(\ga -1) \gs^2\, ,
\end{equation} 
leads to very little loss of generality. Moreover, we have assumed that the law of $Z^\gD$ ha no mass just to be sure that we do not fall into a pathological case for the theory of product of random matrices, but all we need is a condition 
that guarantees the existence of the limit in \eqref{eq:discL}
and that the Markov chain associated to matrix product is ergodic: this is  true in greater generality    \cite{cf:BL}.
\smallskip

Before giving the proof let us show two classes of examples to which Theorem~\ref{prop:scalingDH83g} applies:
\begin{enumerate} 
\item The distribution chosen in \cite{MW1,MWbook} falls into the class
\begin{equation}
 \label{eq:densitylambda1before}
 \NMW \gl_1^{-\NMW}y^{\NMW-1} \ind_{(0, \gl_1)}(y) \, ,
 \end{equation}
with
\begin{equation}
\gl_1\, =\, \gl_1(\ga,\NMW) \, =\,   1+1/\NMW +(1-\ga)/\NMW^2+ o( 1/\NMW^2)\, .
\end{equation}
$\NMW(\to \infty)$ is the parameter that tunes the strength of the disorder and Theorem~\ref{prop:scalingDH83g} can be applied 
by setting  $\gD=\NMW^{-2}$: let us verify the hypotheses.
We compute for every $\nu$
 \begin{equation}
 \label{eq:betanu}
 \bbE \left[ \left( Z_\NMW \right)^\nu\right]\, =\, \frac{\gl_1(\ga,\NMW)^\nu}{1+ \frac \nu \NMW}
 \, =\,1+ \frac{\nu (\nu-\ga)} {2\,  \NMW^2} + o\left( \frac 1{\NMW^2} \right)\, ,
 \end{equation} 
and we directly obtain 
\begin{equation}
\lim_{{\NMW \to \infty}}
\NMW^2\bbE\left[  Z_\NMW -1
\right] \, =\,  \frac{1-\ga}2
\ \ \ \text{ and } \ \ \  \lim_{{\NMW \to \infty}}
\NMW^2\bbE\left[ \left( Z_\NMW -1\right)^2
\right] \, =\, 1\, ,
\end{equation}
and 
\begin{equation}
  \bbE\left[ Z_\NMW^{\pm 2}\right]\,  =\, 
  1+ \frac{ (2\mp\ga)} {\,  \NMW^2} + o\left( \frac 1{\NMW^2} \right)\, =\, 
  \exp \left(  \frac{ (2\mp\ga)} {\,  \NMW^2}\right) 
  + o\left( \frac 1{\NMW^2} \right)\,\, ,
  \end{equation}
Moreover or every $c\in (0,1)$ 
the event $\{\vert  Z_\NMW -1\vert >c\}= \{  Z_\NMW -1 <-c\}$ if $\NMW$ is sufficiently large, because $Z_\NMW
 \le \gl_1(\ga, \NMW)$, which tends to one for $\NMW \to \infty$. On the other hand 
 $\bbP( Z_\NMW -1 <-c)=((1-c)/ \gl_1)^\NMW$, which is bounded by $(1-c)^\NMW$ since $\gl_1>1$.
\item Choose a centered and compactly supported probability density $p(\cdot)$ and set  $\gs ^2:=\int t^2 p(t) \dd t$. Then the  random variable $Z^\gD$ with density given by 
\begin{equation}
y \mapsto \frac 1{\sqrt{\gD}} p \left( \frac{y- m_\gD}{\sqrt{\gD}} \right)
\ \ \ \ \text{ with } m_\gD:= 1 + \frac 12 \gs^2 (1-\ga) \gD\, ,
\end{equation}
with $\gD$ smaller than a suitable $\gD_0>0$, 
satisfies the hypotheses of Theorem~\ref{prop:scalingDH83g}.
\end{enumerate}

\medskip

\noindent
\emph{Proof of Theorem~\ref{prop:scalingDH83g}.}
We start with the proof of \eqref{eq:approx1}, which  is a direct application of the approximation-diffusion principle: we exploit \cite[pp. 266--272]{StroockVaradhan}, notably \cite[Assumptions (2.4)-(2.6), Theorem~11.2.3]{StroockVaradhan}.
Equivalently, one can resort to   \cite[Corollary 4.2 in Chapter 7]{EthierKurtz}. The procedure demands three steps:
\medskip 
\begin{itemize}
\item compute the local drift at $\ux \in \R^2$: uniformly for $\ux=(x_1, x_2)^{\tt t}$ in compact sets  
\begin{equation}
\begin{split}
b^\gD\left(\ux\right) \,&=\,  \gD^{-1} \bbE A^\gD \ux = 
\begin{pmatrix}
0&\gep \\
\gep \bbE\left[Z^\gD \right]
& \frac{\bbE\left[Z^\gD -1 \right]  }{\gD} 
\end{pmatrix} \ux
\\
&\stackrel{\gD \searrow 0}\longrightarrow \,  b\left(\ux\right)\, :=\, b\;  \ux , \qquad \text{ with } b:= 
\begin{pmatrix}
0&\gep \\
\gep & (1-\ga)\frac{\gs^2}{2} 
\end{pmatrix} ,
\end{split}
\end{equation}
where we have applied the first assumption in \eqref{eq:condDH83-g1};
\item
compute the diffusion matrix at $\ux$: again uniformly we have
\begin{equation}
a^\gD(\ux) \, =\, \gD^{-1} \bbE \left[A^\gD \ux \, \ux^{\tt t} (A^\gD)^{\tt t}\right]\\
\stackrel{\gD\searrow 0}\longrightarrow  a(\ux)\,:=\,
\begin{pmatrix}
0&0 \\
0 & \gs^2 x_2^2
\end{pmatrix}  \, ,
\end{equation}
where we have applied both assumptions in \eqref{eq:condDH83-g1};
\item observe that, by  \eqref{eq:condDH83-g1+}, $\gD^{-1} \bbP (  |A^\gD| \geq c) \to 0$ for every $c>0$.
\end{itemize}
\medskip

Then, since  the stochastic differential system  with drift $b(\cdot)$ and 
diffusion matrix $a(\cdot)$   has  unique (strong) solution, the 
Markov chain $X^\gD$ converges in law to the diffusion process
with drift $b(\cdot)$ and 
diffusion matrix $a(\cdot)$, which is precisely the solution $X$ to the stochastic differential system
\eqref{eq:sys}.
This completes the proof of  \eqref{eq:approx1}. 

\smallskip

In order to prove \eqref{eq:approx2} 
we start by observing that $\widehat \cL _{Z^\gD}(\gep)=\widehat \cL _{Z^\gD}(-\gep)$, in agreement with the analogous result for
$\cL_{\gs, \ga}(\cdot)$  (Theorem~\ref{th:Lyap}(2)), because $D(I+A^\gD)D$, with
$D$ the diagonal matrix with $+1$ and $-1$ on the diagonal, is equal to   $I+A^\gD$ with $\gep$ replaced by $-\gep$.
Hence we can restrict to $\gep >0$.
Moreover if
 we set
 $Y^\gD(n):= X_2^\gD(n)/X_1^\gD(n)$, we have  that $Y^\gD\left(\lfloor \cdot/\gD\rfloor\right) \longrightarrow Y(\cdot)$ in law as $\gD \searrow 0$ just because of  \eqref{eq:approx1} and because the map $(x_1,x_2) \mapsto x_2/x_1$, from $(0, \infty)^2$ to 
 $(0, \infty)$, is continuous. Denote by $T_t^\gD$ and $T_t$ the corresponding Markov operator semigroups
\begin{equation}
T_t^{\gD} f(y) \, =\,   E^\gD_y\left[ f\left( Y^\gD\left(\lfloor t/\gD\rfloor\right) \right) \right]\;,\qquad T_t f(y) =  E_y\left[ f\left( Y(t) \right) \right]\, ,
\end{equation}
acting on bounded continuous $f: (0,\infty) \to \R$. Note that we have also introduced the notation $E^\gD$ and $E$ for the expectation with respect
to the two Markov processes we consider.  We claim that:
\medskip

\begin{enumerate}
\item \label{eq:1} For bounded continuous $f: (0,\infty) \to \R$ and $t \in [0, \infty)$, we have that
\begin{equation}
T_t^{\gD}f(y) \stackrel{\gD \searrow 0}\longrightarrow T_t f(y) \quad{\rm uniformly\ for\ } y {\rm \ in \ compact\ subsets\ of\ } (0,\infty)\;.
\end{equation}
\item \label{eq:2} For all positive $\gD$ there exists a unique law $\mu^\gD$ on $ (0,\infty)$ which is  invariant  for the Markov chain $Y^\gD$, which is ergodic. 
\item \label{eq:3} 
Choosing $\gD_0 \in (0, 1/\gep)$ we have 
\begin{equation}
\sup_{\gD \in (0,\gD_0]} 
\int_0^\infty y^{2} \mu^\gD(\dd y) < \infty\;, \qquad 
\sup_{\gD \in (0,\gD_0]} \int_0^\infty y^{-1}  \mu^\gD(\dd y) < \infty\;.
\end{equation}
\end{enumerate}
\medskip

Claim \eqref{eq:1} is a byproduct of the proof of  \eqref{eq:approx1}  \cite[Theorem~11.2.3]{StroockVaradhan}.
Claim \eqref{eq:2} comes from the general theory of products of random matrices. Let us prove  \eqref{eq:3}, and start by writing
\begin{equation} 
Y^\gD(n+1)\,=\,  Z^\gD(n+1) \, u( Y^\gD(n))\;,\qquad u(y)= \frac{y+\gep \gD}{1+\gep \gD y}\;.
\end{equation}
Observing that for $\gD \in (0, 1/\gep]$
\begin{equation} 
\frac{\dd^2 }{\dd z^2} \left(u( z^{1/2} )^2\right) = -  \gep \gD (1-\gep^2 \gD^2) \frac{3z + 4 \gep \gD z^{1/2}  + 1}{2 z^{3/2} (\gep \gD z^{1/2} +1)^4}
\leq 0\, ,
\end{equation}
we obtain by the Markov property and by Jensen's inequality that for a given initial condition $y>0$
\begin{equation} 
\begin{split}
E_{y} \left[ Y^\gD(n+1)^2\right]\, &=\, \bbE \left[\left(Z^\gD(n+1)\right)^2)\right] E_{y} \left[u( Y^\gD(n))^2\right] \\
& \le \,  q_{\gD,+}^2 
 u\left( E_{y}\left[ Y^\gD(n)^2\right]^{1/2}\right)^2 \,, 
\end{split}
\end{equation}
where 
\begin{equation}
q_{\gD,+}\, :=\,  \sqrt{\bbE \left[\left(Z^\gD\right)^2)\right]}\stackrel{\eqref{eq:condDH83-g1}}= 1+ \left(1-\frac \ga 2\right)\gs ^2 \gD + o\left(\gD^2\right)
\, .
\end{equation}
Therefore if we set $x_n:=E_{y} \left[ Y^\gD(n)^2\right]^{1/2}$ we have 
$x_{n+1}\le q_{\gD,+} u(x_n)$ which directly entails that $x_n< \infty$ for every $n$  and, since $u(\cdot)$ is bounded and concave increasing with $u(0)>0$,
 the application $q_{\gD,+} u(\cdot)$ has only one  positive fixed point that attracts every positive number. 
 The fixed point $x^+_{\ga, \gep}(\gD)$ is easily computed:
 \begin{multline}
 \label{eq:defxplus}
 x^+_{\ga, \gep}(\gD)\, =\, \frac 12 \left( \frac{q_{\gD,+}-1}{\gep \gD}+ \sqrt{
 \left( \frac{q_{\gD,+}-1}{\gep \gD}\right)^2
 + 4 q_{\gD,+}
 }\right) 
 \\
 \stackrel{\gD \searrow 0}\sim
  \frac 12 \left( 
  \frac{\left(1-\frac \ga 2\right)}\gep \gs ^2
  + \sqrt{
 \left( \frac{\left(1-\frac \ga 2\right)}\gep \gs ^2\right)^2
 + 4 
 }\right)
  \, .
 \end{multline}
 Therefore 
 $\limsup_n x_n  \le x^+_{\ga, \gep}(\gD)$ 
and $x^+_{\ga, \gep}(\gD)$ is bounded for $\gD\searrow 0$.
Since $\{Y_n^\gD\}_{n=0,1, \ldots}$ converges in law to the  random variable $Y^\gD_\infty$ that is distributed according to 
$\mu^\gD$, by standard measure theory argument we infer that
$\bbE[ (Y^\gD_\infty)^2]=
\int_0^\infty y^{2}  \mu^\gD(\dd y) \le (x^+_{\ga, \gep}(\gD))^2$
 which 
 proves the first claim in \eqref{eq:3}. 
 
 For the other claim in \eqref{eq:3} it is useful to note that $\widetilde{Y}^\gD(n)=
Y^\gD(n)^{-1}$ evolves according to the similar dynamics driven by $1/Z^\gD$,
\begin{equation}
\widetilde{Y}^\gD(n+1)\, =\, \left(Z^\gD(n+1)\right)^{-1} 
\frac{\widetilde{Y}^\gD(n)+ \gep \gD}{1+\gep \gD \widetilde{Y}^\gD(n)}\, .
\end{equation}
We can now proceed in a simpler way than above and exploit directly the concavity of $u(\cdot)$ to get to
\begin{equation}
\tilde x_{n+1} \, :=\, E_{y} \left[ \widetilde{Y}^\gD(n+1)\right]\,  \le \,  q_{\gD,-} 
 u\left( E_{y}\left[ \widetilde{Y}^\gD(n)\right]\right) \,, 
\end{equation}
and $\limsup_n \tilde x_n \le x^-_{\ga,\gep}(\gD)$, with $x^-_{\ga,\gep}(\gD)$ defined replacing $q_{\gD, +}$ with 
$q_{\gD, -}$ in the definition \eqref{eq:defxplus} of 
$x^+_{\ga,\gep}(\gD)$. It is therefore clear that
\eqref{eq:condDH83-g2} tells us that $x^-_{\ga,\gep}(\gD)$ remains bounded for $\gD\searrow 0$ and 
the second claim in \eqref{eq:3} is proven.
\smallskip

\begin{rem}
Of course if we make the stronger, but in practice almost equivalent, condition 
on the second moment of $1/Z^\gD$ in \eqref{eq:condDH83-g2}, the argument for 
the first claim in \eqref{eq:3} applies and directly  yields 
$\sup_{\gD \in (0,\gD_0]} \int_0^\infty y^{-2}  \mu^\gD(\dd y) < \infty$.
\end{rem}
\medskip 

With  \eqref{eq:1}--\eqref{eq:3} at hands, we complete the proof of  \eqref{eq:approx2}.
By \eqref{eq:AgD1}-\eqref{eq:AgD2} and iterating we obtain
\begin{equation}
\label{eq:vite}
\begin{split}
 \log X_1^\gD(n) \,&=\, \log X_1^\gD(n-1) + \log \left(1+ \gep \gD Y^\gD(n-1) \right)
  \\
  &=\,  \log X_1^\gD(0) + \sum_{i=1}^n \log \left(1+ \gep \gD Y^\gD(i-1) \right)  \, .
  \end{split}
\end{equation}
Following \cite[Th.~4.3 in Ch.~III]{cf:BL},  we express the Lyapunov exponent
\begin{multline}
\widehat  \cL_{Z^\gD} (\gep)\,=\, \lim_{n \to \infty} \frac 1n \log \|X^\gD(n)\|  \, =\, \lim_{n \to \infty} \frac 1n \log X_1^\gD(n)\\
 \stackrel{\eqref{eq:vite}}{=} \lim_{n \to \infty} \frac 1n  \sum_{i=1}^n \log \left(1+ \gep \gD Y^\gD(i-1) \right)  \,
=\, \int_0^\infty \log\left( 1+ \gep \gD y\right)  \mu^\gD (\dd y) \, . 
\end{multline}
By \eqref{eq:3}, the family $\{\mu^\gD\}_{ \gD \in (0,\gD_0]}$ of probability measures is tight on $(0,\infty)$.
By \eqref{eq:1} and \cite[Th.~9.10 in Ch.~4]{EthierKurtz}, every weak limit of $\{\mu^\gD\}_{ \gD \in (0,\gD_0]}$ is invariant for $Y$, whose unique invariant measure has the density $p_\gep(\cdot)$, implies that 
$\mu^\gD(\dd y)$ converges weakly to  $p_\gep(y) \dd y$ as $\gD \searrow 0$. Then,
\begin{equation} 
\frac{ \widehat  \cL_{Z^\gD} (\gep)}{ \gD } - \cL_{\gs, \ga} (\gep)
\,=\,   \int_0^\infty  \left(\frac {\log(1+\gep \gD y) }{ \gD }\! -\! \gep y\right) \mu^\gD ( \dd y) +  \int_0^\infty \gep  y  \left( \mu^\gD (\dd y)\!  -\!  p_\gep(y)\dd y\right) 
  \end{equation}
The last term vanishes as $\gD \searrow 0$ by weak convergence and  uniform integrability from claim \eqref{eq:3}. But also  the first term in the right-hand side vanishes for the same reasons because
\begin{multline}
\left\vert 
\frac {\log \left( 1+ \gep \gD y\right) }{ \gD } -\gep y \right\vert \, =\,
  \int_0^y \frac{\gD \gep^2 z}{1+\gD \gep z} \dd z 
\, \le \,
\sqrt{\int_0^y \gep  \dd z
\,
 \int_0^y {\gD \gep^2 z} \dd z}\, =\, \frac{\gep^{3/2}}{\sqrt{2}} \gD^{1/2} y^{3/2}\,,
\end{multline}
where we have used that, for $u\ge 0$,  $u/(1+u)$ is bounded above both by $1$ and by $u$.
This completes the proof of \eqref{eq:approx2} and, therefore, also the proof of Theorem~\ref{prop:scalingDH83g}.
\qed

\appendix

\section{The McCoy-Wu model}
\label{sec:MW}
In \cite{MW1}, McCoy and Wu examined a two-dimensional Ising model with bond disorder of a particular type (subsequently known as the McCoy-Wu model): the couplings between sites in neighboring columns have a constant strength $E_1$, while the couplings between neighboring sites in the same column take  a random value $E_2(n)$ which is fixed within each row but varies independently -- keeping the same distribution -- between different rows (Figure~\ref{fig:1}).  They showed that in the thermodynamic limit the free energy per site of this model is given (up to the subtraction of an analytic function of $\beta$) by  
\begin{equation}
\label{eq:MWfe}
\tf_{\scriptscriptstyle{\textrm{MW}}}(\gb)\, :=\, 
\frac 1{4\pi}\int_{-\pi}^\pi \cL^{\scriptscriptstyle{\textrm{MW}}}_{\gb}(\theta) \dd \theta\,,
\end{equation}
where 
$\cL^{\scriptscriptstyle{\textrm{MW}}}_{\gb}(\theta)$ is the Lyapunov exponent of the random matrix
\begin{equation}
\label{eq:MWmatrix}
M_\beta(\theta)
:=
\begin{pmatrix}
    1 & \frac{a}{a^2+b^2}  \\
    \frac{a}{a^2+b^2} \gl & 
    \frac{\gl}{a^2+b^2}
  \end{pmatrix}\, ,
\end{equation}
with
\begin{equation}
a(\theta)\, =\, -2z_1 \frac{\sin(\theta)}{\left \vert 1+z_1 \exp(i\theta) \right \vert ^2}\ \ 
\text{ and } \ \
b(\theta)\, =\, \frac{1-z_1^2}{\left \vert 1+z_1 \exp(i\theta) \right \vert ^2}\, ,
\end{equation} 
where 
\begin{equation}
z_1\, =\, \tanh \left( \gb E_1\right)\, , \ \
 z_2(n)\, =\, \tanh \left( \gb E_2(n)\right)
 \ \ 
\text{ and } \ \ 
\gl=\gl(n)= z_2^2(n)\, .
\end{equation}
In \cite{ShankarMurthy} a different version of the model has been considered: vertical bounds are random
in the horizontal direction and randomness is repeated in each line. This model, that allows frustration, is richer, but 
the features that are novel with respect to the McCoy-Wu model cannot be appreciated in the weak disorder limit: our analysis applies to \cite{ShankarMurthy} as well, but we will not develop this issue here.

\begin{figure}[htbp]
\centering
\includegraphics[width=6 cm]{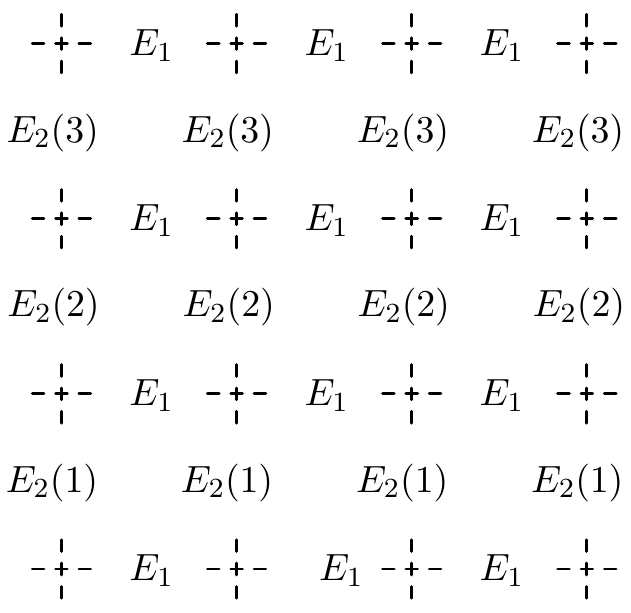}
\vskip-.2cm
\caption{\label{fig:1}  The McCoy-Wu disordered version of the two dimensional Ising model: the  disordered interactions are in the vertical direction and they are distributed in an IID fashion within one column. This disorder is just copied to all the other columns and the horizontal interactions are non random. The disorder enters the free energy formula via
independent copies of the random variable $\gl= \tanh^2 \left( \gb E_2\right)$.}
\end{figure}

To avoid trivialities we assume that $E_1\neq 0$ as well as that $E_2$ is a non degenerate random variable:
it is immediately clear that the sign of $E_2$ does not matter and just a little thought reveals that the sign of $E_1$
is irrelevant too. Therefore we assume that  $E_1\in (0, \infty)$ and that $E_2$ is a random variable taking values in $(0, \infty)$.
It is helpful (mostly to simplify the presentation) to assume that $E_2$ takes values in $[E_2^{-}, E_2^+]$, with $0<E_2^- <E_2^+< \infty$. 


Moreover one directly sees that $a(\cdot)$ is odd and $b(\cdot)$ is even, which yields
that $\cL^{\scriptscriptstyle{\textrm{MW}}}_{\gb}(\cdot)$ is even: in fact  $D_\pm M_\gb(\theta) D_\pm=
M_\gb(-\theta)$, with $D_\pm$ the diagonal matrix with $(+1,-1)$ on the diagonal. Therefore:
\begin{equation}
\tf_{\scriptscriptstyle{\textrm{MW}}}(\gb)\, :=\, 
\frac 1{2\pi}\int_{0}^\pi \cL^{\scriptscriptstyle{\textrm{MW}}}_{\gb}(\theta) \dd \theta\,.
\end{equation}

McCoy and Wu claim that for every $\upsilon \in (0, \pi)$ -- our focus is on $\upsilon$ small -- the function
\begin{equation}
\gb \mapsto \frac 1{2\pi}\int_{\upsilon}^\pi \cL^{\scriptscriptstyle{\textrm{MW}}}_{\gb}(\theta) \dd \theta\,.
\end{equation}
is real analytic on $(0, \infty)$. This can be proven by applying the main result in \cite{Ruelle} (see also \cite{cf:dubois}). We sketch the argument here by considering separately the case
$\theta$ bounded away from $0$ and $\pi$ and the case of $\theta$ near $\pi$: with $\gd>0$ small

\smallskip
\begin{itemize}
\item For $\theta \in [\gd, \pi -\gd]$ the matrix $M_\beta(\theta)$ (with positive entries) maps the closure of  the cone $Q$ -- here $Q$ is first quadrant without the axes, that is the set of vectors with positive entries -- to   $Q\cup \{0\}$. More precisely, 
by the  hypothesis we have made on the suport of $Z$,  for every $\gd\in (0, \pi/2)$ and every $\gvr\in(0,1)$ the closure of $Q$ is mapped  into a cone whose closure  is a subset of $Q\cup \{0\}$ and this subset is the same 
for every choice of $\theta
\in [\gd, \pi -\gd]$ and every $\gb\in [\gvr,1/\gvr]$. This 
uniform cone property implies the real analyticity of 
$\gb \mapsto \cL^{\scriptscriptstyle{\textrm{MW}}}_{\gb}(\theta)$ with a convergence radius that is bounded away from zero uniformly in $\theta
\in [\gd, \pi -\gd]$ and $\gb\in [\gvr,1/\gvr]$.
\item  For $\theta \in [\pi -\gd,\pi]$ we argue by observing first that
\begin{equation}
	a(\pi) = 0, 
	\quad
	b(\pi)
	=
	\frac{1-z_1^2}{(1-z_1)^2}
	=
	\frac{1+ \tanh \beta E_1}{1 - \tanh \beta E_1}
	=
	e^{2 \beta E_1},
\end{equation}
so
\begin{equation}
	M_\beta(\pi)
	\,=\,
	\begin{pmatrix}
		1 & 0 \\
		0 & e^{-4 \beta E_1} \tanh^2 \beta E_2
	\end{pmatrix}.
\end{equation}
Since $e^{-4 \beta E_1} \tanh^2 \beta E_2<1$ the action of $M_\beta(\pi)$ contracts uniformly any cone 
of the form $\{(x,y): \, y \ge \vert x\vert \}$, in the sense there exists $\gvr >0$ such that  $M_\beta(\pi)$
sends $\{(x,y): \, y \ge \vert x\vert \}$ into $\{(x,y): \, y \ge (1+ \gvr) \vert x\vert \}$, uniformly in $\gb>0$ and $E_2$.
Elementary arguments show that this result is only slightly perturbed if we consider $\theta \in [\pi -\gd,\pi]$ with $\gd$ sufficiently small. 
This uniform cone property implies the real analyticity of 
$\gb \mapsto \cL^{\scriptscriptstyle{\textrm{MW}}}_{\gb}(\theta)$ with a convergence radius that is bounded away from zero uniformly in $\theta \in [\pi -\gd,\pi]$ and $\gb>0$.
\end{itemize}
\medskip

Therefore the true issue is the regularity (or lack of it) of 
\begin{equation}
\label{eq:etamap}
\gb \mapsto \frac 1{2\pi}\int_{0}^\upsilon \cL^{\scriptscriptstyle{\textrm{MW}}}_{\gb}(\theta) \dd \theta\,,
\end{equation}
for a $\upsilon>0$ that can be chosen as small as one wishes. 
At this point McCoy and Wu claim that 
the only non analytic point of the map in \eqref{eq:etamap} can be at $\gb_c$ defined by
\begin{equation}
\label{eq:betac}
	2 \beta_c E_1 +  \bbE\left[ \log \tanh \beta_c E_2 \right] \, =\,  0\, .
\end{equation}
To see that this is the only possible candidate,
McCoy and Wu point out that
\begin{equation}
	a(0) = 0, 
	\quad
	b(0)  
	=
	\frac{1-z_1^2}{(1+z_1)^2}
	=
	\frac{1- \tanh \beta E_1}{1 + \tanh \beta E_2}
	=
	e^{-2 \beta E_1},
\end{equation}
so 
\begin{equation}
	M_\beta(0)
	=
	\begin{pmatrix}
		1 & 0 \\
		0 & e^{4 \beta E_1} \tanh^2 \beta E_2
	\end{pmatrix},
\end{equation}
and so
\begin{equation}
  \cL^{MW}_\beta(0)
	=
	\max\left( 0, 4 \beta E_1 + 2 \bbE\left[ \log \tanh \beta E_2 \right] \right)
\end{equation}
for $\beta$ real. This admits an analytic extension in a neighborhood of any positive $\beta$ except for $\beta_c$.

This is of course far from being close to a proof, since one has to control the integral over $\theta \in (0, \upsilon)$ and
not the value in zero. But McCoy and Wu  perform also a more subtle analysis that can be understood 
precisely via the diffusion limit of matrix products that is at the center of our analysis. To explain this  let us make a further manipulation 
to match more sharply our framework.

\medskip

In fact, as it stands, 
$M_\beta(\theta)$, cf. \eqref{eq:MWmatrix}, is not of the form 
\eqref{eq:matrix}. But by 
noting that 
\begin{equation}
	\frac{1}{a^2(\theta)+b^2(\theta)} 
	=
	\frac{(1+z_1)^4}{(1-z_1^2)^2} + O(\theta^2)
	=
	\left( \frac{1+z_1}{1-z_1} \right)^2+ O(\theta^2)
	=
	e^{4 \beta E_1}+ O(\theta^2)\, , 
\end{equation}
and
\begin{equation}
	\frac{a}{a^2(\theta)+b^2(\theta)} 
	= 
	-2 \left( \frac{1+z_1}{1-z_1} \right)^2 \theta + O(\theta^2) = -2
	e^{4 \beta E_1}+ O(\theta^2)\, ,
\end{equation}
if we let
\begin{equation}
\label{eq:varepsgl}
	\tilde\varepsilon
	:= 
	\frac{2 z_1}{(1-z_1)^2} \theta
	\, ,
\end{equation}
we see  that to leading order as $\theta \searrow 0$
\begin{equation}
\label{eq:MWm-1}
	\begin{pmatrix}
		1 & -\tilde\varepsilon \\
		-\tilde\varepsilon \gl & e^{4 \beta E_1} \gl
	\end{pmatrix} 
\end{equation}
is  $M_\beta(\theta)$. The matrix in \eqref{eq:MWm-1}
is of the form \eqref{eq:matrix} up to a conjugation and a change of variables: 
in fact 
\begin{equation}
\label{eq:MWm-2}
	\begin{pmatrix}
1& \gep\\ 
\gep Z  & Z
\end{pmatrix}
\,:=\, 
\begin{pmatrix}
1& \tilde \gep e^{-2 \beta E_1}\\ 
\tilde\gep e^{-2 \beta E_1} \gl  & e^{4 \beta E_1} \gl
\end{pmatrix}
	\, =\, 
	\begin{pmatrix}
		1 & 0 \\
		0 &  - e^{2 \beta E_1}
	\end{pmatrix}
	\begin{pmatrix}
		1 & -\tilde \varepsilon \\
		- \tilde \varepsilon \gl & e^{4 \beta E_1} \gl
	\end{pmatrix}
	\begin{pmatrix}
		1 & 0 \\
		0 & - e^{-2 \beta E_1}
	\end{pmatrix}\, ,
\end{equation}
and we observe -- recall \eqref{eq:varepsgl} -- that
 $\gep = c_\gb \theta$, with $c_\gb= 2 \sinh(2\gb E_1)$.
\medskip

\begin{rem}
\label{rem:gbtoga}
It is important to remark at this stage that the inverse temperature $\gb$ and our fundamental parameter
$\ga$ -- we recall that $\ga$ is the unique non zero real solution to
$\bbE Z^\ga=1$ ($Z=e^{4\gb E_1} \tanh^2(\gb E_2)$ depends on $\gb$!) when such a solution exists and otherwise $\ga=0$ -- should be seen as an analytic change of variable: this is treated in detail in Lemma~\ref{th:gagb}. In particular 
$\ga(\gb_c)=0$ and therefore
$\ga(\gb)= (\gb-\gb_c) \ga ' (\gb_c)+ O((\gb-\gb_c)^2)$, but the \emph{constant} $\ga ' (\gb_c)$ depends of the law of $Z$ (with $\gb=\gb_c$) and this expansion
should be done more carefully when the disorder is weak because, as we will see, $\ga ' (\gb_c)$ becomes large
in this limit: this is treated in \eqref{eq:gDNMW}-\eqref{eq:MWbetas}.
\end{rem}
\medskip

What McCoy and Wu do at this point is 
\smallskip
\begin{itemize}
\item
making a specific choice of 
$Z=Z^\gD=Z^\gD_\gb$ that satisfies the hypotheses of Theorem~\ref{prop:scalingDH83g} (say, with $\gs=1$ for simplicity); this  actually implements two choices:
\begin{enumerate}
 \item the first is evident and it is the fact that disorder can be made weak by making $\gD$ small;
 \item the second is that $\gb-\gb_c$ is chosen small and, precisely, of the order of $\gD$. As we will
 explain, if we set $y=(\gb-\gb_c)/ \gD$ and we keep $y\in \bbR$ fixed, then $\ga(\gb) \sim -C_{\gb_c} y $, and the constant $C_{\gb_c}>0$
will be given explicit in the specific case that we are going to develop, see \eqref{eq:MWbetas}.
\end{enumerate}
\item they choose also $\upsilon \propto \gD$: let us fix in an arbitrary fashion $\upsilon = \gD$.
\end{itemize}
\smallskip

In physical terms these choices correspond to focusing on the critical window in the limit of weak disorder. Cutting
the integral at $\theta= \gD$ is harmless (as we have discussed before), but of course only as far as $\gD$ is kept fixed.

McCoy and Wu  are in the end just dealing (recall \eqref{eq:MWm-2}) with the Lyapunov exponent $\widehat \cL_{\gD, \gb_c + y\gD} (c_{\gb_c}x \gD)$
of the matrix (we perform the change of variable 
$\theta=x \gD$)
\begin{equation}
\begin{pmatrix}
		1 & c_{\gb_c} x \gD \\
		 c_{\gb_c} x \gD Z_{\gb_c + y\gD}^\gD & Z_{\gb_c + y\gD}^\gD
	\end{pmatrix}\, .
\end{equation} 
But Theorem~\ref{prop:scalingDH83g} (see also Theorem~\ref{prop:scalingDH83}) tells
us that $\widehat \cL_{\gD, \gb_c + y\gD} (c_{\gb_c }x)$ is asymptotically equivalent for $\gD$ small to $ \gD \cL_{1, C_{\gb_c}\ga}( c_{\gb_c}x)$ so that
\begin{equation}
\label{eq:3equ}
\int_{0}^\gD \cL^{\scriptscriptstyle{\textrm{MW}}}_{\gb_c + y \gD}(\theta) \dd \theta\, \sim\, 
\gD \int_{0}^1 \widehat \cL_{\gD, \gb_c + y\gD} (c_{\gb_c }x)\dd x\, \sim\, 
\gD^2 \int_{0}^1 \cL_{1,C_{\gb_c}\ga}( c_{\gb_c}x)\dd x
\end{equation}
and we remind the reader that $\cL_{1,C_{\gb_c}\ga}( c_{\gb_c}x)$ has the explict expression \eqref{eq:formula}.
Therefore, up to two inessential  constants we arrived at \eqref{eq:toyMW}. We did not fully justify the 
equivalences in \eqref{eq:3equ}, but this is not really the main problem: the main unresolved mathematical issue is that 
what we are after is proving that, for a fixed (possibly extremely small) value of $\gD$,  the leftmost term in 
\eqref{eq:3equ} is a $C^\infty$ function of $y$ at $0$ and that the same expression is not analytic at zero. McCoy and Wu instead argue (and we prove in Theorem~\ref{th:MW}) that 
$\ga \mapsto \int_{0}^1 \cL_{1,C_{\gb_c}\ga}( c_{\gb_c}x)\dd x$ has these properties: but this second statement does not imply the first. 
\medskip

We now complement our discussion with the analysis of the specific distribution chosen for the disorder law in \cite{MW1,MWbook}. We also 
discuss more in detail the change of variable $\ga(\gb)$. 

\subsubsection{Analysis of the distribution chosen by McCoy and Wu \cite{MW1,MWbook}}
 McCoy and Wu consider the disordered variable $\gl=\tanh^2( \gb E_2)$ that depends on a parameter that they call $\NMW$ and it is large: in fact
 \begin{equation}
 \label{eq:gDNMW}
 \gD\, =\, \NMW^{-2}\, .
 \end{equation}
 The density  of $\gl$ is 
  supported on $(0, \gl_0)$   and  equal to 
 $\NMW \gl_0^{-\NMW}y^{\NMW-1}$ for $y\in (0, \gl_0)$. Necessarily  $\gl_0=\gl_0(\gb)= \tanh^2(\gb E_2^*)$, with 
 $E_2^*$ the maximum value that the random variable $E_2$ can reach. The density of $ Z=  Z_\NMW$ 
 (recall that $Z$ is defined in \eqref{eq:MWm-2})
 is therefore
 \begin{equation}
 \label{eq:densitylambda1}
 \NMW \gl_1^{-\NMW}y^{\NMW-1} \ind_{(0, \gl_1)}(y)\ \ \ \ \text{ with } \ \gl_1(\gb)\, =\, e^{4\gb E_1}\gl_0(\gb)\, .
 \end{equation}
 Note that for every $\nu \in (-\NMW, \infty)$ 
 \begin{equation}
 \label{eq:betaA}
 \bbE \left[ \left( Z_\NMW \right)^\nu\right]\, =\, \frac{\gl_1(\gb)^\nu}{1+ \frac \nu \NMW}\, ,
 \end{equation} 
 and we want to solve for $\ga=\ga(\gb)\neq 0$ the equation
 \begin{equation}
  \label{eq:alphaA}
{\bbE \left[ \left(  Z_\NMW \right)^\ga\right]\, -1} \, =\, \frac{\gl_1(\gb)^\ga - 1- \frac \ga \NMW}{1+ \frac \ga \NMW}\, =\, 0\, .
 \end{equation} 
  On one hand  we compute
 \begin{multline}
 \log (\gl_1(\gb))\, =\\
  \log \tanh ^2 (\gb_c E_2^*) - \bbE\left[\log \tanh ^2 (\gb_c E_2)\right] + 4 (\gb -\gb_c)+ 
 \log \tanh ^2 (\gb E_2^*)- \log \tanh ^2 (\gb_c E_2^*)\, ,
 \end{multline} 
 and a straightforward computation yields 
 \begin{equation}
 \bbE\left[\log \tanh ^2 (\gb E_2)\right]\, =\,  \log \tanh ^2 (\gb E_2^*) - \frac 1  \NMW \, ,
 \end{equation}
 so  for $\gb$ close $\gb_c$ we have 
 \begin{equation}
  \label{eq:gl1logE}
 \log \gl_1(\gb)\, =\, \frac 1  \NMW  + (\gb -\gb_c) \left( 4E_1 + \frac 1 {\sinh (2 \gb_c E_2^*)}\right) + O\left( ( \gb- \gb_c)^2\right)\, .
 \end{equation} 
 On the other hand from \eqref{eq:alphaA}
we see that if $\ga$ is fixed (so we look at $\gb$ as a function of $\ga$) we have 
 \begin{equation}
 \label{eq:gl1log}
 \log \gl_1(\gb) \,=\, \frac{\log\left(1+ \frac \ga \NMW\right)} \ga \stackrel{\NMW \to \infty}=\,   \frac 1\NMW- \frac \ga {2\NMW^2}+ O\left( \frac 1{\NMW^3}\right)\, .
\end{equation}
By comparing \eqref{eq:gl1logE} and \eqref{eq:gl1log} we see that if $(\gb-\gb_c) \NMW^2=O(1)$ then 
\begin{equation}
\label{eq:MWbetas}
\ga (\gb) \, =\, -(\gb -\gb_c) {\NMW^2} \left( 8E_1 + \frac 2 {\sinh (2 \gb_c E_2^*)}\right) + O\left( \NMW^{-1}\right)\, .
\end{equation}

\subsubsection{On the relation between $\gb$ and $\ga$}
Here are the details of the important map that relates $\gb$ and $\ga$: 
\medskip

\begin{lemma}
\label{th:gagb}
Assume that the support of the random variable $E_2$ is bounded away from zero, so $Z= \exp(4 \gb E_2) 
\tanh^2(\gb E_2)$ is supported on a compact subinterval of $(0, \infty)$. Assume also that $E_2$ is not constant. 
Then the  equation
\begin{equation}
\frac{\bbE\left[ Z^\ga \right]-1} \ga \, =\, 0\, ,
\end{equation}
has a unique real solution $\ga$ for every $\gb>0$. This defines a map
$\gb \mapsto \ga(\gb)$ from $(0, \infty)$ to $\bbR$. This map is decreasing, hence it is a bijection,  and it is real analytic.
\end{lemma}

\medskip

\noindent
\emph{Proof.} 
Let $f: \bbR \times (0,\infty) \to \bbR$ be the function defined by
$f(\ga, \gb):=\frac{\bbE\left[ Z^\ga \right]-1} \ga$, for $\ga\in \bbR \setminus \{0\}$
for $\alpha \neq 0$, and $f(0, \gb):= \bbE [\log Z]$.
It is straightforward to see, using the support properties of $E_2$, that $f$ is real analytic on its entire domain. 
Then we observe that, for fixed $\ga$, $Z^\ga$ is an increasing function of $\gb$ and, by the support properties,
this implies that $\partial_\gb f(\ga, \gb)>0$ for every $\gb>0$ and $\ga\in \bbR$.
On the other hand  
 if we set $g_\gb (\ga)= \bbE\left[ Z^\ga \right]-1$
we have that $\partial_\ga f(\ga, \gb)=( \ga g'_\gb(\ga)- g_\gb (\ga)) / \ga^2$. 
But $g_\gb (\cdot)$ is  (strictly) convex and $g_\gb (0)=0$: so $ \ga g'_\gb(\ga)- g_\gb (\ga)>0$ for $\ga\neq 0$
and therefore $\partial_\ga f(\ga, \gb)>0$ for $\ga\neq 0$. For $\ga=0$ it suffices to perform a Taylor expansion
of $g_\gb(\ga)$ at $\ga=0$ to see that $\partial_\ga f(\ga, \gb)\vert_{\ga=0}= g^{\prime \prime}_\gb (0)/2>0$.
The proof is completed by  applying the  Implicit Function Theorem for real analytic functions \cite{cf:primer}.
\qed

\section*{Acknowledgments} 
We are very grateful to  Bernard Derrida for very insightful discussions, and to an anonymous referee for pointing out significant errors in the statement and proof of Proposition~\ref{prop:asympt} in an earlier manuscript. 
Part of this work was developed while G.~G.\ and R.~L.~G.\ were visiting IHP (Paris) during the spring-summer 2017 trimester.  G.~G.\ acknowledges the support of grant ANR-15-CE40-0020.  The work of
R.~L.~G.\
was funded by the European Research Council under the European Union's Horizon 2020 Programme, ERC Consolidator Grant UniCoSM (grant agreement no 724939).


\begin{thebibliography}{99}


\bibitem{cf:AKQ}
Alberts, T.,  Khanin, K. and Quastel, J.; \emph{The intermediate disorder regime for directed polymers in dimension $1 + 1$}. Ann. Probab. {\bf 42} (2014), 1212-1256.

\bibitem{cf:Baxter} 
 Baxter, R. J.; \emph{Exactly solved models in statistical mechanics.} Academic Press,  1982.

\bibitem{cf:BL}
   Bougerol, P.  and Lacroix, J.;
    \emph{Products of random matrices with applications to {S}chr\"odinger
              operators}.
   Progress in Probability and Statistics 
    {\bf 8},
 {Birkh\"auser Boston, Inc., Boston, MA},
      {1985}.
      
\bibitem{cf:bouchaud} 
 Bouchaud, J.-P., Comtet, A., Georges, A.  and Le Doussal, P.;
\emph{Classical diffusion of a particle in a one-dimensional random force field}, 
Ann. Physics {\bf 201} (1990),   285-341.       
     
     
\bibitem{cf:CSZ}      
 Caravenna, F.,   Sun, R. and  Zygouras, N.;
\emph{Polynomial chaos and scaling limits of disordered systems}.
J. Eur. Math. Soc. {\bf 19} (2017), 1-65.

\bibitem{CattiauxChafaiGuillin}
Cattiaux, P.,  Chafa\"i, D. and Guillin, A.; 
\emph{Central limit theorems for additive functionals of ergodic Markov diffusions processes.}
ALEA Lat. Am. J. Probab. Math. Stat. \textbf{ 9} (2012), 337-382.

      
  \bibitem{cf:contemp}
   Cohen, J. E.,   Kesten, H.  and  Newman,  C. M. (Editors); 
   \emph{Random matrices and their applications.} 
Proceedings of the AMS-IMS-SIAM joint summer research conference held at Bowdoin College, Brunswick, Maine, June 17-23, 1984. Contemporary Mathematics, {\bf 50}, American Mathematical Society,  1986.

\bibitem{cf:desbois}
 Comtet, A., Desbois, J. and  Monthus, C.;  \emph{Localization properties in one-dimensional disordered supersymmetric quantum mechanics}, Ann. Physics {\bf 239} (1995),  312-350. 


\bibitem{cf:CLTT}
Comtet, A., Luck, J. M., Texier, C. and
 Tourigny, Y.;
 \emph{The Lyapunov Exponent of Products of Random $2\times 2$ Matrices Close to the Identity.}
  J. Statist. Phys. {\bf 150} (2013), 13-65.
  
 \bibitem{cf:CTT} Comtet, A., Texier, C. and  Tourigny, Y.;
\emph{Lyapunov exponents, one-dimensional Anderson localization and products of random matrices}, 
J. Phys. A {\bf 46} (2013), 254003, 20 pp.  


\bibitem{CPV}
Crisanti, A., Paladin, G. and Vulpiani, A.;
\newblock \emph{{Products of random matrices in statistical physics}}, volume
  104 of \emph{{Springer Series in Solid-State Sciences}}.
\newblock Berlin: Springer-Verlag, 1993.


 
 
\bibitem{DH}
Derrida, B. and Hilhorst, H. J.;
\emph{Singular behaviour of certain infinite products of random {$2\times 2$} matrices}.
J. Phys. A \textbf{16} (1983), 2641-2654. 


  
  \bibitem{cf:dubois}
   Dubois, L.;
    \emph{Real cone contractions and analyticity properties of the characteristic exponents},  Nonlinearity {\bf 21} (2008),  2519-2536.

\bibitem{Dunster}
Dunster, T.M.;
\emph{Bessel functions of purely imaginary order, with an application to second-order linear differential equations having a large parameter}.
SIAM J. Math. Anal. \textbf{21} (1990), 995-1018.

\bibitem{EthierKurtz}
Ethier, S. and Kurtz, T.;
\emph{Markov processes. Characterization and convergence.}
Wiley Series in Probability and Mathematical Statistics: Probability and
   Mathematical Statistics.
1986.

\bibitem{cf:dfisher} Fisher, D. S.;
 \emph{Critical behavior of random transverse-field Ising spin chains}, Phys. Rev. B {\bf 51} (1995), 6411-6461.
 

\bibitem{cf:FW}
 Freidlin, M. I. and Wentzell, A. D.;  \emph{Random perturbations of dynamical systems}.  Third edition. Grundlehren der Mathematischen Wissenschaften, \textbf{260},  Springer, 2012.

\bibitem{Friedlander}
  Friedlander, F. G.;
  \newblock{Diffraction of Pulses by a Circular Cylinder},
  \emph{Comm. Pure and Appl. Math.} \textbf{7} (1954), 705-732.
  
 \bibitem{cf:FrLl}
 Frisch, H.~L. and Lloyd, S.~P.;
 \emph{Electron levels in a one-dimensional random lattice}, Phys. Rev. {\bf 120} (1960), 1175-1189.
 

\bibitem{cf:GGG}
Genovese, G.,  Giacomin, G. and  Greenblatt, R. L.;
\emph{Singular behavior of the leading Lyapunov exponent of a product of random 2$\times$2 matrices},
{Commun. Math. Phys.} \textbf{351} (2017), 923-958. 

\bibitem{cf:GB} Giacomin, G.; \emph{Random polymer models}, Imperial College Press, World Scientific, 2007.
 
 \bibitem{cf:G} Giacomin, G.;
  \emph{Disorder and critical phenomena through basic probability models}, \'Ecole d'\'et\'e de probablit\'es de Saint-Flour XL-2010, Lecture Notes in Mathematics {\bf 2025}, Springer, 2011.

\bibitem{cf:grabsch} 
Grabsch, A., Texier, C. and   Tourigny, Y.; \emph{One-dimensional disordered quantum mechanics and Sinai diffusion with random absorbers},  J. Stat. Phys. {\bf 155} (2014),  237-276.

\bibitem{cf:benjamin} B.~Havret, \emph{Regular expansion for the characteristic exponent of a product of $2 \times 2$ random matrices}, arXiv:1804.06166

\bibitem{cf:KS} 
Karatzas, I. and  Shreve, S. E.;
\emph{Brownian Motion and Stochastic Calculus}.
Graduate Texts in Mathematics
{\bf 113}, Springer, 1998.

\bibitem{cf:primer}
 Krantz, S. G. and  Parks, H. R.;
 \emph{A primer of real analytic functions},   Birkh\"user Verlag,  1992.

\bibitem{Kaup}
Kaup, L. and Kaup, B.;
\emph{Holomorphic Functions of Several Variables}.
1983.


\bibitem{cf:Luck}
 Luck, J. M.;   \emph{Critical behavior of the aperiodic quantum Ising chain in a transverse magnetic field}, J. Statist. Phys. {\bf 72} (1993),  417-458.

\bibitem{MaruyamaTanaka}
Maruyama, G. and Tanaka, H.;
\emph{ Ergodic property of $N$-dimensional recurrent Markov processes}. Mem. Fac. Sci. Kyushu Univ. Ser. A  \textbf{13} (1959), 157-172.

\bibitem{MWbook}
McCoy, B. M. and Wu, T. T.;
\emph{The Two-Dimensional Ising Model}. Harvard University Press, 1973.

\bibitem{MW1}
McCoy, B. M. and Wu, T. T.;
\emph{Theory of a two-dimensional {I}sing model with random impurities.
  {I}. Thermodynamics}.
Phys. Rev. \textbf{176} (1968), 631-643.

\bibitem{DLMF}
Olver, F. W. J., Olde Daalhuis, A. B.,  Lozier, D. W.,  Schneider, B. I.,  Boisvert, R. F.,   Clark, C. W.,   Miller, B. R. and  Saunders, B. V. (Editors);
  \emph{NIST Digital Library of Mathematical Functions}.
  http://dlmf.nist.gov/, Release 1.0.16 of 2017-09-18. 


\bibitem{NL}
 Nieuwenhuizen, Th. M. and Luck, J. M.;  \emph{Exactly soluble random field Ising models in one dimension}. 
{J. Phys. A} \textbf{19} (1986), 1207-1227.

\bibitem{PardouxVeretennikov}
 Pardoux, E. and  Veretennikov, A. Yu.;
\emph{On the Poisson equation and diffusion approximation. I}. \emph{Ann. Probab.} \textbf{29} (2001), 1061-1085.

\bibitem{cf:ramola}
 Ramola, K. and Texier, C.; \emph{Fluctuations of random matrix products and 1D Dirac equation with random mass}, J. Stat. Phys. {\bf 157} (2014), 497-514.

\bibitem{RogersWilliams}
Rogers, L. and Williams, D.;
\emph{Diffusions, Markov processes, and martingales.} Vol. 2. 
It\^o calculus. Cambridge University Press, Cambridge, 2000.

\bibitem{Ruelle}
Ruelle, D.;
\emph{Analytic properties of the characteristic exponents of random matrix products}.
\newblock \emph{Adv. Math} \textbf{32} (1979), 68-80.

\bibitem{cf:Ruelle-book}
Ruelle, D.; \emph{Thermodynamic formalism. The mathematical structures of equilibrium statistical mechanics.} Second edition. Cambridge Mathematical Library. Cambridge University Press,  2004.

\bibitem{cf:Sadel} 
Sadel, C. and  Schulz-Baldes, H.; \emph{Random Lie group actions on compact manifolds: a perturbative analysis.}
Ann. Probab. {\bf 36} (2010), 2224-2257.

\bibitem{cf:schomerus}
 Schomerus, H.,  and  Titov, M.;
\emph{Statistics of finite-time Lyapunov exponents in a random time-dependent potential}, 
Phts. Rev. E {\bf 66} (2002), 066207 (11 pages).

\bibitem{ShankarMurthy}
Shankar, R. and Murthy, G.;
\emph{Nearest-neighbor frustrated random-bond model in \textit{d=2}: Some
  exact results}.
Phys. Rev. B \textbf{36} (1987), 536-545.


\bibitem{StroockVaradhan}
Stroock, D. and Varadhan, S.; 
\emph{Multidimensional diffusion processes}.
Classics in Mathematics. Springer-Verlag, Berlin, 2006.

\bibitem{cf:marina} Taleb, M.; 
\emph{Large deviations for a Brownian motion in a drifted Brownian potential}, 
Ann. Probab. {\bf 29} (2001), 1173-1204. 

\bibitem{cf:Zanon}
Zanon, N. and  Derrida, B.; 
\emph{Weak disorder expansion of Lyapunov exponents in a degenerate case.} J. Stat.
Phys. {\bf 50} (1988), 509-528.

\end{thebibliography}
\end{document}